\documentclass[aps,showpacs,reprint,floatfix,superscriptaddress]{revtex4-2}
\usepackage{graphicx}% Include figure files
\usepackage{bm}% bold math
\usepackage[colorlinks=true,urlcolor=blue,citecolor=blue,linkcolor=blue]{hyperref}% add hypertext capabilities
\usepackage{suffix}

\usepackage{xr}
\usepackage{xfrac}
\usepackage{multirow}
\usepackage{amsmath}
\usepackage{amssymb}
\usepackage{braket}
\usepackage{physics}
\usepackage{amsthm}
\usepackage{natbib}
\usepackage{mathtools}
\usepackage[utf8]{inputenc}
\usepackage[english]{babel}
\usepackage{scalerel}[2014/03/10]
\usepackage{stackengine}
\usepackage{quantikz}
\usepackage{algorithm}
\usepackage[noend]{algpseudocode} 
\usepackage{subfigure}
\usepackage{enumitem}

\newcommand{\algmargin}{\the\ALG@thistlm}
\makeatother
\newlength{\whilewidth}
\algdef{SE}[parWHILE]{parWhile}{EndparWhile}[1]
  {\parbox[t]{\dimexpr\linewidth-\algmargin}{%
     \hangindent\whilewidth\strut\algorithmicwhile\ #1\ \algorithmicdo\strut}}{\algorithmicend\ \algorithmicwhile}%
\algnewcommand{\parState}[1]{\State%
  \parbox[t]{\dimexpr\linewidth-\algmargin}{\strut #1\strut}}

\newtheorem{theorem}{Theorem}[]
\newtheorem*{remark}{Remark}
\newtheorem{corollary}{Corollary}[]
\newtheorem{lemma}[]{Lemma}

\newtheorem{definition}{Definition}

\theoremstyle{definition}

\makeatletter
\@addtoreset{proofpart}{theorem}
\@addtoreset{proofcase}{theorem}
\makeatother
\renewcommand{\exp}[1]{\text{exp}\left( #1 \right)}

\newcommand{\sm}[0]{Supplemental Information}

\begin{document}
\preprint{APS/123-QED}

\title{Diagnosing Quantum Circuits: Noise Robustness, Trainability, and Expressibility}

\author{Yuguo Shao}
\thanks{These authors contributed equally to this work.}
\affiliation{Yau Mathematical Sciences Center, Tsinghua University, Beijing 100084, China}
\affiliation{Yanqi Lake Beijing Institute of Mathematical Sciences and Applications, Beijing 100407, China}

\author{Zhenyu Chen}
\thanks{These authors contributed equally to this work.}
\affiliation{Institute for Interdisciplinary Information Sciences, Tsinghua University, Beijing 100084, China}
%\affiliation{Department of Computer Science and Technology, Tsinghua University, Beijing 100084, China}

\author{Zhaohui Wei}
\email{weizhaohui@gmail.com}
\affiliation{Yau Mathematical Sciences Center, Tsinghua University, Beijing 100084, China}
\affiliation{Yanqi Lake Beijing Institute of Mathematical Sciences and Applications, Beijing 100407, China}

\author{Zhengwei Liu}
\email{liuzhengwei@mail.tsinghua.edu.cn}
\affiliation{Yau Mathematical Sciences Center, Tsinghua University, Beijing 100084, China}
\affiliation{Department of Mathematics, Tsinghua University, Beijing 100084, China}
\affiliation{Yanqi Lake Beijing Institute of Mathematical Sciences and Applications, Beijing 100407, China}

%\date{\today}

\begin{abstract}
  Achieving practical quantum advantage on near-term noisy hardware is a central goal of quantum computation. 
  However, without efficient pre-execution diagnostics, circuit design and scheme selection often rely on costly hardware-in-the-loop trial-and-error, inflating experimental overhead and impeding progress. 
  To address this challenge, we introduce 2MC-OBPPP, a polynomial-time classical estimator that, for parameterized quantum circuits, jointly estimates trainability, expressibility, and robustness to noise.
  For example, our approach visually demonstrates that moderate amplitude damping alleviates barren plateaus (improving trainability) while decreasing expressibility.
  Moreover, the method produces a spatiotemporal ``noise-hotspot" map that pinpoints the most noise-sensitive qubits/gates, enabling targeted noise suppression. In a representative circuit, interventions on fewer than $2\%$ qubits reduce the error up to $90\%$.
  Together, before execution, our approach provides an efficient diagnostic benchmark for circuit/scheme design, and in deployment, guides for targeted interventions that substantially reduce the cost of error suppression.
\end{abstract}

\maketitle

% ------------------------ MAIN BODY ------------------------
\section{Introduction}
Recent developments in quantum hardware have raised hopes for achieving practical quantum advantages in the near future~\cite{preskill2018quantum}.
In this context, parameterized quantum circuits~(PQCs) play a central role, which are widely employed in variational quantum algorithms~(VQAs)~\cite{cerezo2021variational,bharti2022noisy,kandala2017hardware} and quantum machine learning~(QML)~\cite{biamonte2017quantum,cerezo2022challenges,cong2019quantum,abbas2021power}.
However, to realize practical utility, the designed PQCs must satisfy three essential criteria: trainable, sufficiently expressive to approximate intended targets, and robust to hardware noise.
Therefore, assessing these properties and locating the most influential circuit elements are crucial steps toward benchmarking current circuits and informing the development of future architectures.

% Yet, their practical performance often hinges on subtle ``critical points''—circuit structures, depths, or noise strengths at which trainability collapses, output states lose expressibility, or hardware errors become dominant.  
% Identifying these thresholds, along with the circuit components that most strongly influence them, is crucial for benchmarking existing devices and informing the design of architectures.

A principled route to this goal is to quantify how noise and circuit structure jointly affect three key figures of merit: 
(i) the mean-squared error~(MSE) between noisy and ideal expectation values that we want, which reflects the circuit's robustness to noise; 
(ii) the variance of the parameter gradient, whose exponential decay signals barren plateaus~\cite{mcclean2018barren,larocca2025barren} and thus a loss of trainability~\cite{larocca2025barren,mcclean2018barren,holmes2022connecting,Yu_2024,ortiz2021entanglement,wang2021noise,fontana2024characterizing,ragone2024lie}; 
and (iii) the deviation from the Haar ensemble, which serves as a proxy for expressibility~\cite{sim2019expressibility,holmes2022connecting,thanasilp2024exponential,Yu_2024,nakaji2021expressibility}.
Exact evaluation of these quantities generally scales exponentially with the number of qubits, whereas existing polynomial-time methods are typically restricted to Pauli noise or Clifford circuits~\cite{gottesman1998heisenberg,aaronson2004improved,heyraud2023efficient}, excluding many realistic scenarios.

\begin{figure}[htbp]
  \centering
  \includegraphics[width=\columnwidth]{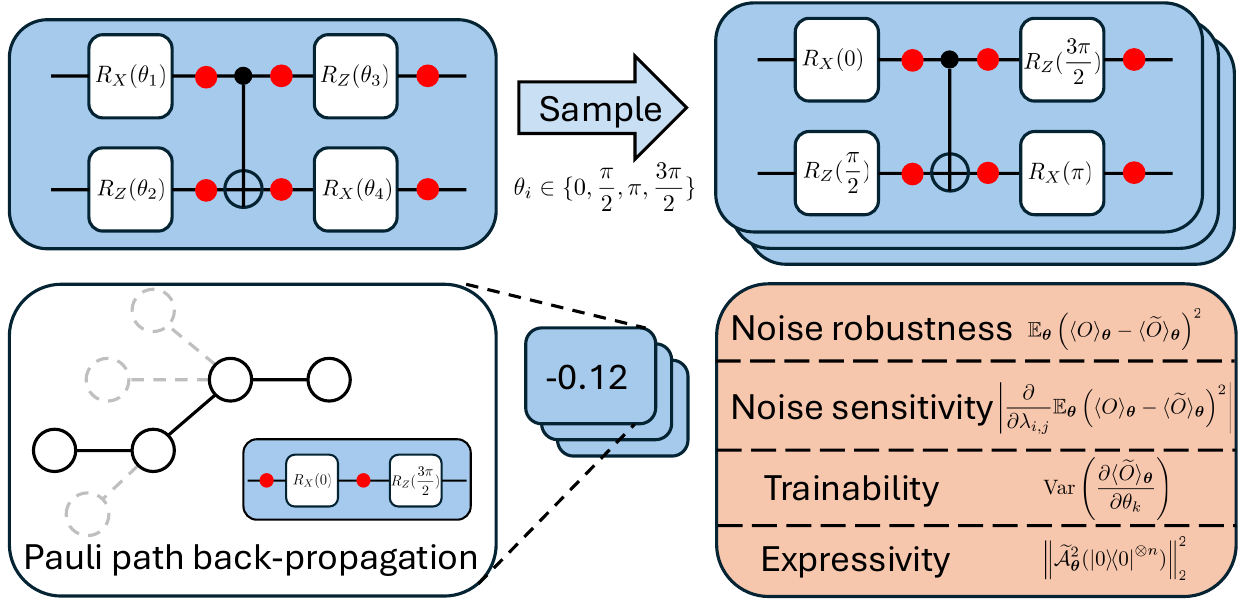}
  \caption{Workflow overview of 2MC-OBPPP.
  For a given PQC, each rotation angle is sampled from a Clifford-compatible discrete set~(blue arrow). 
  For each sampled configuration of rotation angles, the algorithm starts from Pauli terms in final observable and recursively back-propagates the operators through the circuit. 
  Upon encountering a PCS1 channel~(red points), a Pauli operator~(solid circles) is sampled from the set of candidate operators~(dashed circles), the resulting sequence of Pauli operators defines a Pauli path~(solid line) and produces a sample value.
  Finally, post-processing and averaging over samples yield estimators of the diagnostic quantities.
  }
  \label{fig:main_fig}
\end{figure}

In this work, we introduce \textit{2-moment Monte Carlo using observable's back-propagation on Pauli paths}~(2MC-OBPPP), a classical framework for estimating the three aforementioned metrics in any noisy PQCs, where the noise can belong to a broad class of non-unital channels. The algorithm combines Monte Carlo sampling with Pauli path back-propagation~\cite{gao2018efficient,aharonov2023polynomial,shao2024simulating,schuster2024polynomial,mele2024noise,zhang2024clifford,angrisani2025simulating,fontana2025classical,angrisani2024classically,martinez2025efficient,li2025dual} to yield unbiased estimators.
In particular, the computational cost to estimate each metric to a prescribed accuracy (e.g., additive error $\varepsilon$ with high probability) scales linearly with qubit number and circuit depth without any assumptions on circuit connectivity.
Using this approach, we numerically observe that amplitude damping suppresses the barren plateau by increasing the variance of gradient, albeit at the expense of reduced expressibility.

Beyond estimating these metrics, the same framework can also identify ``noise bottlenecks''—gates or qubits whose errors have the greatest impact on the circuit output.
This information enables targeted pulse-shaping~\cite{motzoi2009simple,chen2016measuring,werninghaus2021leakage,evered2023high} and other noise-reduction techniques, directing interventions to the elements where they yield the greatest benefit and thereby improving circuit performance with minimal overhead.
In particular, by focusing on the identified bottlenecks, our approach can significantly reduce the cost of quantum error mitigation~(QEM)~\cite{takagi2022fundamental,takagi2023universal,tsubouchi2023universal,quek2024exponentially}, allowing it to scale to larger system sizes.
In a typical quantum circuit, we show that, guided by our approach, interventions on fewer than $2\%$ qubits can reduce the error up to $90\%$.
Altogether, our approach offers an efficient and scalable toolkit to quantify key properties of large-scale PQCs, and to guide hardware improvements and targeted noise-reduction strategies, thereby supporting progress toward practical quantum advantage.

\section{2MC-OBPPP}
An $n$-qubit parameterized quantum circuit~(PQC) is specified by  $\mathcal{C}(\bm{\theta}) = U_{N_g}(\theta_{N_g}) \cdots U_1(\theta_1)$, where $\bm{\theta} = (\theta_1, \theta_2, \ldots, \theta_{N_g})$ with $\theta_i \in \left[0, 2\pi \right)$ collects the gate parameters, and $N_g$ is the number of parameters.
We assume that each parameterized gate can be written as a Pauli rotation followed by a non-parameterized Clifford gate $U_i(\theta_i) = C_i e^{-i\frac{\theta_i}{2} P}$, where $P \in \mathcal{P}_n = \{\mathbb{I}, X, Y, Z\}^{\otimes n}$.

Except for ideal unitary operations, real quantum devices inevitably suffer from noise. To capture these effects and other non-unitary processes, we consider the following \textit{Pauli column-wise sum at most one}~(PCS1) channels encountered in the circuit:
\begin{definition}[PCS1]
  Let $\mathcal{S}_{\mathcal{E}}$ be the Pauli transfer matrix of a channel $\mathcal{E}$, with entries $\left( \mathcal{S}_\mathcal{E}\right)_{i,j}=\tr{\mathcal{E}(\sigma_i) \sigma_j}$, where $\{\sigma_i\}$ is the normalized Pauli basis.
  The channel $\mathcal{E}$ is PCS1 if for each column $j$, $\sum_{i} \abs{\left(\mathcal{S}_\mathcal{E}\right)_{i,j}}\leq 1$ holds.
\end{definition}
PCS1 channels encompass a broad class of channels that capture many realistic scenarios, such as depolarizing noise, Pauli error channel, amplitude damping, thermal relaxation process, and mid-circuit measurement with Pauli feed-forward~(MMPF, defined as
\scalebox{0.45}{
\begin{quantikz}[baseline=(current bounding box.center)]
  \meter{} &  \cwbend{1} \\
  \qw      & \gate{U} & \qw
\end{quantikz}} for $U$ being Pauli), making them a representative framework for modeling noise in quantum circuits, more details are provided in \sm.
We denote the circuit with PCS1 channels as $\widetilde{\mathcal{C}}(\bm{\theta})$.
Without loss of generality, we assume that each gate $U_i(\theta_i)$ is followed by a PCS1 channel $\mathcal{E}_{i}$, so that $\widetilde{\mathcal{C}}(\bm{\theta}) = \mathcal{E}_{N_g} \circ U_{N_g}(\theta_{N_g}) \cdots \mathcal{E}_{1} \circ U_1(\theta_1)$. We write $\widetilde{\mathcal{C}}(\bm{\theta})(\cdot)$ for its action on density matrices.

Several key properties of a PQC—such as trainability, expressibility, and robustness to noise—play a critical role in determining its performance.
In this work, we focus on a class of particular quantities that capture these properties and share a common structure: each is determined by the circuit's $2$-fold channel~\cite{roberts2017chaos}, defined as $\mathcal{E}_2(\mathcal{C}) = \mathbb{E}_{\bm{\theta}} \widetilde{\mathcal{C}}(\bm{\theta})^{\otimes 2}(\cdot)$, where $\bm{\theta}$ is drawn uniformly from $\left[0, 2\pi\right)^{N_g}$.
As shown in Ref.~\cite{heyraud2023efficient,zhu2025scalable,chen2025quantum}, the continuous average can be replaced by an average over the discrete set $\{0, \frac{\pi}{2}, \pi, \frac{3\pi}{2}\}^{N_g}$.
Consequently, these diagnostic quantities can be written as $\mathbb{E}_{\bm{\theta}} h_{*}(\bm{\theta})$, where $\bm{\theta}$ is sampled uniformly from $\{0, \frac{\pi}{2}, \pi, \frac{3\pi}{2}\}^{N_g}$, and $h_{*}(\bm{\theta})$ denotes a bounded function determined by the diagnostic of interest. More details are provided in \sm.

To investigate the circuit properties under PCS1 channel settings—particularly beyond the Pauli noise regime, which is typically assumed in stabilizer-based methods—we develop 2MC-OBPPP.
As illustrated in Fig.~\ref{fig:main_fig}, 2MC-OBPPP consists of two parts:
\begin{enumerate}[label=\arabic*., topsep=1pt, itemsep=0pt, parsep=0pt, leftmargin=*]
  \item Outer-layer (over $\bm{\theta}$): sample $\bm{\theta}$ uniformly from $\{0, \frac{\pi}{2}, \pi, \frac{3\pi}{2}\}^{N_g}$; if needed, also sample additional randomized configurations; call the inner layer to estimate $h_{*}(\bm{\theta})$, and average over samples to approximate $\mathbb{E}_{\bm{\theta}} h_{*}(\bm{\theta})$.
  \item Inner-layer (estimating $h_{*}(\bm{\theta})$): apply Pauli path back-propagation and Monte Carlo sampling to produce an unbiased estimator of $h_{*}(\bm{\theta})$.
\end{enumerate}
For an $n$-qubit circuit with $L$ layers, an $\varepsilon$-accurate estimate of these diagnostic quantities can be obtained with success probability at least $1-\delta$ at a computational cost scaling as:
\begin{equation}
  \order{\frac{\norm{O}_{\mathrm{Pauli},1}^2 \norm{O}_\infty^6 (nL + n\norm{\rho}_0)}{\varepsilon^4} \ln{\frac{\norm{O}_\infty^2}{\varepsilon}}\ln{\frac{1}{\delta}} },
\end{equation}
provided that the PCS1 channels are $\order{1}$-local (i.e., each $\mathcal{E}_i$ acts on at most a constant number of qubits).
Here, $\norm{\rho}_0$ denotes the number of non-zero entries in the density matrix of the input state $\rho$; for $\ketbra{0}{0}^{\otimes n}$, we have $\norm{\rho}_0=1$.
For an observable $O=\sum_{h=1}^{N_O} c_h P_h$ with $P_h \in \mathcal{P}_n$, $\norm{O}_\infty$ denotes its spectral norm, and $\norm{O}_{\mathrm{Pauli},1}\coloneq \sum_{h=1}^{N_O} \abs{c_h}$ is the Pauli $l_1$-norm.
For the ``expressibility'' estimation discussed in the subsequent section, no observable $O$ is involved; thus, the $\norm{O}_\infty$ and $\norm{O}_{\mathrm{Pauli},1}$ in the above expression should be replaced by $1$.
In many QML tasks~\cite{biamonte2017quantum,cerezo2022challenges,cong2019quantum,abbas2021power,jing2025quantum}, the measurement is a Pauli operator or a computational-basis projector, in which case $\norm{O}_\infty \leq \norm{O}_{\mathrm{Pauli},1}=1$.
When $\norm{O}_{\mathrm{Pauli},1}=\order{1}$ and $\norm{\rho}_0=\order{1}$, the cost simplifies to $\order{\frac{nL}{\varepsilon^4} \ln{\frac{1}{\varepsilon}}\ln{\frac{1}{\delta}} }$.
Consequently, these quantities can be estimated efficiently.
Full algorithmic details and overhead analysis are provided in \sm, which also include numerical benchmarks showing that the proposed estimators are unbiased and exhibit low variance.

\section{Noise Robustness and Noise Bottleneck}
The noise robustness of a circuit can be quantified by the MSE between the ideal and noisy expectation values of an observable $O$ with respect to the parameters $\bm{\theta}$:
\begin{small}
\begin{equation}\label{eq:noise_mse}
  \mathrm{MSE}(\langle O \rangle) = \mathbb{E}_{\bm{\theta}} \left( \langle O \rangle_{\bm{\theta}} - \langle \widetilde{O} \rangle_{\bm{\theta}} \right)^2,
\end{equation}
\end{small}
where $\mathbb{E}_{\bm{\theta}}$ denotes the expectation over $\bm{\theta}$ sampled uniformly from $\left[0, 2\pi \right)^{N_g}$, $\langle O \rangle_{\bm{\theta}}=\tr{O \mathcal{C}(\bm{\theta})\rho \mathcal{C}(\bm{\theta})^\dagger}$ denotes the expectation value of $O$ under the ideal quantum circuit $\mathcal{C}(\bm{\theta})$, and $\langle \widetilde{O} \rangle_{\bm{\theta}}$ denotes the corresponding value under the noisy circuit $\widetilde{\mathcal{C}}(\bm{\theta})$.

In realistic quantum devices, noise is ubiquitous but often affects different circuit components unevenly; errors from particular gates or qubits can impact the computation more adversely than others.
We refer to this especially sensitive subset of circuit elements as the \emph{noise bottleneck}.
To characterize the bottleneck, we consider a gate-based noise model in which each gate $U_{i,j}$ is followed by a noise channel $\mathcal{N}_{i,j}$; here the index $(i,j)$ labels the $j$-th gate in the $i$-th layer.
An analogous formulation applies to a qubit-based noise model where $(i,j)$ instead indexes the $j$-th qubit in the $i$-th layer.
Each $\mathcal{N}_{i,j}$ is assumed to be a PCS1 channel acting on at most a constant number of qubits.

The sensitivity of a specific element $(i,j)$ in the circuit can be quantified by:
\begin{small}
\begin{equation}\label{eq:noise_gradient}
  \abs{\frac{\partial \mathrm{MSE}(\langle O \rangle)}{\partial \lambda_{i,j}}} = \abs{  \frac{\partial }{\partial \lambda_{i,j}} \mathbb{E}_{\bm{\theta}} \left( \langle O \rangle_{\bm{\theta}} - \langle \widetilde{O} \rangle_{\bm{\theta}} \right)^2 },
\end{equation}
\end{small}
where $\lambda_{i,j}$ is a parameter of the noise channel $\mathcal{N}_{i,j}$ (e.g., Pauli error rate, amplitude damping rate, thermal relaxation coefficient).
This sensitivity is closely related to the classical Fisher information~\cite{meyer2021fisher} about $\lambda_{i,j}$ encoded in $\langle \widetilde{O} \rangle_{\bm{\theta}}$, more details are provided in \sm.
A larger value indicates that variations in $\lambda_{i,j}$ exert a stronger influence on the MSE.
Therefore, prioritizing noise suppression for elements with the largest sensitivity provides a cost-effective strategy to reduce the MSE, which motivates the term ``noise bottleneck".

\begin{figure*}[!htbp]
  \centering
  \subfigure[]{
    \includegraphics[width=0.36\textwidth]{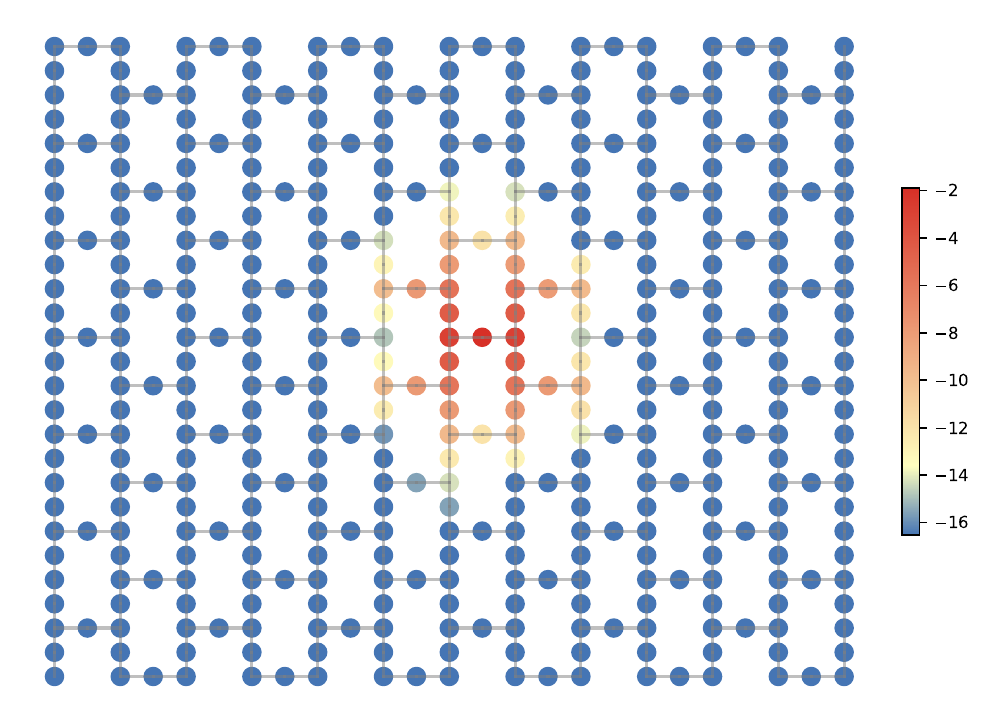}
  }
  \hspace{0.1\textwidth}
  \subfigure[]{
    \includegraphics[width=0.35\textwidth]{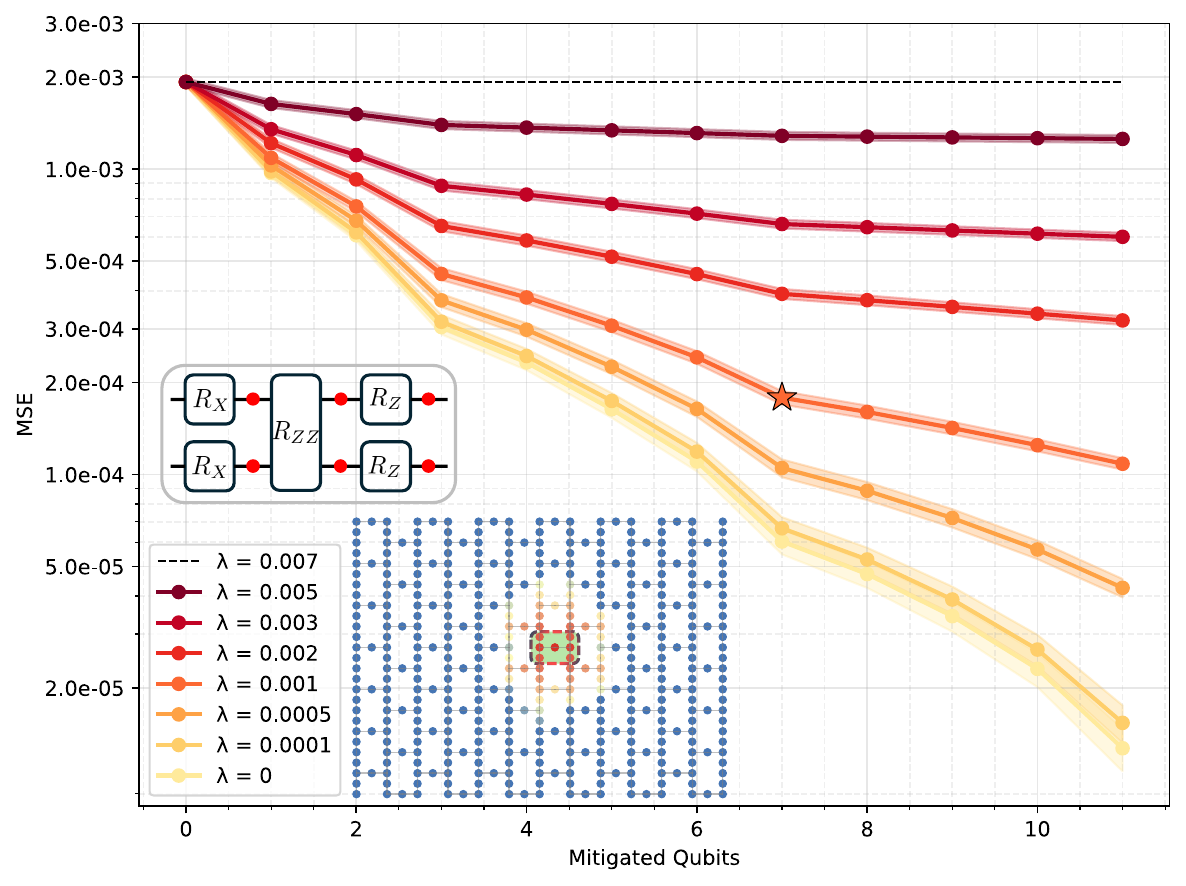}
  }
  \caption{(a) Sensitivity of qubits in a 75-layer, 435-qubit circuit with 1-weight measurement. Dots represent qubits, and lines indicate two-qubit gates.
  The color of each qubit reflects the logarithm of the gradient norm of the MSE~\eqref{eq:noise_gradient} with respect to the qubit's noise strength.
  (b) Reduction of MSE~\eqref{eq:noise_mse} under the bottleneck-first error-suppression strategy.
  The x-axis denotes the number of mitigated qubits, and the y-axis shows the corresponding MSE. 
  The legend of $\lambda$ presents the noise strength of the selected qubits after error mitigation.
  Shaded regions represent the standard deviation over $100$ independent trials.
  The $\star$ marks the case where $7$ qubits are mitigated to $\lambda=0.001$.
  The inset highlights the $7$ qubits with the largest gradient norms on the chip layout.
  }
  \label{fig:chip435}
\end{figure*}

As a concrete example, we identify the noise bottlenecks in a 75-layer, 435-qubit circuit, designed on a hypothetical chip architecture and composed of $R_X$, $R_Z$ and $R_{ZZ}$ gates.
Further details on the circuit design are provided in \sm.
The noise model is given by qubit-wise depolarizing noise, $\mathcal{N}_{dep}(\rho)=(1-\lambda)\rho + \lambda\frac{\mathbb{I}}{2}$, applied gate-by-gate.
Here, the depolarizing rate $\lambda$ is uniformly set to $\lambda=0.007$ across all qubits, matching the fitted noise strength of IBM's Eagle quantum device~\cite{shao2024simulating}.
The results are shown in Fig.~\ref{fig:chip435}(a), with bottlenecks highlighted in red.

These bottlenecks are naturally the prime targets for noise-reduction techniques, such as QEM~\cite{cai2023quantum,van2023probabilistic,kandala2019error} or enhanced quantum control~\cite{rol2017restless,motzoi2009simple,chen2016measuring,werninghaus2021leakage,evered2023high}.
In contrast, globally applied QEM schemes face fundamental scalability limitations—the overhead can grow exponentially in the worst case~\cite{takagi2022fundamental,takagi2023universal,tsubouchi2023universal,quek2024exponentially}.
Consequently, concentrating mitigation efforts on the identified bottlenecks can substantially reduce the overall cost of noise suppression.
As demonstrated in Fig.~\ref{fig:chip435}(b), reducing the noise strength to $\lambda=0.001$ on fewer than $2\%$ of the qubits~(7 out of 435), leads to a tenfold reduction in the MSE.

\section{Trainability Estimation}
A PQC $\mathcal{C}(\bm{\theta})$ with measurement observable $O$ is said to exhibit a \textit{barren plateau} if the probability that its gradient significantly deviates from zero decays exponentially with the system size, formally, $\mathbb{P}\left( \abs{\frac{\partial \langle O \rangle_{\bm{\theta}}}{\partial \theta_k}} > \epsilon \right) \leq \order{\exp{ -\alpha n}}$ for all $k$ , where $n$ is the system size, and $\alpha>0$ is a constant.
This phenomenon was related to various properties of circuits, including expressibility~\cite{mcclean2018barren,holmes2022connecting,Yu_2024}, entanglement~\cite{ortiz2021entanglement,patti2021entanglement} and noise~\cite{wang2021noise}.
A key quantity for diagnosing barren plateaus—and more broadly, for assessing trainability—is the variance of gradient, defined as:
\begin{small}
\begin{equation}\label{eq:trainability_variance}
  \mathrm{Var}\left(\frac{\partial \langle \widetilde{O} \rangle_{\bm{\theta}}}{\partial \theta_k}\right) = \mathbb{E}_{\bm{\theta}} \left( \frac{\partial \langle \widetilde{O} \rangle_{\bm{\theta}}}{\partial \theta_k} \right)^2 - \left( \mathbb{E}_{\bm{\theta}} \frac{\partial \langle \widetilde{O} \rangle_{\bm{\theta}}}{\partial \theta_k} \right)^2.
\end{equation}
\end{small}
By Chebyshev's inequality, we have $\mathbb{P}\left( \abs{\frac{\partial \langle \widetilde{O} \rangle_{\bm{\theta}}}{\partial \theta_k}} > \epsilon \right) \leq {\mathrm{Var}(\frac{\partial \langle \widetilde{O} \rangle_{\bm{\theta}}}{\partial \theta_k})} \big/ {\epsilon^2}$, which implies that a vanishing variance for all parameters directly implies a barren plateau.
Conversely, a non-vanishing variance guarantees large fluctuations of the randomly initialized gradients, thereby indicating favorable trainability in the early stages of optimization~\cite{mcclean2018barren,wang2021noise,cerezo2021variational}.

Theoretically, representation-theoretic tools can, in principle, compute the variance of gradient exactly, but only for circuits with structural features (such as symmetries or algebraic regularity) that make the analysis tractable~\cite{fontana2024characterizing,ragone2024lie}.
In the absence of such structures, analytical approaches become difficult to apply, and numerical methods are needed.
In the noiseless setting, stabilizer-based methods are applicable~\cite{heyraud2023efficient}.
However, in the presence of noise, the picture is more nuanced: noise can distort the loss function landscape and potentially induce barren plateaus, whereas under appropriate conditions, there is evidence that non-unital noise channels and MMPFs can alleviate this issue~\cite{wang2021noise,mele2024noise,fefferman2024effect,deshpande2024dynamic,cao2025measurement}.
As a result, trainability analysis in noisy circuits (especially non-unital channels) has attracted increasing attention~\cite{mele2024noise,fefferman2024effect,deshpande2024dynamic,cao2025measurement}.

\begin{figure}[!htbp]
  \centering
  \includegraphics[width=0.85\columnwidth]{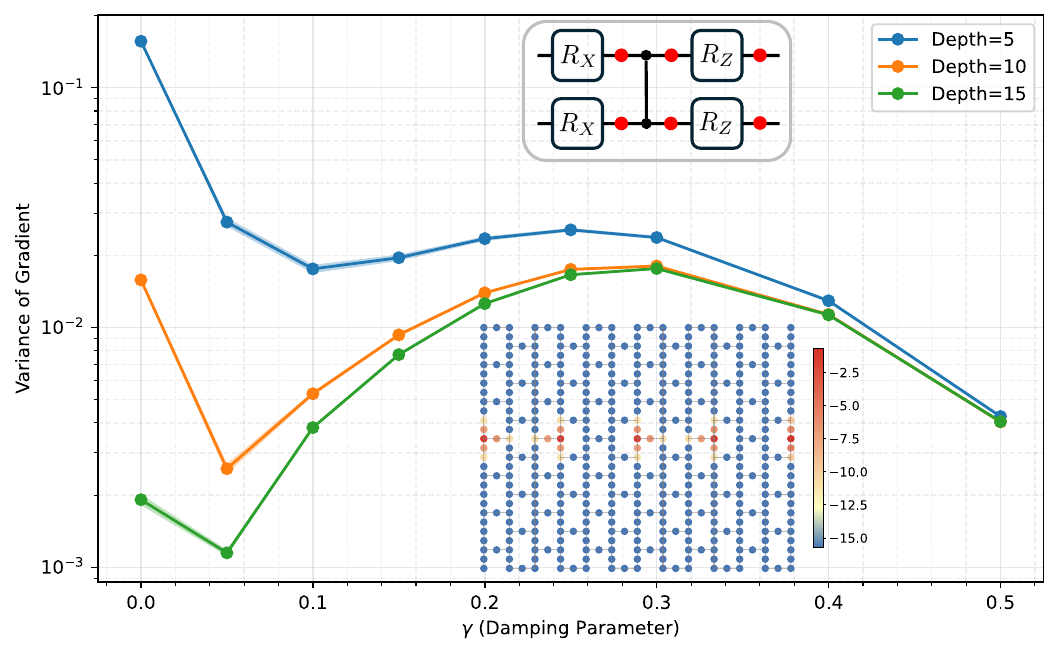}
  \caption{
    The sum of the variance of gradient over all parameters $\sum_k \mathrm{Var}\left(\frac{\partial \langle \widetilde{O} \rangle_{\bm{\theta}}}{\partial \theta_k}\right)$ in circuits with 5-weight measurement under amplitude damping noise.
    The x-axis represents the damping intensities $\gamma$.
    Different colors indicate different circuit depths.
    The shaded area denotes the standard deviation over $100$ independent trials.
    The inset shows the spatial distribution of the variance of gradient at $\gamma=0.25$ and circuit depth $15$.
  }
  \label{fig:variance_estimation}
\end{figure}

As a demonstration, we employ 2MC-OBPPP to analyze circuits subjected to amplitude damping noise, $\mathcal{N}_{ad}\left(\begin{pmatrix}
  \rho_{00} & \rho_{01} \\
  \rho_{10} & \rho_{11}
\end{pmatrix}\right)=\begin{pmatrix}
  \rho_{00}+\gamma\rho_{11} & \sqrt{1-\gamma}\rho_{01} \\
  \sqrt{1-\gamma}\rho_{10} & (1-\gamma)\rho_{11}
\end{pmatrix}$, where $\gamma$ is the damping parameter. 
The circuits follow the same hypothetical chip architecture shown in Fig.~\ref{fig:chip435}(a).
The results in Fig.~\ref{fig:variance_estimation} show that the variance exhibits the expected significant decay with increasing circuit depth.
Remarkably, as the damping strength $\gamma$ increases, the variance exhibits a non-monotonic behavior: it initially decreases, then increases, and eventually decreases again.
The intermediate rise reflects the noise-induced shallow-circuit effect~\cite{mele2024noise}, where moderate amplitude damping reduces the effective circuit depth, leading to a shallower overall evolution.
Consequently, the variance increases, indicating that the barren plateau is temporarily mitigated.
This observation suggests the existence of a noise-induced ``sweet spot", where moderate amplitude damping can partially alleviate the barren plateau.
However, as $\gamma$ increases further, excessive decoherence overwhelms the evolution dynamics, causing the variance to decrease once more.
These results highlight a nuanced role of noise in training processes, revealing that certain types and magnitudes of noise may transiently improve trainability before ultimately degrading it.

Beyond providing a circuit-level trainability diagnostic, our method can also pinpoint gates that are effectively untrainable (Fig.~\ref{fig:variance_estimation} inset), in the sense that the variance of their gradients vanishes.
Such gate-resolved information can inform circuit design and training strategies, e.g., by adjusting the architecture to avoid these gates or by removing them to simplify the circuit.

In addition to the estimation method, theoretically, the Pauli path provides a direct and intuitive perspective on the emergence of the noise-induced barren plateau phenomenon~\cite{wang2021noise}.
This reveals how noise suppresses nontrivial Pauli paths, leading to vanishing gradients.
More details can be found in \sm.

\section{Expressibility Estimation}
The expressibility of a circuit characterizes its ability to explore the Hilbert space.
A common approach for assessing expressibility involves comparing the circuit-induced ensemble with the Haar ensemble~\cite{sim2019expressibility,holmes2022connecting,thanasilp2024exponential,Yu_2024,nakaji2021expressibility}.
Focusing on the 2-th moment, the deviation is characterized by:
\begin{small}
\begin{equation}\label{eq:deviation_expressibility}
  \mathcal{A}_{\mathcal{C}}^{2}(\cdot)=\mathbb{E}_{U} U^{\otimes 2}(\cdot)^{\otimes 2}\left(U^{\dagger}\right)^{\otimes 2}-\mathbb{E}_{\bm{\theta}} \mathcal{C}(\bm{\theta})^{\otimes 2}(\cdot)^{\otimes 2}\left(\mathcal{C}(\bm{\theta})^{\dagger}\right)^{\otimes 2},
\end{equation}
\end{small}
where $U$ is drawn from the Haar measure. 
We quantify the deviation by the Hilbert-Schmidt~(HS) norm: $\mathcal{M}_2^2=\norm{\mathcal{A}_{\mathcal{C}}^{2}(\ketbra{0}^{\otimes n})}_2^2$, with smaller values indicating higher expressibility (i.e., closer to a random ensemble).

While previous studies have primarily addressed noiseless setting, we consider an analogous metric for the noisy setting:
\begin{small}
\begin{equation}\label{eq:2_moment_expressibility}
  \widetilde{\mathcal{M}}_2^2=\norm{\mathcal{A}_{\widetilde{\mathcal{C}}}^{2}(\ketbra{0}^{\otimes n})}_2^2,
\end{equation}
\end{small}
where $\mathcal{A}_{\widetilde{\mathcal{C}}}^{2}(\cdot)$ is obtained from Eq.~\eqref{eq:deviation_expressibility} by replacing $\mathcal{C}(\bm{\theta})$ with the noisy circuit $\widetilde{\mathcal{C}}(\bm{\theta})$. 
Unlike the robustness and trainability estimators, which only require channels in the circuits to be PCS1, computing $\widetilde{\mathcal{M}}_2^2$ via 2MC-OBPPP additionally requires that their adjoints are PCS1.
Nevertheless, when this condition is not met, we provide a computable lower bound, denoted $\widetilde{\mathcal{M}}_{2\leq}^2\leq \widetilde{\mathcal{M}}_2^2$, which serves as a proxy for noisy expressibility.
Further details are provided in \sm.

Fig.~\ref{fig:variance_estimation} shows that moderate damping~($\gamma\in[0.2,0.3]$) can alleviate the barren plateau.
To provide a complementary perspective, we investigate how amplitude damping influences expressibility in Fig.~\ref{fig:expressibility_estimation}, where $\widetilde{\mathcal{M}}_{2\leq}^2$ increases with $\gamma$ and can exceed its noiseless counterpart $\mathcal{M}_2^2$ when $\gamma>0.2$.
Taken together, these results reveal a fundamental trade-off: whilst amplitude damping can improve trainability in certain scenarios, it simultaneously drives the circuit ensemble away from a unitary 2-design, thereby degrading its expressibility.

In addition to the HS norm, we derive a computable lower bound $\norm{\mathcal{A}_{\widetilde{\mathcal{C}}}^{2}(\ketbra{0}^{\otimes n})}_1$ that can explain the barren plateau phenomenon induced by strong expressibility, as previously established in Ref.~\cite{holmes2022connecting}.
It offers a new perspective on the relationship between expressibility and trainability.
More details can be found in \sm.

\begin{figure}[tbp]
  \centering
  \includegraphics[width=0.85\columnwidth]{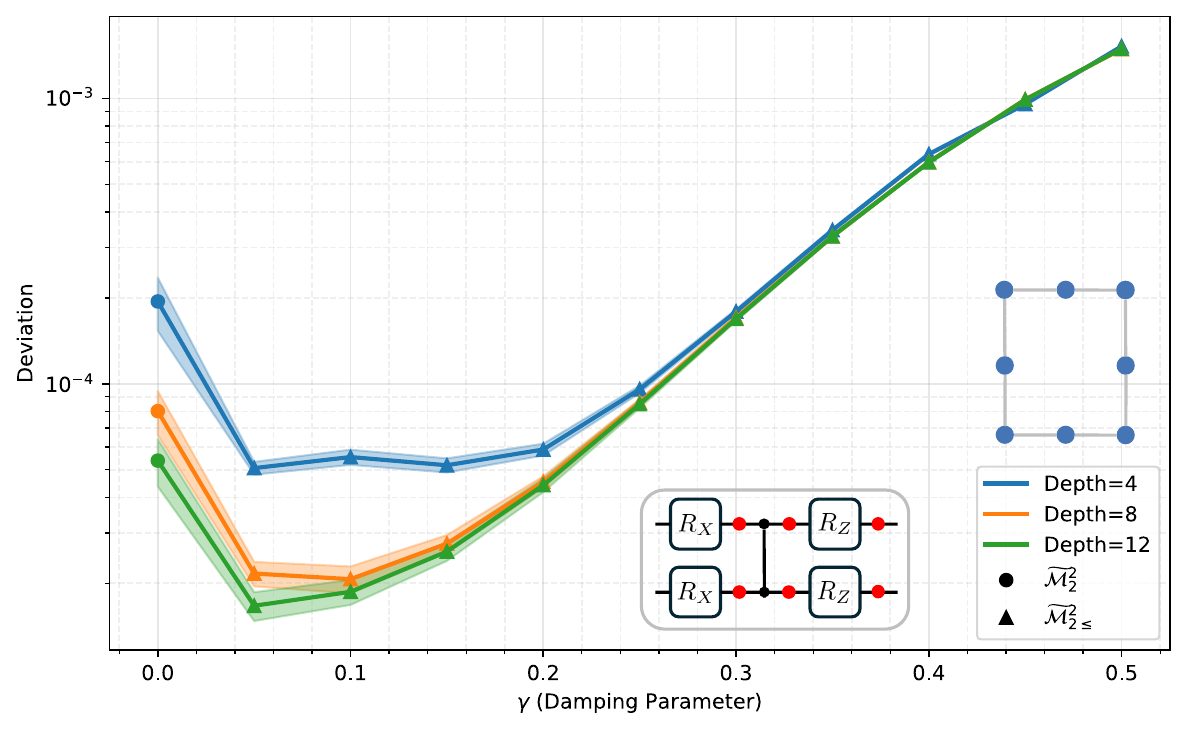}
  \caption{
    Deviation from $2$-design for circuits with 8 qubits under amplitude damping noise.
    The circuits comprise $R_X$, $R_Z$ and nearest-neighbor $CZ$ gates on a ring, which can be viewed as a reduced unit used in Fig.~\ref{fig:variance_estimation}.
    The x-axis represents the damping intensities $\gamma$; different colors indicate different circuit depths. 
    Markers $\bullet$ and $\blacktriangle$ denote the noiseless metric $\mathcal{M}_2^2$ and the noisy lower bound $\widetilde{\mathcal{M}}_{2\leq}^2$, respectively.
    The shaded area denotes the standard deviation over $100$ independent trials.
  }
  \label{fig:expressibility_estimation}
\end{figure}

\section{Discussions}
We introduced the 2MC-OBPPP framework to estimate several key metrics for quantum circuits efficiently: the noise-induced MSE, the variance of gradient, the deviation from a random unitary ensemble, and establish a noise bottleneck analysis.
These results have immediate practical utility. 
The bottleneck map provides a guide for allocating error mitigation or pulse-shaping resources, replacing an indiscriminate noise-reduction strategy with a focused one. 
Fast classical benchmarking enables the prescreening of ansatzes and depths before hardware deployment, thereby shortening the design-test cycle of quantum algorithms.

Future work could examine how the introduction of MMPFs affects circuit performance, particularly under realistic noise models.
Ultimately, integrating this framework into compilation pipelines, extending it to time-dependent noise and more diverse mid-circuit measurements with feed-forward, will establish a noise-aware toolkit that advances the co-design of hardware and algorithms toward practical quantum advantage.

\begin{acknowledgments}
We thank Song Cheng, Weikang Li, Ziwen Liu, Weixiao Sun, Fuchuan Wei and Ruiqi Zhang for valuable discussions.
YS, CZ, ZW and ZL were supported by Beijing Natural Science Foundation Key Program (Grant No. Z220002).
YS and ZL were supported by BMSTC and ACZSP (Grant No. Z221100002722017). ZW was supported by Beijing Science and Technology Planning Project (Grant No. Z25110100810000).
ZC and ZW were also supported by the National Natural Science Foundation of China (Grant Nos. 62272259 and 62332009).
ZL was supported by NKPs (Grant No. 2020YFA0713000).\\
\end{acknowledgments}

%\nocite{*}
%\bibliographystyle{apsrev}
\bibliography{main}

% 开始附录部分
\clearpage

\onecolumngrid % 切换为单栏布局
\section*{Supplemental Material}
\appendix
\tableofcontents

\clearpage

\section{Pauli path integral formalism}
\label{sec:pauli_path_integral}

\subsection{Pauli path integral for ideal quantum circuits}
In the current NISQ era, parameterized quantum circuits~(PQC) are widely used in many near-term algorithms~\cite{kandala2017hardware,farhi2014quantum}.
A typical $n$-qubit PQC, denoted as $\mathcal{C}(\bm{\theta})$, consists of a sequence of Pauli rotation gates and non-parameterized Clifford gates.
The Pauli rotation gates are represented as $e^{-i\frac{\theta}{2} P}$, where $P\in\{\mathbb{I},X,Y,Z\}^{\otimes n}$. The Clifford gates are the unitary operators that normalize the Pauli group $Cl_n \coloneqq \{C\in U(2^n) \mid C\mathcal{P}_nC^\dagger=\mathcal{P}_n\}$, where $\mathcal{P}_n$ is the Pauli group on $n$ qubits. Any unitary operator $U\in Cl_n$ is equivalent to a circuit generated using Hadamard, CNOT, and phase gates $S$~\cite{gottesman2016surviving}.

Without loss of generality, we assume that PQCs follow the form:
\begin{equation}\label{eq:parameterized_circuit}
  \mathcal{C}(\bm{\theta})=U_{N_g}(\theta_{N_g})  \cdots {U}_1(\theta_1),
\end{equation}
where $\bm{\theta}=(\theta_1,\cdots,\theta_{N_g})$ are rotation angles and ${N_g}$ is the number of parameters. Each unitary ${U}_i(\theta_i):=\exp{-i \frac{\theta_i}{2} P_i}C_i $ comprises a Clifford operator $C_i$ and a rotation $\exp{-i \frac{\theta_i}{2} P_i}$ on Pauli operator $P_i\in\{\mathbb{I},X,Y,Z\}^{\otimes n}$ with angle $\theta_i$.

In this context, the quantum circuit $\mathcal{C}(\bm{\theta})$ is applied to an initial state $\rho$, and what we are interested in is the expectation value of an observable $O$, given by
\begin{equation}\label{eq:expectation_value_noiseless}
  \langle O \rangle = \tr{O \mathcal{C}(\bm{\theta})\rho \mathcal{C}(\bm{\theta})^\dagger}.
\end{equation}
Without loss of generality, we assume that the observable is traceless, i.e., $\tr{O}=0$, otherwise we can replace $O$ with $O-\frac{\tr{O}}{\tr{\mathbb{I}^{\otimes n}}} \mathbb{I}^{\otimes n}$.

A Pauli path is a sequence $\vec{s}=(s_0,\cdots,s_{N_g})\in \bm{P}^{{N_g}+1}_n$, where $\bm{P}_n=\{\sfrac{\mathbb{I}}{\sqrt{2}},\sfrac{X}{\sqrt{2}},\sfrac{Y}{\sqrt{2}},\sfrac{Z}{\sqrt{2}}\}^{\otimes n} $ represents the set of all normalized $n$-qubit Pauli operators. 
Using the fact that the normalized $n$-qubit Pauli group $\bm{P}_n$ forms a basis of the $2^n \times 2^n$ Hermitian operator space, we can express any observable $O$ as a linear combination of Pauli operators:
\begin{equation}
  O=\sum_{s\in \bm{P}_n} \tr{O s} s,
\end{equation}
Iteratively applying the Pauli operator decomposition, we can express the expectation value as the sum of contributions from all Pauli paths:
\begin{equation}\label{eq:pauli_path_integral_noiseless}
  \begin{aligned}
  \langle O \rangle &= \sum_{s_{N_g}} \tr{O s_{N_g}} \tr{s_{N_g} \mathcal{C}(\bm{\theta})\rho \mathcal{C}(\bm{\theta})^\dagger}\\
  &= \sum_{s_{N_g}} \tr{O s_{N_g}} \tr{s_{N_g} U_{N_g}(\theta_{N_g}) \cdots U_1(\theta_1) \rho U_1^\dagger(\theta_1) \cdots U_{N_g}^\dagger(\theta_{N_g})}\\
  &= \sum_{s_{N_g},s_{{N_g}-1}} \tr{O s_{N_g}} \tr{s_{N_g} U_{N_g}(\theta_{N_g}) s_{{N_g}-1} U_{N_g}^\dagger(\theta_{N_g})} \tr{s_{{N_g}-1} U_{{N_g}-1}(\theta_{{N_g}-1}) \cdots U_1(\theta_1) \rho U_1^\dagger(\theta_1) \cdots U_{{N_g}-1}^\dagger(\theta_{{N_g}-1})}\\
  &\vdots\\
  &= \sum_{s_{N_g},s_{{N_g}-1},\cdots,s_0} \tr{O s_{N_g}} \tr{s_{N_g} U_{N_g}(\theta_{N_g}) s_{{N_g}-1} U_{N_g}^\dagger(\theta_{N_g})} \cdots \tr{s_1 U_1(\theta_1) s_0 U_1^\dagger(\theta_1)} \tr{s_0 \rho}\\
  &= \sum_{s_{N_g},s_{{N_g}-1},\cdots,s_0} \tr{O s_{N_g}} \tr{s_0 \rho} \prod_{i=1}^{{N_g}} \tr{s_i U_i(\theta_i) s_{i-1} U_i^\dagger(\theta_i)} \\
  &=\sum_{\vec{s}} f(\vec{s},\bm{\theta},O,\rho),
  \end{aligned}
\end{equation}
where 
\begin{equation}\label{eq:contribution_function_noiseless}
  f(\vec{s},\bm{\theta},O,\rho)=\tr{O s_{N_g}} \tr{s_0 \rho} \prod_{i=1}^{{N_g}} \tr{s_i U_i(\theta_i) s_{i-1} U_i^\dagger(\theta_i)}
\end{equation}
denotes the contribution of a specfic Pauli path $\vec{s}=(s_0,\cdots,s_{N_g})$ in the expectation value $\langle O \rangle$.

The evolution of the Pauli operator $s$ under the operator $U_i(\theta_i)=\exp{-i \frac{\theta_i}{2} P_i}C_i$ is given by:
\begin{equation}
  \begin{aligned}
    U_i(\theta_i) s_{i-1} U_i^\dagger(\theta_i) &= \exp{-i \frac{\theta_i}{2} P_i} \underbrace{C_i s_{i-1} C_i^\dagger}_Q \exp{i \frac{\theta_i}{2} P_i}\\
    &= \begin{cases}
      Q, & [P_i, Q] = 0, \\
      \cos(\theta_i) Q - i \sin(\theta_i) P_i Q & \{P_i, Q\} = 0,
      \end{cases}
  \end{aligned}
\end{equation}
where $Q=C_i s_{i-1} C_i^\dagger$ is the transformed Pauli operator after applying the Clifford gate $C_i$ to $s_{i-1}$.
The above equation shows that the factor $\tr{s_i U_i(\theta_i) s_{i-1} U_i^\dagger(\theta_i)}$ in $f(\vec{s},\bm{\theta},O,\rho)$ can be expressed as:
\begin{equation}\label{eq:gate_term_in_f}
  \begin{aligned}
    \tr{s_i U_i(\theta_i) s_{i-1} U_i^\dagger(\theta_i)} &= \begin{cases}
      \tr{s_i Q}, & [P_i, Q] = 0, \\
      \cos(\theta_i) \tr{s_i Q} - i \sin(\theta_i) \tr{s_i P_i Q}, & \{P_i, Q\} = 0.
      \end{cases}\\
      &= \begin{cases}
        0, & s_i\neq Q \quad \mathrm{ and } \quad s_i\neq iQP_i  \\
        1, & [P_i, Q] = 0, s_i=Q \\
        \cos(\theta_i) , & \{P_i, Q\} = 0, s_i=Q \\
        \sin(\theta_i) , & \{P_i, Q\} = 0, s_i=iQP_i. \\
      \end{cases}
  \end{aligned}
\end{equation}
Here, we ignore the sign $\pm$ in front of the Pauli operator $s_i$ in the above equation. 

Because of $Q=C_i s_{i-1} C_i^\dagger$, and $C_i$ is Clifford operator, the Pauli operator $Q$ is also a normalized Pauli operator in $\bm{P}_n$ except for the sign $\pm$. 
When $-Q\in \bm{P}_n$ and $s_i$ takes the form of $s_i=-Q$, the equation $\tr{s_i U_i(\theta_i) s_{i-1} U_i^\dagger(\theta_i)}$ equals to $-1$ for the case of $[P_i, Q] = 0$ and $\tr{s_i U_i(\theta_i) s_{i-1} U_i^\dagger(\theta_i)}$ equals to $-\cos(\theta_i)$ for the case of $\{P_i, Q\} = 0$, respectively.
By the property of the Pauli operator $-iXYZ=I$, we can see that each product of a anti-commuting Pauli operator pair equals to a Pauli operator with additional imaginary factor $\pm i$, e.g., $XY=iZ$, $XZ=-iY$, $YZ=iX$. Therefore, when $\{P_i, Q\} = 0$, the operator $iQ P_i$ is also a normalized Pauli operator in $\bm{P}_n$ except for the sign $\pm$.
In summary, the contribution function $f(\vec{s},\bm{\theta},O,\rho)$ is a product of factors, each of which is either $0$, $1$, $\pm \cos(\theta_i)$, or $\pm \sin(\theta_i)$.

\subsection{Discrete rotation angles}
Specifically, if the rotation angle $\theta_i$ takes the value in $\{0,\frac{\pi}{2},\pi,\frac{3\pi}{2}\}$, the Pauli rotation $\exp{-i \frac{\theta_i}{2} P_i}$ falls into the set of Clifford gates, and the factor $\tr{s_i U_i(\theta_i) s_{i-1} U_i^\dagger(\theta_i)}$ in Eq.~\eqref{eq:gate_term_in_f} can be expressed as:
\begin{equation}\label{eq:gate_term_in_f_discrete}
  \begin{aligned}
    &\tr{s_i U_i(\theta_i) s_{i-1} U_i^\dagger(\theta_i)} = \begin{cases}
      \tr{s_i Q}, & [P_i, Q] = 0 \\
      \cos(\theta_i) \tr{s_i Q} - i \sin(\theta_i) \tr{s_i P_i Q}, & \{P_i, Q\} = 0
      \end{cases}\\
      &= \begin{cases}
        0, & [P_i, Q] = 0, s_i\neq Q  \\
        1, & [P_i, Q] = 0, s_i=Q \\
        \pm 1 , & \{P_i, Q\} = 0, s_i=Q \\
      \end{cases}
      \mathrm{when} \quad \theta_i\in\{0,\pi\} \quad \mathrm{or} 
      = \begin{cases}
        0, & [P_i, Q] = 0, s_i\neq Q  \\
        1, & [P_i, Q] = 0, s_i=Q \\
        \pm 1 , & \{P_i, Q\} = 0, s_i=iQP_i \\
      \end{cases}
      \mathrm{when} \quad \theta_i\in\{\frac{\pi}{2},\frac{3\pi}{2}\},
  \end{aligned}
\end{equation}
where $Q=C_i s_{i-1} C_i^\dagger$ and similarly, we ignore the sign $\pm$ in front of the Pauli operator $s_i$ in the above equation. 
When the rotation angle $\theta_i$ take the value in $\{0,\frac{\pi}{2},\pi,\frac{3\pi}{2}\}$, for any given $s_{i-1}$, there is only one Pauli operator $s_i$ that satisfies the equation $\tr{s_i U_i(\theta_i) s_{i-1} U_i^\dagger(\theta_i)}\neq 0$. 
In return, for any given $s_i$, there is only one Pauli operator $s_{i-1}$ that satisfies the equation $\tr{s_i U_i(\theta_i) s_{i-1} U_i^\dagger(\theta_i)}\neq 0$. Therefore, we have the following results.

For the case that rotation angles $\{\theta_i\}$ in the circuit $\mathcal{C}(\bm{\theta})$ take the value in $\{0,\frac{\pi}{2},\pi,\frac{3\pi}{2}\}$, the Pauli path $\vec{s}$ with non-zero contribution $f(\vec{s},\bm{\theta},O,\rho)\neq 0$, is uniquely determined by final Pauli operator $s_{N_g}$. 
Because of the factor $\tr{O s_{N_g}}$ in Eq.~\eqref{eq:contribution_function_noiseless}, to avoid the zero contribution, the final Pauli operator $s_{N_g}$ must be a Pauli operator that is contained in the observable $O$.
Therefore, the expectation value $\langle O \rangle$ can be expressed as:
\begin{equation}\label{eq:pauli_path_integral_discrete}
  \begin{aligned}
  \langle O \rangle &= \sum_{\vec{s}} f(\vec{s},\bm{\theta},O,\rho)\\
  &= \sum_{s_{N_g}\in \{\sigma\mid \tr{\sigma O}\neq 0\}} f(\vec{s}_{\bm{\theta}}(s_{N_g}),\bm{\theta},O,\rho),
  \end{aligned}
\end{equation}
where $\vec{s}_{\bm{\theta}}(s_{N_g})$ is the Pauli path that is uniquely determined by the final Pauli operator $s_{N_g}$ with non-zero contribution $f(\vec{s}_{\bm{\theta}}(s_{N_g}),\bm{\theta},O,\rho)\neq 0$, and $\{\sigma\mid \tr{\sigma O}\neq 0\}$ is the set of Pauli operators that are contained in the observable $O$.

\section{PCS1 channel model}

\subsection{PCS1 channel}

In realistic quantum circuits, the quantum gates are not perfect and the quantum states are subject to noise. These noise can be described by quantum channels $\mathcal{N}$, which are completely positive trace-preserving~(CPTP) maps that transforms an input quantum state $\rho$ into an output state $\mathcal{N}(\rho)$.

In our discussion, we focus on the case that the quantum noise $\mathcal{E}$ is a Pauli column-wise sum at most one channel~(PCS1):
\begin{definition}{(Pauli column-wise sum at most one channel)}\label{def:PCS1}
  Let $\mathcal{S}_{\mathcal{E}}$ be the Pauli transfer matrix of a channel $\mathcal{E}$, with entries $\left( \mathcal{S}_\mathcal{E}\right)_{i,j}=\tr{\mathcal{E}(\sigma_i) \sigma_j}$, where $\{\sigma_i\}$ is the normalized Pauli group $\bm{P}_n$.
  We call $\mathcal{E}$ a PCS1 channel if every column has $l_1$-norm at most 1, i.e., 
  \begin{equation}
    \sum_{i} \abs{\left(\mathcal{S}_\mathcal{E}\right)_{i,j}}\leq 1, \quad \forall j.
  \end{equation}
\end{definition}

In our discussion, we assume that each PCS1 channel is $\order{1}$-local, i.e., the each quantum channel acts on a constant number of qubits at most.
This does not mean that only a constant number of qubits are subject to noise, but rather that the overall noise process to all qubits can be expressed as a tensor product of channels, each acting on at most a constant number of qubits.
This assumption is reasonable in many practical scenarios, and consistent with the locality of the physical interactions that lead to decoherence.
And we don't limit the PCS1 channel to act on adjacent qubits, i.e., the PCS1 channel can act on any qubits in the quantum circuit, as long as the total number of qubits that acts on is a constant, so the long range cross-talks between gates are also contained in our discussion.

Without loss of generality, we assume the PCS1 channels are following the quantum gates in quantum circuits, as shown in Fig~\ref{fig:noisy_gate}.
For the case that the PCS1 channels aren't following the gates, we can assume that there are identity gates before the PCS1 channels. 
We denote the quantum circuit with PCS1 channels as $\widetilde{\mathcal{C}}(\bm{\theta})$, which can be expressed as:
\begin{equation}\label{eq:noisy_circuit}
  \widetilde{\mathcal{C}}(\bm{\theta})(\rho)=\mathcal{N}_{N_g}\circ U_{N_g}(\theta_{N_g}) \circ \cdots \circ \mathcal{N}_1\circ U_1(\theta_1)(\rho),
\end{equation}
where $\mathcal{N}_i$ is a PCS1 channel and $U_i(\theta_i)(\cdot)$ is the conjugate action of the unitary operator $U_i(\theta_i)$, i.e., $U_i(\theta_i)(\rho)=U_i(\theta_i)\rho U_i(\theta_i)^\dagger$.

\begin{figure}[htbp]
  \begin{quantikz}
    \lstick{} & \gate[2]{U} & \gate[wires=2,style={starburst,draw=red,line width=1pt,inner xsep=-4pt,inner ysep=-5pt}]{\mathcal{N}} & \qw \\
    \lstick{} & \qw & \qw & \qw \\
  \end{quantikz}
  \caption{The noisy gate $\widetilde{U}$: ideal gate $U$ followed by noise channel $\mathcal{N}$ acting on the output.}\label{fig:noisy_gate}
\end{figure}
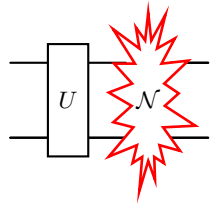

At this time, the noisy expectation value $\langle \widetilde{O} \rangle = \tr{O \widetilde{\mathcal{C}}(\bm{\theta})(\rho)}$, corresponding to Eq.~\eqref{eq:pauli_path_integral_noiseless}, we can express it by introducing an auxiliary Pauli path $\vec{\tau}$, as:
\begin{equation}\label{eq:pauli_path_integral_noisy}
  \begin{aligned}
  \langle \widetilde{O} \rangle = &\sum_{s_{N_g}} \tr{O s_{N_g}} \tr{s_{N_g} \widetilde{\mathcal{C}}(\bm{\theta})(\rho)}\\
  = &\sum_{s_{N_g}} \tr{O s_{N_g}} \tr{s_{N_g} \mathcal{N}_{N_g}\circ U_{N_g}(\theta_{N_g}) \circ \cdots \circ \mathcal{N}_1\circ U_1(\theta_1)(\rho)}\\
  =& \sum_{s_{N_g}} \tr{O s_{N_g}} \tr{ \mathcal{N}^\dagger_{N_g}(s_{N_g})\cdot U_{N_g}(\theta_{N_g}) \circ \cdots \circ \mathcal{N}_1\circ U_1(\theta_1)(\rho)}\\
  =& \sum_{s_{N_g},\tau_{N_g}} \tr{O s_{N_g}} \tr{ \mathcal{N}^\dagger_{N_g}(s_{N_g}) \tau_{N_g}}\tr{\tau_{N_g} U_{N_g}(\theta_{N_g}) \circ \cdots \circ \mathcal{N}_1\circ U_1(\theta_1)(\rho)}\\
  =& \sum_{s_{N_g},\tau_{N_g}} \tr{O s_{N_g}} \tr{s_{N_g} \mathcal{N}_{N_g}(\tau_{N_g})} \tr{\tau_{N_g} U_{N_g}(\theta_{N_g}) s_{N_g-1} U_{N_g}^\dagger(\theta_{N_g})} \\
  & \qquad \quad \tr{s_{N_g-1} \mathcal{N}_{N_g-1}\circ U_{N_g-1}(\theta_{N_g-1}) \circ \cdots \circ \mathcal{N}_1\circ U_1(\theta_1)(\rho)}\\
  \vdots\\
  =& \sum_{\tau_{N_g},\tau_{{N_g}-1},\cdots,\tau_1} \sum_{s_{N_g},s_{{N_g}-1},\cdots,s_0}  \tr{O s_{N_g}} \tr{s_0 \rho} \prod_{i=1}^{{N_g}} \tr{\tau_i U_i(\theta_i) s_{i-1} U_i^\dagger(\theta_i)} \tr{s_i \mathcal{N}_i(\tau_i)} \\
  =&\sum_{\vec{s}} \sum_{\vec{\tau}} \widetilde{f}^{(\vec{\tau})}(\vec{s},\bm{\theta},O,\rho) \\
  =&\sum_{\vec{s}} \widetilde{f}(\vec{s},\bm{\theta},O,\rho),
  \end{aligned}
\end{equation}
where $\vec{\tau}=(\tau_1,\cdots,\tau_{N_g}), \quad \tau_i\in \bm{P}_n$ and $\vec{s}=(s_0,\cdots,s_{N_g}), \quad s_i\in \bm{P}_n$ are Pauli paths.
$\widetilde{f}(\vec{s},\bm{\theta},O,\rho)$ denotes the contribution of a Pauli path $\vec{s}$:
\begin{equation}\label{eq:contribution_function_noise}
  \widetilde{f}(\vec{s},\bm{\theta},O,\rho) = \sum_{\vec{\tau}} \widetilde{f}^{(\vec{\tau})}(\vec{s},\bm{\theta},O,\rho),
\end{equation}
and $\widetilde{f}^{(\vec{\tau})}(\vec{s},\bm{\theta},O,\rho)$ is the contribution of $\vec{\tau}$ in $\widetilde{f}^{(\vec{\tau})}(\vec{s},\bm{\theta},O,\rho)$:
\begin{equation}\label{eq:contribution_function_noise_tau}
   \widetilde{f}^{(\vec{\tau})}(\vec{s},\bm{\theta},O,\rho)= \tr{O s_{N_g}} \tr{s_0 \rho} \prod_{i=1}^{{N_g}} \tr{\tau_i U_i(\theta_i) s_{i-1} U_i^\dagger(\theta_i)} \tr{s_i \mathcal{N}_i(\tau_i)}
\end{equation}

\subsection{Pauli noise}
For the case of Pauli type noises, which is a common type of noise in quantum circuits and can be described by a quantum channel $\mathcal{N}$:
\begin{equation}
  \mathcal{N}(\rho)=\sum_{i} p_i \sigma_i \rho \sigma_i^\dagger,
\end{equation}
where $\sigma_i$ is a Pauli operator, $p_i$ is the corresponding probability and $\sum_{i} p_i=1$.

Because of the (anti-)commuting property of the Pauli operators, the Pauli noise channel $\mathcal{N}$ acting on the Pauli operator $s$ can be expressed as:
\begin{equation}\label{eq:pauli_noise_channel}
  \mathcal{N}(s)=\sum_{i} p_i \sigma_i s \sigma_i^\dagger=\sum_{i} (-1)^{\mathbf{1}_{ac}(s,\sigma_i)} p_i s,
\end{equation}
where $\mathbf{1}_{ac}(s,\sigma_i)$ is the indicator function that equals to $1$ if $s$ and $\sigma_i$ anti-commute, otherwise it equals to $0$.
Thus there is $\mathcal{N}(s)\propto s$ for Pauli noise channel $\mathcal{N}$, and because of $\sum_i p_i=1$, we have $\abs{\sum_{i} (-1)^{\mathbf{1}_{ac}(s,\sigma_i)} p_i}\leq 1$.
Thus the Pauli noise channel $\mathcal{N}$ satisfies the PCS1 condition in Definition~\ref{def:PCS1}.

Because of $\mathcal{N}(s)\propto s$ for Pauli noise, we have $\tr{s_j \mathcal{N}(\tau)}\neq 0$ only when $\tau=s_j$. 
Therefore $\mathbb{P}(\tau|\vec{s}) = \delta_{\tau,s_j}$ and we have:
\begin{equation}\label{eq:pauli_noise_effect_factor}
  \widetilde{f}(\vec{s},\bm{\theta},O,\rho)= \tr{s_j \mathcal{N}(s_j)} f^j_{s_j}(\vec{s},\bm{\theta},O,\rho)=g_{\mathcal{N}}(\vec{s})f(\vec{s},\bm{\theta},O,\rho),
\end{equation}
where $g_{\mathcal{N}}(\vec{s})$ is the noise effect factor of $\mathcal{N}$ on the Pauli path $\vec{s}$ with $\abs{g_{\mathcal{N}}(\vec{s})}\leq 1$, in this case $g_{\mathcal{N}}(\vec{s})=\tr{s_i \mathcal{N}(s_i)}$. 

If except for $\mathcal{N}$ in the $i$-layer, there is other noise channel $\mathcal{N'}$ in the $j$-layer, similarly, we can express the noisy expectation value $\langle \widetilde{O} \rangle$ as:
\begin{equation}
  \begin{aligned}
    \langle \widetilde{O} \rangle &= \sum_{s_{N_g}} \tr{\widetilde{O} s_{N_g}} \tr{s_{N_g} \widetilde{\mathcal{C}}(\bm{\theta})(\rho)}\\
    &= \sum_{s_{N_g},s_{{N_g}-1},\cdots,s_0} \tr{s_i \mathcal{N}(s_i)}\tr{s_j \mathcal{N}(s_j)} \tr{O s_{N_g}} \tr{s_0 \rho} \prod_{i=1}^{{N_g}} \tr{s_i U_i(\theta_i) s_{i-1} U_i^\dagger(\theta_i)} \\
    &=\sum_{\vec{s}} g_{\mathcal{N}}g_{\mathcal{N'}}(\vec{s})f(\vec{s},\bm{\theta},O,\rho)\\
    &=\sum_{\vec{s}} \widetilde{f}(\vec{s},\bm{\theta},O,\rho).
  \end{aligned}
\end{equation}
Thus we can see that when there are multiple noise channels in the quantum circuit, the noise effect factor is the product of the noise effect factors of each noise channel, i.e.:
\begin{equation}\label{eq:noise_effect_factor_product}
  g_{\mathcal{N}\cup \mathcal{N'}}(\vec{s})=g_{\mathcal{N}}(\vec{s})g_{\mathcal{N'}}(\vec{s}).
\end{equation}

\subsection{Amplitude damping noise}
For amplitude damping channel, this is other common type of noise but is non-unital, which can be described by a quantum channel $\mathcal{N}^{(amp)}$:
\begin{equation}
  \mathcal{N}^{(amp)}(\rho)= K_0 \rho K_0^\dagger + K_1 \rho K_1^\dagger,
\end{equation}
where $K_0=\begin{pmatrix}
  1 & 0 \\
  0 & \sqrt{1-\gamma}
\end{pmatrix}$ and $K_1=\begin{pmatrix}
  0 & \sqrt{\gamma} \\
  0 & 0
\end{pmatrix}$, $\gamma$ is the probability of the amplitude damping channel.
The amplitude damping channel $\mathcal{N}^{(amp)}$ acting on a density matrix $\rho=\begin{pmatrix}
  \rho_{00} & \rho_{01} \\
  \rho_{10} & \rho_{11}
\end{pmatrix}$ can be expressed as:
\begin{equation}
  \mathcal{N}^{(amp)}(\rho)= \begin{pmatrix}
    \rho_{00}+\gamma\rho_{11} & \sqrt{1-\gamma}\rho_{01} \\
    \sqrt{1-\gamma}\rho_{10} & (1-\gamma)\rho_{11}
  \end{pmatrix} .
\end{equation}

The Pauli transform matrix $\left(\mathcal{S}_{\mathcal{N}^{(amp)}}\right)_{i,j}=\tr{\mathcal{N}^{(amp)}(\sigma_i) \sigma_j}$ of the amplitude damping channel $\mathcal{N}$ can be expressed as:
\begin{equation}\label{eq:amplitude_damping_ptm}
  \mathcal{S}_{\mathcal{N}^{(amp)}}=\begin{pmatrix}
    1 & 0 & 0 & \gamma \\
    0 & \sqrt{1-\gamma} & 0 & 0 \\
    0 & 0 & \sqrt{1-\gamma} & 0 \\
    0 & 0 & 0 & (1-\gamma)
  \end{pmatrix},
\end{equation}
where $\sigma_0=\mathbb{I}$, $\sigma_1=X$, $\sigma_2=Y$, and $\sigma_3=Z$.
The amplitude damping channel $\mathcal{N}^{(amp)}$ satisfies the PCS1 condition in Definition~\ref{def:PCS1} because the 1-norm of each column is no more than $1$.

\subsection{Thermal Relaxation process}
The thermal relaxation process is a common type of noise in quantum circuits, which can be described by a quantum channel $\mathcal{N}^{(thermal)}$.
The thermal relaxation process $\mathcal{N}^{(thermal)}$ acting on a density matrix $\rho=\begin{pmatrix}
  1-\rho_{11} & \rho_{01} \\
  \bar{\rho}_{01} & \rho_{11}
\end{pmatrix}$ can be expressed as:
\begin{equation}
  \mathcal{N}^{(thermal)}(\rho)= \begin{bmatrix}
    1-\rho_{11}e^{-\frac{t}{T_1}} & \rho_{01}e^{-\frac{t}{2T_1}-\frac{t}{T_\phi}} \\
    \bar{\rho}_{01}e^{-\frac{t}{2T_1}-\frac{t}{T_\phi}} & \rho_{11}e^{-\frac{t}{T_1}}
  \end{bmatrix},
\end{equation}
where $T_1$ and $T_2$ denote the energy relaxation time and the dephasing time, respectively, satisfying $T_2\leq 2T_1$.
The parameter $t\geq 0$ represents the duration of the thermal relaxation process, and there is a relationship $\frac{1}{T_\phi}=\frac{1}{T_2}-\frac{1}{2T_1}$.

The Kraus operators of the thermal relaxation process $\mathcal{N}^{(thermal)}$ can be expressed as:
\begin{equation}
  \mathcal{N}^{(thermal)}(\rho)= K_0 \rho K_0^\dagger + K_1 \rho K_1^\dagger + K_2 \rho K_2^\dagger,
\end{equation}
where $K_0=\begin{pmatrix}
  1 & 0 \\
  0 & \sqrt{1-\lambda-\gamma}
\end{pmatrix}$, $K_1=\begin{pmatrix}
  0 & \sqrt{\gamma} \\
  0 & 0
\end{pmatrix}$ and $K_2=\begin{pmatrix}
  0 & 0 \\
  0 & \sqrt{\lambda}
\end{pmatrix}$.

The thermal parameters $\gamma$ and $\lambda$ satisfy the following relationship:
\begin{equation}
  \gamma=1-e^{-\frac{t}{T_1}}, \quad \lambda=e^{-\frac{t}{T_1}}-e^{-\frac{t}{2T_2}}.
\end{equation}

The Pauli transform matrix $\left(\mathcal{S}_{\mathcal{N}^{(thermal)}}\right)_{i,j}=\tr{\mathcal{N}^{(thermal)}(\sigma_i) \sigma_j}$ of the thermal relaxation process $\mathcal{N}^{(thermal)}$ can be expressed as:
\begin{equation}
  \mathcal{S}_{\mathcal{N}^{(thermal)}}=\begin{pmatrix}
    1 & 0 & 0 & \gamma \\
    0 & \sqrt{1-\lambda-\gamma} & 0 & 0 \\
    0 & 0 & \sqrt{1-\lambda-\gamma} & 0 \\
    0 & 0 & 0 & 1-\gamma
  \end{pmatrix},
\end{equation}
where $\sigma_0=\mathbb{I}$, $\sigma_1=X$, $\sigma_2=Y$, and $\sigma_3=Z$.
The thermal relaxation process $\mathcal{N}^{(thermal)}$ satisfies the PCS1 condition in Definition~\ref{def:PCS1} because the 1-norm of each column is no more than $1$.

\subsection{Measurement-based feedback control}
Except for the noise channels, there are other quantum channels that usually appear in quantum circuits, this includes Measurement-based feedback control.
The measurement-based feedback control is usually implemented by measuring the qubits and applying a unitary operation on the remaining qubits based on the measurement results, as shown in Fig~\ref{fig:measurement_based_feedback_control}.

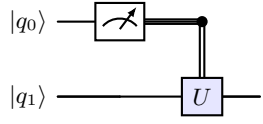
\begin{figure}[htbp]
  \begin{quantikz}
    % Quantum circuit lines
    \lstick{$\ket{q_0}$}  & \meter{} & \cwbend{1} \\
    \lstick{$\ket{q_1}$}  & \qw & \gate[style={fill=blue!10}]{U} & \qw 
  \end{quantikz}
  \caption{Measurement-based feedback control: A classical conditional operation triggered by measurement outcome, if the measurement result is $0$, the operation is identity, otherwise the operation is $U$ gate.}\label{fig:measurement_based_feedback_control}
\end{figure}

In our discussion, we assume that the measurement-based feedback control appears in the quantum circuit, and the classical conditional operation is a Pauli operator, i.e., if the measurement result is $0$, the operation is identity, otherwise the operation is a Pauli operator. 
Without loss of generality, the measured qubits are reset to $\frac{\mathbb{I}}{2}$.
For this case, the measurement-based feedback control can be expressed as:
\begin{equation}
  \mathcal{E}(\rho)=\frac{\mathbb{I}}{2} \otimes \left[\bra{0}\otimes \mathbb{I} \cdot \rho \cdot \ket{0}\otimes \mathbb{I} + \bra{1}\otimes P \cdot \rho \cdot \ket{1}\otimes P \right],
\end{equation}
where $\mathbb{I}$ is the identity operator on the remaining qubits and $P$ is a Pauli operator. For the measurement-based feedback control Pauli gate, we have the following lemma.
\begin{lemma}\label{lem:measurement_based_feedback_control}
  For the measurement-based feedback control, if the classical conditional operation is a Pauli operator. Then the the measurement-based feedback control process is a PCS1 quantum channel.
\end{lemma}
\begin{proof}
  For the measurement-based feedback control, the Pauli operator $P$ is applied to the remaining qubits based on the measurement result, and the process can be expressed as:
  \begin{equation}
    \mathcal{E}(\rho)=\frac{\mathbb{I}}{2} \otimes \left[\bra{0}\otimes \mathbb{I} \cdot \rho \cdot \ket{0}\otimes \mathbb{I} + \bra{1}\otimes P \cdot \rho \cdot \ket{1}\otimes P \right].
  \end{equation}
  If we show that the Pauli transform matrix $\left(\mathcal{S}_\mathcal{E}\right)_{i,j}=\tr{\mathcal{E}(\sigma_i) \sigma_j}$ has at most one non-zero entry in each column, i.e., for a given Pauli operator $s$, there is at most one Pauli operator $\sigma$ that satisfies the equation $\tr{\mathcal{E}(\sigma) s}\neq 0$.
  And if the non-zero entry is $\pm 1$, then the Pauli transform matrix $\left(\mathcal{S}_\mathcal{E}\right)_{i,j}$ satisfies the PCS1 condition in Definition~\ref{def:PCS1}.

  Because $\sigma$ is a Pauli operator, it can be decomposed as $\sigma = \sigma_1 \otimes \sigma_2$, where $\sigma_1$ acts on the measured qubit and $\sigma_2$ acts on the remaining qubits.
  On the other hand, the first qubit of $\mathcal{E}(\rho)$ is always $\frac{\mathbb{I}}{2}$, we can replace $s$ by $\frac{\mathbb{I}}{\sqrt{2}} \otimes s$.
  Therefore, we have:
  \begin{equation}
    \begin{aligned}
      \tr{\mathcal{E}(\sigma) \cdot \frac{\mathbb{I}}{\sqrt{2}} \otimes s} &= \frac{1}{\sqrt{2}} \left[\tr{\bra{0}\otimes \mathbb{I} \cdot \sigma \cdot \ket{0}\otimes \mathbb{I} s} + \tr{\bra{1}\otimes P \cdot \sigma \cdot \ket{1}\otimes P s}\right]\\
      &=\frac{1}{\sqrt{2}} \left[\tr{\bra{0}\sigma_1\ket{0}\otimes \sigma_2 s} + \tr{\bra{1}\sigma_1\ket{1}\otimes P\sigma_2 P \cdot s} \right]\\
      &=\begin{cases}
        \frac{1}{2}(\tr{\sigma_2 s}+\tr{P\sigma_2 P \cdot s}), & \sigma_1=\frac{\mathbb{I}}{\sqrt{2}} \\
        \frac{1}{2}(\tr{\sigma_2 s}-\tr{P\sigma_2 P \cdot s}), & \sigma_1=\frac{Z}{\sqrt{2}}\\
        0, & \sigma_1\neq \frac{\mathbb{I}}{\sqrt{2}},\frac{Z}{\sqrt{2}}. \\
      \end{cases}
    \end{aligned}
  \end{equation}
  For the case of $[ P,s]=0 $, we have $\tr{P\sigma_2 P \cdot s}=\tr{\sigma_2 s}$, it means that $\sigma=\mathbb{I}\otimes s$ otherwise $\tr{\mathcal{E}(\sigma) s}=0$. While for the case of $\{ P,s\}=0 $, we have $\tr{P\sigma_2 P \cdot s}=-\tr{\sigma_2 s}$, it means that $\sigma=Z\otimes s$ otherwise $\tr{\mathcal{E}(\sigma) s}=0$.
  Therefore, we can see that for a given Pauli operator $s$, there is at most one Pauli operator $\sigma$ that satisfies the equation $\tr{\mathcal{E}(\sigma) s}\neq 0$, and the non-zero entry is $\pm 1$. 
  Therefore, the Pauli transform matrix $\left(\mathcal{S}_\mathcal{E}\right)_{i,j}$ satisfies the PCS1 condition in Definition~\ref{def:PCS1}.
\end{proof}

\section{2-fold channel of circuits}\label{sec:2-design}
Let $R_P(\theta)=\exp{-i \frac{\theta}{2} P}=\cos(\frac{\theta}{2})\mathbb{I}-i\sin(\frac{\theta}{2})P$, it is not hard to see that the set
$\mathcal{G}_{P} := {\{ R_P(\theta) \}}_{\theta \in [0, 2\pi]}$
forms a group, which is a subgroup of the qubits unitary group
$\mathbb{U}(2^n)$.
Similar to unitary $t$-design, here we consider the $t$-design over the group $\mathcal{G}_{P}$, which we call
the \emph{quantum rotation $t$-design}.

\begin{definition}
  A set of  unitary matrices ${\left\{A_i\right\}}_{i=1}^K$ is called a \emph{quantum rotation $t$-design} with respect to
  $R_P(\theta)$, if
  \begin{align}\label{def:two_design}
    \frac{1}{K} \sum_{i=1}^K {\left(A_i \otimes A_i^{\dagger} \right)}^{\otimes t} =
    \frac{1}{2\pi} \int_{0}^{2\pi} {\left(R_P(\theta)
    \otimes R_P(-\theta) \right)}^{\otimes t} d\theta.
  \end{align}
\end{definition}

We now prove that the following gate set forms a quantum rotation $2$-design.
\begin{theorem}\label{thm:two_design}
  ${\left\{R_P(\theta)\right\}}_{\theta=0, \pi/2, \pi, 3\pi/2}$ forms a
  quantum rotation $2$-design with respect to
  ${\{ R_P(\theta) \}}_{\theta \in [0, 2\pi]}$.
\end{theorem}

\begin{proof}
  Utilizing the relations
  \begin{equation}
    \begin{aligned}
    &\frac{1}{2 \pi} \int_0^{2 \pi} \cos ^4 \frac{\theta}{2} {\rm d}\theta
      = \frac{1}{2 \pi} \int_0^{2 \pi} \sin ^4 \frac{\theta}{2} {\rm d}\theta
      = \frac{3}{8},\\
    & \frac{1}{2 \pi} \int_0^{2 \pi} \cos \frac{\theta}{2}
      \sin ^3 \frac{\theta}{2} {\rm d}\theta
      = \frac{1}{2 \pi} \int_0^{2 \pi} \cos ^3 \frac{\theta}{2}
      \sin \frac{\theta}{2} {\rm d}\theta=0, \\
    & \frac{1}{2 \pi} \int_0^{2 \pi} \cos ^2 \frac{\theta}{2}
      \sin ^2 \frac{\theta}{2} {\rm d}\theta=\frac{1}{8},
    \end{aligned}
  \end{equation}
  we have
  \begin{equation}
    \begin{aligned}
      & \frac{1}{2\pi}\int_{0}^{2\pi} {R_P(\theta)}^{\otimes 2} \otimes
      {R_P(-\theta)}^{\otimes 2} {\rm d}\theta \\
      & = \sum_{i_1,\ldots, i_4=0}^1 \frac{1}{2 \pi} \int_0^{2 \pi} i^{-i_1-i_2+i_3+i_4}
        {\Bigl(\cos \frac{\theta}{2}\Bigr)}^{4-\sum_j i_j} {\Bigl(\sin \frac{\theta}{2}\Bigr)}^{\sum_j i_j}
      \times \biggl(\bigotimes_{j=1}^{4} P^{i_{j}} \biggr)\,
      {\rm d}\theta\\
      & = \frac{3}{8} I^{\otimes 4} + \frac{3}{8} P^{\otimes 4} +
        \frac{1}{8} \sum_{\substack{i_1, \ldots, i_4=0 \\
      i_1+\cdots+i_4=2}}^1 i^{-i_1-i_2+i_3+i_4}
      \bigotimes_{j=1}^{4} P^{i_{j}}.
    \end{aligned}
  \end{equation}

  Meanwhile, it can be verified that
  \begin{equation}
    \begin{aligned}
      & \frac{1}{4} \sum_{k=0}^3 {R_P \Bigl( \frac{k\pi}{2}
      \Bigr)}^{\otimes 2} \otimes {R_P \Bigl(\frac{-k\pi}{2}
      \Bigr)}^{\otimes 2} \\
      & = \frac{1}{4} \left(1+\frac{1}{4}+\frac{1}{4}\right)I^{\otimes 4}+
        \frac{1}{4} \left(1+\frac{1}{4}+\frac{1}{4}\right)P^{\otimes 4}+ \frac{1}{4} \left(\frac{1}{4}+\frac{1}{4}\right)
     \sum_{\substack{i_1, \ldots, i_4=0 \\
      i_1+\cdots+i_4=2}}^1 i^{-i_1-i_2+i_3+i_4}
      \bigotimes_{j=1}^{4} P^{i_{j}}\\
      & = \frac{3}{8} I^{\otimes 4} + \frac{3}{8} P^{\otimes 4} +
        \frac{1}{8} \sum_{\substack{i_1, \ldots, i_4=0 \\
      i_1+\cdots+i_4=2}}^1 i^{-i_1-i_2+i_3+i_4}
      \bigotimes_{j=1}^{4} P^{i_{j}},
    \end{aligned}
  \end{equation}
  which concludes the proof.
\end{proof}

Thus for arbitrary operators $A, B, C, D$, we have the following corollary:
\begin{corollary}\label{cor:two_design}
  For any $n$-qubit operators $A, B, C, D$, the following equation holds:
  \begin{equation}\label{eq:two_design_cor}
    \mathbb{E}_{\theta} \tr{A R_P(\theta) B R_P(-\theta)}\tr{C R_P(\theta) D R_P(-\theta)}=
    \frac{1}{4}\sum_{\theta \in \{0,\frac{\pi}{2},\pi,\frac{3\pi}{2}\}}\tr{A R_P(\theta) B R_P(-\theta)}\tr{C R_P(\theta) D R_P(-\theta)}.
  \end{equation}
  Moreover, we have the following equation:
  \begin{equation}\label{eq:two_design_cor_2}
    \mathbb{E}_{\theta} R_P(\theta)^{\otimes 2}(\cdot)\left(R_P(\theta)^{\dagger}\right)^{\otimes 2} = \frac{1}{4}\sum_{\theta \in \{0,\frac{\pi}{2},\pi,\frac{3\pi}{2}\}} R_P(\theta)^{\otimes 2}(\cdot)\left(R_P(\theta)^{\dagger}\right)^{\otimes 2}.
  \end{equation}
\end{corollary}
\begin{proof}
  The proof is straightforward by using the definition of the quantum rotation $2$-design in Eq.~\eqref{def:two_design} and the fact of Thm.~\ref{thm:two_design}, we have:
  \begin{equation}
    \mathbb{E}_{\theta} R_P(\theta) \otimes R_P(-\theta) \otimes R_P(\theta) \otimes R_P(-\theta)=\frac{1}{4}\sum_{\theta \in \{0,\frac{\pi}{2},\pi,\frac{3\pi}{2}\}}R_P(\theta) \otimes R_P(-\theta) \otimes R_P(\theta) \otimes R_P(-\theta).
  \end{equation}
  The left-hand side of Eq.~\eqref{eq:two_design_cor} can be expressed as:
  \begin{equation}
    \begin{aligned}
      &\mathbb{E}_{\theta} \tr{A R_P(\theta) B R_P(-\theta)}\tr{C R_P(\theta) D R_P(-\theta)}\\
      =& \mathbb{E}_{\theta} \left(\sum_{i,j} \bra{i}A R_P(\theta) B \ket{j} \bra{j} R_P(-\theta)\ket{i} \right)
      \left(\sum_{k,l} \bra{k}C R_P(\theta) D\ket{l}\bra{l} R_P(-\theta)\ket{k}\right)\\
      =& \mathbb{E}_{\theta} \left(\sum_{i,j} \bra{i}\otimes\bra{j} \cdot (A R_P(\theta) B) \otimes R_P(-\theta)\cdot \ket{j}\otimes\ket{i} \right)
      \left(\sum_{k,l} \bra{k}\otimes\bra{l} \cdot(C R_P(\theta) D) \otimes R_P(-\theta) \cdot\ket{l}\otimes\ket{k}\right)\\
      =& \mathbb{E}_{\theta} \left(\sum_{i,j,k,l} \bra{i}\otimes\bra{j}\otimes\bra{k}\otimes\bra{l} \cdot (A R_P(\theta) B) \otimes R_P(-\theta) \otimes (C R_P(\theta) D) \otimes R_P(-\theta) \cdot\ket{j}\otimes\ket{i}\otimes\ket{l}\otimes\ket{k}\right)\\
      =& \mathbb{E}_{\theta} \left(\sum_{i,j,k,l} \bra{i}A\otimes \bra{j}\otimes\bra{k} C\otimes\bra{l} \cdot R_P(\theta) \otimes R_P(-\theta) \otimes R_P(\theta)  \otimes R_P(-\theta) \cdot B\ket{j}\otimes\ket{i}\otimes D\ket{l}\otimes\ket{k}\right)\\
      =& \sum_{i,j,k,l} \bra{i}A\otimes \bra{j}\otimes\bra{k}C\otimes\bra{l} \cdot \mathbb{E}_{\theta} \left(   R_P(\theta) \otimes R_P(-\theta) \otimes R_P(\theta)  \otimes R_P(-\theta)\right) \cdot B\ket{j}\otimes\ket{i}\otimes D\ket{l}\otimes\ket{k}\\
      =& \sum_{i,j,k,l} \bra{i}A\otimes \bra{i}\otimes\bra{k} \cdot \frac{1}{4}\sum_{\theta \in \{0,\frac{\pi}{2},\pi,\frac{3\pi}{2}\}} \left(  C\otimes\bra{k} R_P(\theta) \otimes R_P(-\theta) \otimes R_P(\theta)  \otimes R_P(-\theta)\right) \cdot B\ket{j}\otimes\ket{i}\otimes D\ket{l}\otimes\ket{k}\\
      =&\frac{1}{4}\sum_{\theta \in \{0,\frac{\pi}{2},\pi,\frac{3\pi}{2}\}}\left(\sum_{i,j} \bra{i}A R_P(\theta) B \ket{j} \bra{j} R_P(-\theta)\ket{i} \right)
      \left(\sum_{k,l} \bra{k}C R_P(\theta) D\ket{l}\bra{l} R_P(-\theta)\ket{k}\right)\\
      =&\frac{1}{4}\sum_{\theta \in \{0,\frac{\pi}{2},\pi,\frac{3\pi}{2}\}}\tr{A R_P(\theta) B R_P(-\theta)}\tr{C R_P(\theta) D R_P(-\theta)}.
    \end{aligned}
  \end{equation}

  For the Eq.~\eqref{eq:two_design_cor_2}, we have for any $i,j,k,l,a,b,c,d$:
  \begin{equation}
    \begin{aligned}
      &\bra{i}\otimes \bra{j} \left[\mathbb{E}_{\theta} R_P(\theta)^{\otimes 2}(\ket{a}\bra{b} \otimes \ket{c} \bra{d} )\left(R_P(\theta)^{\dagger}\right)^{\otimes 2} \right] \ket{k}\otimes\ket{l}\\
      =& \mathbb{E}_{\theta} \left[ \bra{i} R_P(\theta) \ket{a} \times \bra{j} R_P(-\theta) \ket{c} \times \bra{b} R_P(\theta) \ket{k} \times \bra{d} R_P(-\theta) \ket{l} \right]\\
      =& \mathbb{E}_{\theta} \left[ (\bra{i}\otimes \bra{j} \otimes \bra{b} \otimes \bra{d}) \left(R_P(\theta) \otimes R_P(-\theta) \otimes R_P(\theta) \otimes R_P(-\theta)  \right) (\ket{a} \otimes \ket{c} \otimes \ket{k} \otimes \ket{l}) \right]\\
      =& \frac{1}{4}\sum_{\theta \in \{0,\frac{\pi}{2},\pi,\frac{3\pi}{2}\}} \left[ (\bra{i}\otimes \bra{j} \otimes \bra{b} \otimes \bra{d}) \left(R_P(\theta) \otimes R_P(-\theta) \otimes R_P(\theta) \otimes R_P(-\theta)  \right) (\ket{a} \otimes \ket{c} \otimes \ket{k} \otimes \ket{l}) \right]\\
      =& \frac{1}{4}\sum_{\theta \in \{0,\frac{\pi}{2},\pi,\frac{3\pi}{2}\}} \left[ \bra{i} R_P(\theta) \ket{a} \times \bra{j} R_P(-\theta) \ket{c} \times \bra{b} R_P(\theta) \ket{k} \times \bra{d} R_P(-\theta) \ket{l} \right]\\
      =& \bra{i}\otimes \bra{j} \left[\frac{1}{4}\sum_{\theta \in \{0,\frac{\pi}{2},\pi,\frac{3\pi}{2}\}} R_P(\theta)^{\otimes 2}(\ket{a}\bra{b} \otimes \ket{c} \bra{d} )\left(R_P(\theta)^{\dagger}\right)^{\otimes 2} \right] \ket{k}\otimes\ket{l}.
    \end{aligned}
  \end{equation}
  By $(\ket{a}\bra{b} \otimes \ket{c} \bra{d} )$ and $\bra{i}\otimes \bra{j}, \ket{k}\otimes\ket{l}$ being arbitrary, we have the Eq.~\eqref{eq:two_design_cor_2}.
\end{proof}

By iteratively applying the Eq.~\eqref{eq:two_design_cor_2} in circuit $\mathcal{C}(\bm{\theta})$, we have the following lemma for the 2-fold channel:
\begin{lemma}
  The $2$-fold channel of $\mathcal{C}(\bm{\theta})$ equals to 
  \begin{small}
  \begin{equation}
      \mathcal{E}_2(\mathcal{C}(\bm{\theta})) = \frac{1}{4^{N_g}}\sum_{\bm{\theta} \in \{0,\frac{\pi}{2},\pi,\frac{3\pi}{2}\}^{N_g}} \mathcal{C}(\bm{\theta})^{\otimes 2}(\cdot)\left(\mathcal{C}(\bm{\theta})^{\dagger}\right)^{\otimes 2}.
    \end{equation}
  \end{small}
\end{lemma}
\begin{proof}
  The proof is straightforward by using Corollary~\ref{cor:two_design}, we have:
  \begin{equation}
    \begin{aligned}
      &\mathcal{E}_2(\mathcal{C}(\bm{\theta}))\\
      =& \mathbb{E}_{\bm{\theta}} \mathcal{C}(\bm{\theta})^{\otimes 2}(\cdot)\left(\mathcal{C}(\bm{\theta})^{\dagger}\right)^{\otimes 2}\\
      =& \mathbb{E}_{\bm{\theta}} \left[ \left( U_{N_g}(\theta_{N_g})  \cdots {U}_1(\theta_1) \right)^{\otimes 2}(\cdot)\left( {U}_1(\theta_1)^\dagger \cdots U_{N_g}(\theta_{N_g})^\dagger \right)^{\otimes 2}\right]\\
      =& \mathbb{E}_{\bm{\theta}} \left[ \left( \exp{-i \frac{\theta_{N_g}}{2} P_{N_g}}C_{N_g}  \cdots {U}_1(\theta_1) \right)^{\otimes 2}(\cdot)\left( {U}_1(\theta_1)^\dagger \cdots C_{N_g}^\dagger \exp{i \frac{\theta_{N_g}}{2} P_{N_g}} \right)^{\otimes 2}\right]\\
      =& \mathbb{E}_{\bm{\theta}} \left[ \left( R_{P_{N_g}}(\theta_{N_g}) C_{N_g}  \cdots {U}_1(\theta_1) \right)^{\otimes 2}(\cdot)\left( {U}_1(\theta_1)^\dagger \cdots C_{N_g}^\dagger R_{P_{N_g}}(\theta_{N_g})^\dagger \right)^{\otimes 2}\right]\\
      =& \mathbb{E}_{\theta_{N_g}} R_{P_{N_g}}(\theta_{N_g})^{\otimes 2} \left[\mathbb{E}_{\bm{\theta} \setminus \theta_{N_g}} \left[ \left( C_{N_g} \cdots {U}_1(\theta_1) \right)^{\otimes 2}(\cdot)\left( {U}_1(\theta_1)^\dagger \cdots C_{N_g}^\dagger \right)^{\otimes 2}\right]\right]  \left( R_{P_{N_g}}(\theta_{N_g})^\dagger \right)^{\otimes 2}\\
      =& \frac{1}{4}\sum_{\theta_{N_g} \in \{0,\frac{\pi}{2},\pi,\frac{3\pi}{2}\}} R_{P_{N_g}}(\theta_{N_g})^{\otimes 2} \left[\mathbb{E}_{\bm{\theta} \setminus \theta_{N_g}} \left[ \left( C_{N_g} \cdots {U}_1(\theta_1) \right)^{\otimes 2}(\cdot)\left( {U}_1(\theta_1)^\dagger \cdots C_{N_g}^\dagger \right)^{\otimes 2}\right]\right]  \left( R_{P_{N_g}}(\theta_{N_g})^\dagger \right)^{\otimes 2}\\
      & \vdots \\
      =& \frac{1}{4^{N_g}}\sum_{\bm{\theta} \in \{0,\frac{\pi}{2},\pi,\frac{3\pi}{2}\}^{N_g}} \left[ \left( U_{N_g}(\theta_{N_g})  \cdots {U}_1(\theta_1) \right)^{\otimes 2}(\cdot)\left( {U}_1(\theta_1)^\dagger \cdots U_{N_g}(\theta_{N_g})^\dagger \right)^{\otimes 2}\right]\\
      =& \frac{1}{4^{N_g}}\sum_{\bm{\theta} \in \{0,\frac{\pi}{2},\pi,\frac{3\pi}{2}\}^{N_g}} \mathcal{C}(\bm{\theta})^{\otimes 2}(\cdot)\left(\mathcal{C}(\bm{\theta})^{\dagger}\right)^{\otimes 2},
    \end{aligned}
  \end{equation}
  where the sixth equality is followed by the Eq.~\eqref{eq:two_design_cor_2}. This finishes the proof.
\end{proof}

For the contribution of Pauli path $f(\vec{s},\bm{\theta},O,\rho)$, defined Eq.~\eqref{eq:pauli_path_integral_noiseless}, and the noisy contribution $\widetilde{f}(\vec{s},\bm{\theta},O,\rho)$, expressed in Eq.~\eqref{eq:pauli_path_integral_noisy}, we have the following lemma:
\begin{lemma}\label{lem:cross_term}
Let $\vec{s}$ and $\vec{s'}$ be two Pauli paths. Then, their contributions satisfy the following equation:
  \begin{align}
      \mathbb{E}_{\theta} f(\vec{s},\bm{\theta},O,\rho)f(\vec{s'},\bm{\theta},O,\rho)= \frac{1}{4^{N_g}}\sum_{\bm{\theta} \in \{0,\frac{\pi}{2},\pi,\frac{3\pi}{2}\}^{N_g}}f(\vec{s},\bm{\theta},O,\rho)f(\vec{s'},\bm{\theta},O,\rho), \\
    \mathbb{E}_{\theta} \widetilde{f}(\vec{s},\bm{\theta},O,\rho)\widetilde{f}(\vec{s'},\bm{\theta},O,\rho)= \frac{1}{4^{N_g}}\sum_{\bm{\theta} \in \{0,\frac{\pi}{2},\pi,\frac{3\pi}{2}\}^{N_g}}\widetilde{f}(\vec{s},\bm{\theta},O,\rho)\widetilde{f}(\vec{s'},\bm{\theta},O,\rho), \\
    \mathbb{E}_{\theta} f(\vec{s},\bm{\theta},O,\rho)\widetilde{f}(\vec{s'},\bm{\theta},O,\rho)= \frac{1}{4^{N_g}}\sum_{\bm{\theta} \in \{0,\frac{\pi}{2},\pi,\frac{3\pi}{2}\}^{N_g}}f(\vec{s},\bm{\theta},O,\rho)\widetilde{f}(\vec{s'},\bm{\theta},O,\rho),
  \end{align}
  where $\bm{\theta}=\{\theta_1,\ldots,\theta_{N_g}\}$ is the set of rotation angles and ${N_g}$ is number of rotation gates.
\end{lemma}
\begin{proof}
  The proof is straightforward by using Corollary~\ref{cor:two_design}, we have:
  \begin{equation}
    \begin{aligned}
      & \mathbb{E}_{\bm{\theta}}f(\vec{s},\bm{\theta},O,\rho)f(\vec{s'},\bm{\theta},O,\rho)\\
      =& \tr{O s_{N_g}}\tr{O s'_{N_g}} \tr{s_0 \rho}\tr{s'_0 \rho} \prod_{i=1}^{{N_g}}\mathbb{E}_{\theta_i} \tr{s_i U_i(\theta_i) s_{i-1} U_i^\dagger(\theta_i)}\tr{s'_i U_i(\theta_i) s'_{i-1} U_i^\dagger(\theta_i)}
    \end{aligned}
  \end{equation}
  
  For terms $\mathbb{E}_{\theta_i} \tr{s_i U_i(\theta_i) s_{i-1} U_i^\dagger(\theta_i)}\tr{s'_i U_i(\theta_i) s'_{i-1} U_i^\dagger(\theta_i)}$, using Eq.~\eqref{eq:two_design_cor}, we have:
  \begin{equation}
    \mathbb{E}_{\theta_i} \tr{s_i U_i(\theta_i) s_{i-1} U_i^\dagger(\theta_i)}\tr{s'_i U_i(\theta_i) s'_{i-1} U_i^\dagger(\theta_i)}= \frac{1}{4}\sum_{\theta_i \in \{0,\frac{\pi}{2},\pi,\frac{3\pi}{2}\}} \tr{s_i U_i(\theta_i) s_{i-1} U_i^\dagger(\theta_i)}\tr{s'_i U_i(\theta_i) s'_{i-1} U_i^\dagger(\theta_i)}.
  \end{equation}
  Therefore, we have:
  \begin{equation}
    \begin{aligned}
      &\mathbb{E}_{\bm{\theta}}f(\vec{s},\bm{\theta},O,\rho)f(\vec{s'},\bm{\theta},O,\rho)\\
      =& \tr{O s_{N_g}}\tr{O s'_{N_g}} \tr{s_0 \rho}\tr{s'_0 \rho} \prod_{i=1}^{{N_g}}\frac{1}{4}\sum_{\theta_i \in \{0,\frac{\pi}{2},\pi,\frac{3\pi}{2}\}} \tr{s_i U_i(\theta_i) s_{i-1} U_i^\dagger(\theta_i)}\tr{s'_i U_i(\theta_i) s'_{i-1} U_i^\dagger(\theta_i)}\\
      =& \frac{1}{4^{N_g}}\sum_{\theta_1,\ldots,\theta_{N_g} \in \{0,\frac{\pi}{2},\pi,\frac{3\pi}{2}\}}f(\vec{s},\bm{\theta},O,\rho)f(\vec{s'},\bm{\theta},O,\rho).
    \end{aligned}
  \end{equation}

The expression for $\widetilde{f}(\vec{s},\bm{\theta},O,\rho)$ is analogous, with the only difference being that $\tr{s_i U_i(\theta_i) s_{i-1} U_i^\dagger(\theta_i)}$ is replaced by $\tr{\tau_i U_i(\theta_i) s_{i-1} U_i^\dagger(\theta_i)}$.
\end{proof}

\section{Noise robustness}\label{sec:noise}

\subsection{Estimation of noise effect}

To evaluate the noise effect on a quantum circuit structure, we can use the mean squared error (MSE) of the noisy expectation value $\langle \widetilde{O} \rangle$ with respect to the parameter $\bm{\theta}$ in the parameterized quantum circuit $\widetilde{\mathcal{C}}(\bm{\theta})$. The MSE is defined as:
\begin{equation}
  \mathrm{MSE}(\langle O \rangle) = \mathbb{E}_{\bm{\theta}} \left( \langle O \rangle_{\bm{\theta}} - \langle \widetilde{O} \rangle_{\bm{\theta}} \right)^2,
\end{equation}
where $\langle O \rangle_{\bm{\theta}}$ is the expectation value of the observable $O$ with respect to rotation angles $\bm{\theta}$ in noiseless case, and $\langle \widetilde{O} \rangle_{\bm{\theta}}$ is the expectation value of the observable $O$ with respect to rotation angles $\bm{\theta}$ in noisy case. 

Using Eq.~\eqref{eq:pauli_path_integral_noiseless} and Eq.~\eqref{eq:pauli_path_integral_noisy}, the MSE can be expressed as:
\begin{equation}\label{ap:eq:mse_noise_robustness_derivation}
  \begin{aligned}
    \mathrm{MSE}(\langle O \rangle) &= \mathbb{E}_{\bm{\theta}} \left( \langle O \rangle_{\bm{\theta}} - \langle \widetilde{O} \rangle_{\bm{\theta}} \right)^2\\
    &= \mathbb{E}_{\bm{\theta}} \left( \sum_{\vec{s}} f(\vec{s},\bm{\theta},O,\rho) - \widetilde{f}(\vec{s},\bm{\theta},O,\rho) \right)^2\\
    &= \mathbb{E}_{\bm{\theta}} \left( \sum_{\vec{s}} f(\vec{s},\bm{\theta},O,\rho) - \widetilde{f}(\vec{s},\bm{\theta},O,\rho) \right)\left( \sum_{\vec{s'}} f(\vec{s'},\bm{\theta},O,\rho) - \widetilde{f}(\vec{s'},\bm{\theta},O,\rho) \right)\\
    &= \mathbb{E}_{\bm{\theta}} \left( \sum_{\vec{s},\vec{s'}} f(\vec{s},\bm{\theta},O,\rho) f(\vec{s'},\bm{\theta},O,\rho) +  \widetilde{f}(\vec{s},\bm{\theta},O,\rho) \widetilde{f}(\vec{s'},\bm{\theta},O,\rho) - \widetilde{f}(\vec{s},\bm{\theta},O,\rho)f(\vec{s'},\bm{\theta},O,\rho) - f(\vec{s},\bm{\theta},O,\rho)\widetilde{f}(\vec{s'},\bm{\theta},O,\rho)\right)\\
    &= \frac{1}{4^{N_g}}\sum_{\bm{\theta} \in \{0,\frac{\pi}{2},\pi,\frac{3\pi}{2}\}^{N_g}} \left( \sum_{\vec{s}} f(\vec{s},\bm{\theta},O,\rho) - \widetilde{f}(\vec{s},\bm{\theta},O,\rho) \right)\left( \sum_{\vec{s'}} f(\vec{s'},\bm{\theta},O,\rho) - \widetilde{f}(\vec{s'},\bm{\theta},O,\rho) \right)\\
    &= \frac{1}{4^{N_g}}\sum_{\bm{\theta} \in \{0,\frac{\pi}{2},\pi,\frac{3\pi}{2}\}^{N_g}} \left( \langle O \rangle_{\bm{\theta}} - \langle \widetilde{O} \rangle_{\bm{\theta}} \right)^2\\
    &= \frac{1}{4^{N_g}}\sum_{\bm{\theta} \in \{0,\frac{\pi}{2},\pi,\frac{3\pi}{2}\}^{N_g}}\left( \sum_{\vec{s}} f(\vec{s},\bm{\theta},O,\rho) - \widetilde{f}(\vec{s},\bm{\theta},O,\rho) \right)^2,
  \end{aligned}
\end{equation}
where the fifth equality is obtained by using Lemma~\ref{lem:cross_term} and the expectation in the last equality is taken over $\bm{\theta}$ uniformly distributed in the set $\{0, \frac{\pi}{2}, \pi, \frac{3\pi}{2}\}^{N_g}$.

Thus, rewriting the MSE, we have:
\begin{equation}\label{ap:eq:mse_noise_robustness}
  \mathrm{MSE}(\langle O \rangle) = \mathbb{E}^*_{\bm{\theta}} h_{\mathrm{robustness}}(\bm{\theta}),
\end{equation}
where $\mathbb{E}^*_{\bm{\theta}}$ represents the expectation value that $\bm{\theta}$ takes values in $\{0, \frac{\pi}{2}, \pi, \frac{3\pi}{2}\}^{N_g}$ with uniform distribution, $h_{\mathrm{robustness}}(\bm{\theta}) = \left( \sum_{\vec{s}} f(\vec{s},\bm{\theta},O,\rho) - \widetilde{f}(\vec{s},\bm{\theta},O,\rho) \right)^2$ and its absolute value is bounded by $h_{\mathrm{robustness}}(\bm{\theta}) \leq 4\norm{O}_\infty^2$.

\subsection{Estimation of effect for variable noise}
\label{sec:varnoise}

In noisy quantum circuit architectures, to determine the locations most sensitive to a given level of noise, one may assess the impact of local noise on the circuit output by analyzing the gradient of the MSE of the noisy expectation value $\langle \widetilde{O} \rangle$ with respect to the noise strength. A higher gradient value implies that noise at the corresponding location exerts a greater influence on the circuit's output.

\subsubsection{Single-qubit depolarizing noise}
We first consider the case of single-qubit depolarizing noise.
For convenient, we assume that in the quantum circuit $\mathcal{C} = U_{N_g}(\theta_{N_g}) \cdots U_1(\theta_1)$, there is a single-qubit depolarizing noise $\mathcal{N}_i$ after each gate $U_i(\theta_i)$.
If there are multiple noises after a gate, we can add additional identity gates to separate them.
The single-qubit depolarizing noise $\mathcal{N}_{i}$ with strength $\lambda_{i}$ can be expressed as:
\begin{equation}
  \mathcal{N}_{i}(\rho) = (1-\lambda_{i})\rho + \lambda_{i}\tr{\rho} \frac{\mathbb{I}}{2},
\end{equation}
which can be viewed as a equal probability mixture of randomly applying the Pauli operators $X$, $Y$, and $Z$ with probability $\frac{\lambda_i}{4}$.

For $l$-th layer, the gradient of the MSE with respect to the noise strength $\lambda_{l}$ can be expressed as:
\begin{equation}
  \begin{aligned}
    \frac{\partial \mathrm{MSE}(\langle O \rangle)}{\partial \lambda_{l}} =& \frac{\partial}{\partial \lambda_{l}} \left[\frac{1}{4^{N_g}}\sum_{\bm{\theta} \in \{0,\frac{\pi}{2},\pi,\frac{3\pi}{2}\}^{N_g}} \left( \langle O \rangle_{\bm{\theta}} - \langle \widetilde{O} \rangle_{\bm{\theta}} \right)^2\right]\\
    =& -2 \frac{1}{4^{N_g}}\sum_{\bm{\theta} \in \{0,\frac{\pi}{2},\pi,\frac{3\pi}{2}\}^{N_g}} \left( \langle O \rangle_{\bm{\theta}} - \langle \widetilde{O} \rangle_{\bm{\theta}} \right) \frac{\partial}{\partial \lambda_{l}} \langle \widetilde{O} \rangle_{\bm{\theta}}\\
    =& -2 \frac{1}{4^{N_g}}\sum_{\bm{\theta} \in \{0,\frac{\pi}{2},\pi,\frac{3\pi}{2}\}^{N_g}} \left(  \sum_{\vec{s}} f(\vec{s},\bm{\theta},O,\rho) - \widetilde{f}(\vec{s},\bm{\theta},O,\rho) \right) \frac{\partial}{\partial \lambda_{l}} \sum_{\vec{s'}} \widetilde{f}(\vec{s'},\bm{\theta},O,\rho)\\
    =& -2 \frac{1}{4^{N_g}}\sum_{\bm{\theta} \in \{0,\frac{\pi}{2},\pi,\frac{3\pi}{2}\}^{N_g}} \left(  \sum_{\vec{s}} (1-g_{\mathcal{N}}(\vec{s}))f(\vec{s},\bm{\theta},O,\rho) \right) \frac{\partial}{\partial \lambda_{l}} \sum_{\vec{s'}} g_{\mathcal{N}}(\vec{s'})f(\vec{s'},\bm{\theta},O,\rho)\\
    =& -2 \frac{1}{4^{N_g}}\sum_{\bm{\theta} \in \{0,\frac{\pi}{2},\pi,\frac{3\pi}{2}\}^{N_g}}  \sum_{\vec{s},\vec{s'}} \left( \frac{\partial}{\partial \lambda_{l}} g_{\mathcal{N}}(\vec{s'}) \right)(1-g_{\mathcal{N}}(\vec{s}))f(\vec{s},\bm{\theta},O,\rho)  f(\vec{s'},\bm{\theta},O,\rho)\\
    =& -2 \sum_{\vec{s},\vec{s'}} \left( \frac{\partial}{\partial \lambda_{l}} g_{\mathcal{N}}(\vec{s'}) \right)(1-g_{\mathcal{N}}(\vec{s})) \frac{1}{4^{N_g}}\sum_{\bm{\theta} \in \{0,\frac{\pi}{2},\pi,\frac{3\pi}{2}\}^{N_g}} f(\vec{s},\bm{\theta},O,\rho)  f(\vec{s'},\bm{\theta},O,\rho)\\
    =& -2 \sum_{\vec{s},\vec{s'}} \left( \frac{\partial}{\partial \lambda_{l}} g_{\mathcal{N}}(\vec{s'}) \right)(1-g_{\mathcal{N}}(\vec{s})) \frac{1}{4^{N_g}}\sum_{\bm{\theta} \in \{0,\frac{\pi}{2},\pi,\frac{3\pi}{2}\}^{N_g}} f(\vec{s},\bm{\theta},O,\rho)  f(\vec{s'},\bm{\theta},O,\rho)\\
    =& -\frac{2}{4^{N_g}}\sum_{\bm{\theta} \in \{0,\frac{\pi}{2},\pi,\frac{3\pi}{2}\}^{N_g}} \sum_{\vec{s},\vec{s'}} \left( \frac{\partial}{\partial \lambda_{l}} g_{\mathcal{N}}(\vec{s'}) \right)(1-g_{\mathcal{N}}(\vec{s})) f(\vec{s},\bm{\theta},O,\rho)  f(\vec{s'},\bm{\theta},O,\rho),
  \end{aligned}
\end{equation}
where $g_{\mathcal{N}}(\vec{s})$ is the noise effect factor of all noises in the circuit, appearing in Eq.~\eqref{eq:pauli_noise_effect_factor}, and the seventh equality is obtained by using Lemma~\ref{lem:cross_term}.

By Eq.~\eqref{eq:noise_effect_factor_product}, we can express the noise effect factor $g_{\mathcal{N}}(\vec{s})$ as $g_{\mathcal{N}}(\vec{s})=\prod_{i} g_{\mathcal{N}_{i}}(\vec{s})$, where $g_{\mathcal{N}_{i}}(\vec{s})$ is the noise effect factor of the $i$-th layer. The noise effect factor $g_{\mathcal{N}_{i}}(\vec{s})$ can be expressed as:
\begin{equation}
  \begin{aligned}
    g_{\mathcal{N}_{i}}(\vec{s}) &= \tr{s_i|_{\mathcal{N}_{i}} \mathcal{N}_{i,j}(s_i|_{\mathcal{N}_{i}})}\\
    &= \tr{s_i|_{\mathcal{N}_{i}} \left( (1-\lambda_{i})s_i|_{\mathcal{N}_{i}} + \lambda_{i}\tr{s_i|_{\mathcal{N}_{i}}} \frac{\mathbb{I}}{2} \right)}\\
    &= 1-\lambda_{i} + \lambda_{i}\mathbf{1}_{(s_i|_{\mathcal{N}_{i}}=\mathbb{I}/\sqrt{2})}(\vec{s}),
  \end{aligned}
\end{equation}
where $s_i|_{\mathcal{N}_{i}}$ is the restriction of the Pauli operator $s_i$ to the subsystem acted on by the noise $\mathcal{N}_{i}$, and $\mathbf{1}_{(s_i|_{\mathcal{N}_{i}}=\mathbb{I}/\sqrt{2})}(\vec{s})$ is an indicator function that equals 1 if $s_i|_{\mathcal{N}_{i}}=\mathbb{I}/\sqrt{2}$ and 0 otherwise.
Thus, we can express the gradient of the MSE with respect to the noise strength $\lambda_{l}$ as:
\begin{equation}
  \begin{aligned}
    \frac{\partial \mathrm{MSE}(\langle O \rangle)}{\partial \lambda_{l}} =& -\frac{2}{4^{N_g}}\sum_{\bm{\theta} \in \{0,\frac{\pi}{2},\pi,\frac{3\pi}{2}\}^{N_g}} \sum_{\vec{s},\vec{s'}} \left( \frac{\partial}{\partial \lambda_{l}} g_{\mathcal{N}}(\vec{s'}) \right)(1-g_{\mathcal{N}}(\vec{s})) f(\vec{s},\bm{\theta},O,\rho)  f(\vec{s'},\bm{\theta},O,\rho)\\
    =& -\frac{2}{4^{N_g}}\sum_{\bm{\theta} \in \{0,\frac{\pi}{2},\pi,\frac{3\pi}{2}\}^{N_g}} \sum_{\vec{s},\vec{s'}} \left( \mathbf{1}_{(s_l|_{\mathcal{N}_{l}}=\mathbb{I}/\sqrt{2})}(\vec{s'})-1 \right)g_{\mathcal{N}\setminus \mathcal{N}_{l}}(\vec{s'})(1-g_{\mathcal{N}}(\vec{s})) f(\vec{s},\bm{\theta},O,\rho)  f(\vec{s'},\bm{\theta},O,\rho)\\
    =& -\frac{2}{4^{N_g}}\sum_{\bm{\theta} \in \{0,\frac{\pi}{2},\pi,\frac{3\pi}{2}\}^{N_g}} \left(\sum_{\vec{s}} (1-g_{\mathcal{N}}(\vec{s}))f(\vec{s},\bm{\theta},O,\rho) \right) \left( \sum_{\vec{s'}}  \left( \mathbf{1}_{(s_l|_{\mathcal{N}_{l}}=\mathbb{I}/\sqrt{2})}(\vec{s'})-1 \right)g_{\mathcal{N}\setminus \mathcal{N}_{l}}(\vec{s'})  f(\vec{s'},\bm{\theta},O,\rho)  \right),
  \end{aligned}
\end{equation}
where $g_{\mathcal{N}\setminus \mathcal{N}_{l}}(\vec{s})$ is the noise effect factor of all noises in the circuit except for the $l$-th layer.

Hence, we can rewrite the gradient of the MSE with respect to the noise strength $\lambda_{l}$ as:
\begin{equation}\label{ap:eq:noise_sensitivity}
  \abs{\frac{\partial \mathrm{MSE}(\langle O \rangle)}{\partial \lambda_{l}}} = \abs{\mathbb{E}^*_{\bm{\theta}} h_{\mathrm{sensitivity}}(\bm{\theta})},
\end{equation}
where $\mathbb{E}^*_{\bm{\theta}}$ represents the expectation value that $\bm{\theta}$ takes values in $\{0, \frac{\pi}{2}, \pi, \frac{3\pi}{2}\}^{N_g}$ with uniform distribution, and
\begin{equation}
  h_{\mathrm{sensitivity}}(\bm{\theta}) = -2 \left(\sum_{\vec{s}} (1-g_{\mathcal{N}}(\vec{s}))f(\vec{s},\bm{\theta},O,\rho) \right) \left( \sum_{\vec{s'}}  \left( \mathbf{1}_{(s_l|_{\mathcal{N}_{l}}=\mathbb{I}/\sqrt{2})}(\vec{s'})-1 \right)g_{\mathcal{N}\setminus \mathcal{N}_{l}}(\vec{s'})  f(\vec{s'},\bm{\theta},O,\rho)  \right).
\end{equation}
Actually, we have $\abs{h_{\mathrm{sensitivity}}(\bm{\theta})} \leq 8 \norm{O}_\infty^2$, by following inequalities:
\begin{equation}
\abs{\sum_{\vec{s}} (1-g_{\mathcal{N}}(\vec{s}))f(\vec{s},\bm{\theta},O,\rho)} = \abs{ \langle O \rangle_{\bm{\theta}} - \langle \widetilde{O} \rangle_{\bm{\theta}} } \leq 2 \norm{O}_\infty
\end{equation}  and 
\begin{equation}
\abs{\sum_{\vec{s'}}  \left( \mathbf{1}_{(s_l|_{\mathcal{N}_{l}}=\mathbb{I}/\sqrt{2})}(\vec{s'})-1 \right)g_{\mathcal{N}\setminus \mathcal{N}_{l}}(\vec{s'})  f(\vec{s'},\bm{\theta},O,\rho)  }= \abs{ \langle \widetilde{O} \rangle_{\bm{\theta}}' - \langle \widetilde{O} \rangle_{\bm{\theta}} '' } \leq 2\norm{O}_\infty,
\end{equation}
where $\langle \widetilde{O} \rangle_{\bm{\theta}}'$ and $\langle \widetilde{O} \rangle_{\bm{\theta}}''$ are the noisy expectation values of the observable $O$ with respect to set the noise strength $\lambda_{l}$ as $1$ and $0$, respectively.

\subsubsection{Single-qubit amplitude damping noise}
In the last section, we analyzed the noise sensitivity for single-qubit depolarizing noise. But depolarizing noise is not a representative example for PCS1 channels, since it is unital and diagonal.
Here, we consider another common noise model, the single-qubit amplitude damping noise $\mathcal{N}^{(amp)}_{l}$ with strength $\gamma_{l}$, where $l$ is the layer index.
The single-qubit amplitude damping noise $\mathcal{N}^{(amp)}_{l}$ can be expressed using Pauli transfer matrix as shown in Eq.~\eqref{eq:amplitude_damping_ptm}.

Using Eq.~\eqref{ap:eq:mse_noise_robustness_derivation} and Eq.~\eqref{eq:pauli_path_integral_noisy}, we can express the gradient of the MSE with respect to the noise strength $\gamma_{l}$ as:
\begin{equation}
  \begin{aligned}
    \frac{\partial \mathrm{MSE}(\langle O \rangle)}{\partial \gamma_{l}} =& \frac{\partial}{\partial \gamma_{l}} \left[\frac{1}{4^{N_g}}\sum_{\bm{\theta} \in \{0,\frac{\pi}{2},\pi,\frac{3\pi}{2}\}^{N_g}} \left( \langle O \rangle_{\bm{\theta}} - \langle \widetilde{O} \rangle_{\bm{\theta}} \right)^2\right]\\
    =& \frac{-2 }{4^{N_g}}\sum_{\bm{\theta} \in \{0,\frac{\pi}{2},\pi,\frac{3\pi}{2}\}^{N_g}} \left( \langle O \rangle_{\bm{\theta}} - \langle \widetilde{O} \rangle_{\bm{\theta}} \right) \frac{\partial}{\partial \gamma_{l}} \langle \widetilde{O} \rangle_{\bm{\theta}}\\
    =& \frac{-2}{4^{N_g}}\sum_{\bm{\theta} \in \{0,\frac{\pi}{2},\pi,\frac{3\pi}{2}\}^{N_g}} \left(  \sum_{\vec{s}} f(\vec{s},\bm{\theta},O,\rho) - \widetilde{f}(\vec{s},\bm{\theta},O,\rho) \right) \frac{\partial}{\partial \gamma_{l}} \sum_{\vec{s}} \widetilde{f}(\vec{s},\bm{\theta},O,\rho)\\
    =& \frac{-2 }{4^{N_g}}\sum_{\bm{\theta} \in \{0,\frac{\pi}{2},\pi,\frac{3\pi}{2}\}^{N_g}} \left(  \sum_{\vec{s}} f(\vec{s},\bm{\theta},O,\rho) - \widetilde{f}(\vec{s},\bm{\theta},O,\rho) \right)\\
    & \cdot \frac{\partial}{\partial \gamma_{l}} \left\{\sum_{\vec{s}} \sum_{\vec{\tau}}  \tr{O s_{N_g}} \tr{s_0 \rho} \prod_{i=1}^{{N_g}} \tr{\tau_i U_i(\theta_i) s_{i-1} U_i^\dagger(\theta_i)} \tr{s_i \mathcal{N}^{(amp)}_{i}(\tau_i)} \right\} \\
    =& \frac{-2}{4^{N_g}}\sum_{\bm{\theta} \in \{0,\frac{\pi}{2},\pi,\frac{3\pi}{2}\}^{N_g}} \left(  \sum_{\vec{s}} f(\vec{s},\bm{\theta},O,\rho) - \widetilde{f}(\vec{s},\bm{\theta},O,\rho) \right)  \left\{\sum_{\vec{s}} \sum_{\vec{\tau}}  \tr{O s_{N_g}} \tr{s_0 \rho} \right.\\
    & \cdot \left. \left(\prod_{i=1}^{{N_g}} \tr{\tau_i U_i(\theta_i) s_{i-1} U_i^\dagger(\theta_i)}\right) \left(\prod_{i=1, i\neq l}^{{N_g}}\tr{s_i \mathcal{N}^{(amp)}_{i}(\tau_i)}\right) \left(\frac{\partial}{\partial \gamma_{l}}\tr{s_l \mathcal{N}^{(amp)}_{l}(\tau_l)}\right) \right\}.
  \end{aligned}
\end{equation}

By the Pauli transfer matrix of the single-qubit amplitude damping noise shown in Eq.~\eqref{eq:amplitude_damping_ptm}, we have:
\begin{equation}
  \frac{\partial}{\partial \gamma_{l}}\tr{s_l \mathcal{N}^{(amp)}_{l}(\tau_l)} 
  =\begin{cases}
    -\frac{1}{2\sqrt{1-\gamma_{l}}} & \text{if } s_l = \tau_l \text{ and } s_l\mid_{\mathcal{N}^{(amp)}_{l}} \in \{X/\sqrt{2}, Y/\sqrt{2}\},\\
    -1 & \text{if } s_l = \tau_l \text{ and } s_l\mid_{\mathcal{N}^{(amp)}_{l}} = Z/\sqrt{2},\\
    1 & \text{if }s_l \mid_{\mathcal{N}^{(amp)}_{l}}= \mathbb{I}/\sqrt{2}, \tau_l\mid_{\mathcal{N}^{(amp)}_{l}} = Z/\sqrt{2}, \text{ and } s_l\mid^c_{\mathcal{N}^{(amp)}_{l}} = \tau_l\mid^c_{\mathcal{N}^{(amp)}_{l}},\\
    0 & \text{otherwise}.
  \end{cases}
\end{equation}
where $s_l\mid_{\mathcal{N}^{(amp)}_{l}}$ and $\tau_l\mid_{\mathcal{N}^{(amp)}_{l}}$ are the restrictions of the Pauli operators $s_l$ and $\tau_l$ to the subsystem acted on by the noise $\mathcal{N}^{(amp)}_{l}$, respectively, and $s_l\mid^c_{\mathcal{N}^{(amp)}_{l}}$ and $\tau_l\mid^c_{\mathcal{N}^{(amp)}_{l}}$ are the restrictions of the Pauli operators $s_l$ and $\tau_l$ to the complement of the subsystem acted on by the noise $\mathcal{N}^{(amp)}_{l}$, respectively.

We consider a linear map $\mathcal{E}^{(amp)}_{l}$, defined such that its Pauli transfer matrix is given by $\frac{1}{2}\frac{\partial}{\partial \gamma_{l}}\tr{s_l \mathcal{N}^{(amp)}_{l}(\tau_l)}$.
Equivalently, its Pauli transfer matrix $\mathcal{S}_{\mathcal{E}^{(amp)}_{l}}$, with entries $\left(\mathcal{S}_{\mathcal{E}^{(amp)}_{l}}\right)_{i,j}=\tr{\mathcal{E}^{(amp)}_{l}(\sigma_i) \sigma_j}$ can be expressed as:
\begin{equation}
  \mathcal{S}_{\mathcal{E}^{(amp)}_{l}} = \begin{pmatrix}
    \frac{1}{2} & 0 & 0 & \frac{1}{2} \\
    0 & -\frac{1}{4\sqrt{1-\gamma_{l}}} & 0 & 0 \\
    0 & 0 & -\frac{1}{4\sqrt{1-\gamma_{l}}} & 0 \\
    0 & 0 & 0 & -\frac{1}{2}
  \end{pmatrix},
\end{equation}
where $\sigma_0=\mathbb{I}$, $\sigma_1=X$, $\sigma_2=Y$, and $\sigma_3=Z$.
Although the map $\mathcal{E}^{(amp)}_{l}$ is not a quantum channel, we can still use it to express the gradient of the MSE with respect to the noise strength $\gamma_{l}$ as:
\begin{equation}\label{ap:eq:amplitude_damping_noise_sensitivity}
  \abs{\frac{\partial \mathrm{MSE}(\langle O \rangle)}{\partial \gamma_{l}}} = \abs{\mathbb{E}^*_{\bm{\theta}} h_{\mathrm{sensitivity}}^{(amp)}(\bm{\theta})},
\end{equation}
where $\mathbb{E}^*_{\bm{\theta}}$ represents the expectation value that $\bm{\theta}$ takes values in $\{0, \frac{\pi}{2}, \pi, \frac{3\pi}{2}\}^{N_g}$ with uniform distribution, and
\begin{equation}
    h_{\mathrm{sensitivity}}^{(amp)}(\bm{\theta}) = -2 \left(  \sum_{\vec{s}} f(\vec{s},\bm{\theta},O,\rho) - \widetilde{f}(\vec{s},\bm{\theta},O,\rho) \right) \left\{\sum_{\vec{s},\vec{\tau}} 2 \widetilde{f}_{*}^{(\vec{\tau})}(\vec{s},\bm{\theta},O,\rho) \right\},
\end{equation}
with
\begin{equation}
  \widetilde{f}_{*}^{(\vec{\tau})}(\vec{s},\bm{\theta},O,\rho) = \tr{O s_{N_g}} \tr{s_0 \rho} \left(\prod_{i=1}^{{N_g}} \tr{\tau_i U_i(\theta_i) s_{i-1} U_i^\dagger(\theta_i)}\right) \left(\prod_{i=1, i\neq l}^{{N_g}}\tr{s_i \mathcal{N}^{(amp)}_{i}(\tau_i)}\right) \tr{s_l \mathcal{E}^{(amp)}_{l}(\tau_l)}.
\end{equation}

\begin{remark}
  We futher assume that the noise strength $\gamma_{l}$ satisfies $\gamma_{l} < 0.5$.
  Under this condition, $\mathcal{E}^{(amp)}_{l}$ is a PCS1 map, which enable the efficient back-propagation algorithm in Section~\ref{sec:OBPPP}.
  This assumption is reasonable because, in practical quantum devices, the amplitude damping noise strength is typically much smaller than $\frac{1}{2}$.
\end{remark}

It's worth noting that the quantity $\frac{\partial \mathrm{MSE}(\langle O \rangle)}{\partial \gamma_{l}}$ can be alternatively estimated by using an amplitude damping noise channel $\mathcal{N}^{(amp)'}_{l}$ with strength $\gamma'_{l}=1$ and a Pauli noise channel $\mathcal{N}^{(Pauli)}_{l}(\rho)= (1-p_x-p_y-p_z) \rho + p_x X \rho X + p_y Y \rho Y + p_z Z \rho Z$ with $p_x=p_y=0$ and $p_z=\frac{1}{2}-\frac{1}{4\sqrt{1-\gamma_l}}$, respectively. Specifically, we have:
\begin{equation}
  \frac{\partial}{\partial \gamma_{l}}\tr{s_l \mathcal{N}^{(amp)}_{l}(\tau_l)} 
  = \tr{s_l \mathcal{N}^{(amp)'}_{l}(\tau_l)} - \tr{s_l \mathcal{N}^{(Pauli)}_{l}(\tau_l)}.
\end{equation}
By linearity, we have:
\begin{equation}
  \begin{aligned}
    &\abs{\left\{\sum_{\vec{s},\vec{\tau}} \tr{O s_{N_g}} \tr{s_0 \rho}  \left(\frac{\partial}{\partial \gamma_{l}}\tr{s_l \mathcal{N}^{(amp)}_{l}(\tau_l)}\right) \left(\prod_{i=1}^{{N_g}} \tr{\tau_i U_i(\theta_i) s_{i-1} U_i^\dagger(\theta_i)}\right) \left(\prod_{i=1, i\neq l}^{{N_g}}\tr{s_i \mathcal{N}^{(amp)}_{i}(\tau_i)}\right) \right\} } \\
    &= \abs{ \langle \widetilde{O} \rangle_{\bm{\theta}}' - \langle \widetilde{O} \rangle_{\bm{\theta}} '' } \leq 2\norm{O}_\infty,
  \end{aligned}
\end{equation}
where $\langle \widetilde{O} \rangle_{\bm{\theta}}'$ and $\langle \widetilde{O} \rangle_{\bm{\theta}}''$ are the noisy expectation values of the observable $O$ with respect to replace the noise channel at layer $l$ as $\mathcal{N}^{(amp)'}_{l}$ and $\mathcal{N}^{(Pauli)}_{l}$, respectively, while keeping other noise channels unchanged.

Therefore, combining with the bound of $\abs{ \langle O \rangle_{\bm{\theta}} - \langle \widetilde{O} \rangle_{\bm{\theta}} } \leq 2 \norm{O}_\infty$, we have $\abs{h_{\mathrm{sensitivity}}^{(amp)}(\bm{\theta})} \leq 4\norm{O}_\infty^2$.

\subsection{Noise sensitivity to Fisher Information}

In this section, we establish a connection between the noise sensitivity and the (classical) Fisher information associated with estimating the noise rate $\lambda$ from measurement outcomes.
Fix a noise channel $\mathcal{N}(\lambda)$ with noise rate $\lambda$, the noise sensitivity is given by:
\begin{equation}
\left|\frac{\partial \mathrm{MSE}(\langle O \rangle)}{\partial \lambda}\right| = \left| \mathbb{E}_{\bm{\theta}} \frac{\partial }{\partial \lambda} \left( \langle O \rangle_{\bm{\theta}} - \langle \widetilde{O} \rangle_{\bm{\theta}} \right)^2 \right| = \left| -2\mathbb{E}_{\bm{\theta}}  \left( \langle O \rangle_{\bm{\theta}} - \langle \widetilde{O} \rangle_{\bm{\theta}} \right)  \frac{\partial }{\partial \lambda}  \langle \widetilde{O} \rangle_{\bm{\theta}}  \right|.
\end{equation}

To make this connection explicit, consider estimating the expectation value of a Pauli observable $O \in \{\mathbb{I},X,Y,Z\}^{\otimes n}$, whose measurement outputs are $X \in \{+1,-1\}$ with probabilities $\Pr(X=+1)=p(\lambda)$ and $\Pr(X=-1)=1- p(\lambda)$, respectively.
Such measurements can be implemented by performing a projective measurement in the eigenbasis of $O$.Accordingly, the observed noisy expectation value is:
\begin{equation}
  \langle \widetilde{O} \rangle_{\bm{\theta}} = \mathbb{E}\left[X|\lambda,\bm{\theta} \right] = 2p(\lambda)-1.
\end{equation}

Treating $\lambda$ as an unknown parameter in a quantum metrology task, the classical Fisher information associated with estimating $\lambda$ from the measurement outcome $X$ is given by:
\begin{equation}
  \begin{aligned}
  F_{\bm{\theta}}(\lambda) &= \mathbb{E}_{X}  \left( \frac{\partial}{\partial \lambda} \ln p(X|\lambda,\bm{\theta}) \right)^2 = \mathbb{E}_{X}  \left( \frac{1}{p(X|\lambda,\bm{\theta})} \frac{\partial p(X|\lambda,\bm{\theta})}{\partial \lambda} \right)^2 \\
&= \sum_{X} \frac{1}{p(X|\lambda,\bm{\theta})} \left( \frac{\partial p(X|\lambda,\bm{\theta})}{\partial \lambda} \right)^2 
  = \frac{p^\prime(\lambda)^2}{p(\lambda)}+\frac{p^\prime(\lambda)^2}{1-p(\lambda)} = \frac{p^\prime(\lambda)^2}{(1-p(\lambda))p(\lambda)},
  \end{aligned}
\end{equation}
where $p(X|\lambda,\bm{\theta})$ is the probability of observing measurement outcome $X$ given noise rate $\lambda$ and parameter configuration $\bm{\theta}$, and $p^\prime(\lambda) = \frac{\partial p(\lambda)}{\partial \lambda}$.

Rewriting the expression for the noise sensitivity of $\lambda$, we have:
\begin{equation}
  \begin{aligned}
  \left|\frac{\partial \mathrm{MSE}(\langle O \rangle)}{\partial \lambda}\right| & = \left|\mathbb{E}_{\bm{\theta}} 4 \left(2p(\lambda) - 1 - \langle O \rangle_{\bm{\theta}}\right)p^\prime(\lambda)\right| \\
&\leq \mathbb{E}_{\bm{\theta}} 4 \left|\left(2p(\lambda) - 1-\langle O \rangle_{\bm{\theta}} \right)p^\prime(\lambda)\right| \\
&= \mathbb{E}_{\bm{\theta}} 2\left|\left(\langle \widetilde{O} \rangle_{\bm{\theta}}-\langle O \rangle_{\bm{\theta}}  \right)\right|\sqrt{F_{\bm{\theta}}(\lambda) \mathrm{Var}\left[X\mid \lambda,\bm{\theta} \right]}.
  \end{aligned}
\end{equation}
Employing the Cauchy–Schwarz inequality and using the fact that $\mathrm{Var}\left[X \mid \lambda, \bm{\theta}\right] \leq 1$, we obtain
\begin{equation}
\left|\frac{\partial \mathrm{MSE}(\langle O \rangle)}{\partial \lambda}\right| \leq 2 \sqrt{\mathrm{MSE}(\langle O \rangle)} \sqrt{\mathbb{E}_{\bm{\theta}} F_{\bm{\theta}}(\lambda)}.
\end{equation}
This inequality implies that there exists at least one configuration $\bm{\theta}$ such that
\begin{equation}
F_{\bm{\theta}}(\lambda) \geq \frac{1}{4\mathrm{MSE}(\langle O \rangle)} \left|\frac{\partial \mathrm{MSE}(\langle O \rangle)}{\partial \lambda}\right|^2.
\end{equation}
Therefore, our framework can also be employed to evaluate the power of the PQC architecture with respect to the estimation of the noise rate $\lambda$.

Assume the expectation value of the observable $O$ is obtained by performing positive operator-valued measure~(POVM) $\mathcal{M}=\{M_1,M_2,\cdots,M_k\}$ with probability $p(x|\lambda)$.
Let $X$ be the random variable representing the measurement outcome, which is distributed according to $p(x|\lambda)$.

The classical Fisher information of the measurement outcome $X$ is given by:
\begin{equation}
  F(\lambda) = \mathbb{E}_{X}  \left( \frac{\partial}{\partial \lambda} \ln p(X|\lambda,\bm{\theta}) \right)^2 = \mathbb{E}_{X}  \left( \frac{1}{p(X|\lambda,\bm{\theta})} \frac{\partial p(X|\lambda,\bm{\theta})}{\partial \lambda} \right)^2 = \sum_{X} \frac{1}{p(X|\lambda,\bm{\theta})} \left( \frac{\partial p(X|\lambda,\bm{\theta})}{\partial \lambda} \right)^2.
\end{equation}

On the other hand, the noise sensitivity of the expectation value of the observable $O$ with respect to the noise strength $\lambda$ is given by:
\begin{equation}
  \begin{aligned}
  \abs{\frac{\partial \mathrm{MSE}(\langle O \rangle)}{\partial \lambda}} &= \abs{ \mathbb{E}_{\bm{\theta}} \frac{\partial \left(\langle O \rangle_{\bm{\theta}} -\sum_{X} X \cdot p(X|\lambda,\bm{\theta}) \right)^2}{\partial \lambda} }\\
  & = \abs{ \mathbb{E}_{\bm{\theta}} \left[ 2\left(\sum_{X} X \cdot p(X|\lambda,\bm{\theta}) \right)\left(\sum_{X} X \cdot \frac{\partial p(X|\lambda,\bm{\theta})}{\partial \lambda} \right)  - 2\langle O \rangle_{\bm{\theta}} \left(\sum_{X} X \cdot \frac{\partial p(X|\lambda,\bm{\theta})}{\partial \lambda} \right) \right] } \\
  & = 2\abs{ \mathbb{E}_{\bm{\theta}} \left[ \left(\langle \widetilde{O} \rangle_{\bm{\theta}} - \langle O \rangle_{\bm{\theta}} \right) \left(\sum_{X} X \cdot \frac{\partial p(X|\lambda,\bm{\theta})}{\partial \lambda} \right)   \right] }
  \end{aligned}
\end{equation}
By the Cauchy-Schwarz inequality, we have:
\begin{equation}
  \abs{ \sum_{X} X \cdot \frac{\partial p(X|\lambda,\bm{\theta})}{\partial \lambda} } \leq \sqrt{ \sum_{X} X^2 \cdot p(X|\lambda,\bm{\theta})} \sqrt{ \sum_{X} \frac{1}{p(X|\lambda,\bm{\theta})} \left( \frac{\partial p(X|\lambda,\bm{\theta})}{\partial \lambda} \right)^2 } = \sqrt{ \mathbb{E}_{X} X^2 } \sqrt{ F(\lambda)}.
\end{equation}

We can obtain the following inequality:
\begin{equation}
  \abs{\frac{\partial \mathrm{MSE}(\langle O \rangle)}{\partial \lambda} } \leq 2 \sqrt{\mathrm{MSE}(\langle O \rangle)} \sqrt{\mathbb{E}_{\bm{\theta}} \abs{ \sum_{X} X \cdot \frac{\partial p(X|\lambda,\bm{\theta})}{\partial \lambda} }^2} \leq 2 \sqrt{\mathrm{MSE}(\langle O \rangle)} \sqrt{\mathbb{E}_{\bm{\theta}} \mathbb{E}_{X} X^2 F(\lambda)} .
\end{equation}

\section{Trainability}

\subsection{Estimation of the variance of the gradient}
\label{sec:vargrad}

The gradient of the expectation value $\langle O \rangle$ with respect to the parameter $\theta_i$ in the parameterized quantum circuit $\mathcal{C}(\bm{\theta})$ is given by:
\begin{equation}
  \frac{\partial \langle O \rangle}{\partial{\theta_i}} = \frac{\partial}{\partial{\theta_i}} \sum_{\vec{s}} f(\vec{s},\bm{\theta},O,\rho) = \sum_{\vec{s}} \frac{\partial f(\vec{s},\bm{\theta},O,\rho)}{\partial{\theta_i}}.
\end{equation}
And by using the parameter shift rule~\cite{schuld2019evaluating}, we can express the gradient as:
\begin{equation}
  \frac{\partial \langle O \rangle}{\partial{\theta_i}} =\frac{ \langle O \rangle_{\bm{\theta}+\frac{\pi}{2}e_i}-\langle O \rangle_{\bm{\theta}-\frac{\pi}{2}e_i}}{2} = \sum_{\vec{s}} \frac{f(\vec{s},\bm{\theta}+\frac{\pi}{2}e_i,O,\rho)-f(\vec{s},\bm{\theta}-\frac{\pi}{2}e_i,O,\rho)}{2},
\end{equation}
where $ \langle O \rangle_{\bm{\theta}}$ is the expectation value of the observable $O$ with respect to rotation angles $\bm{\theta}$, and $e_i$ is the $i$-th unit vector in the parameter space $\bm{\theta}$.

The variance of the gradient $\frac{\partial \langle O \rangle}{\partial{\theta_i}}$ is an important quantity in quantum machine learning, as it provides a measure of the trainability of the quantum circuit.
Because the expectation of the gradient $\frac{\partial \langle O \rangle}{\partial{\theta_i}}$ can be calculated by:
\begin{equation}
  \begin{aligned}
    \mathbb{E}_{\bm{\theta}} \left( \frac{\partial \langle O \rangle}{\partial{\theta_i}} \right) & = \mathbb{E}_{\bm{\theta}} \left(\frac{ \langle O \rangle_{\bm{\theta}+\frac{\pi}{2}e_i}-\langle O \rangle_{\bm{\theta}-\frac{\pi}{2}e_i}}{2}\right)\\
    &=  \frac{ \mathbb{E}_{\bm{\theta}}\langle O \rangle_{\bm{\theta}}-\mathbb{E}_{\bm{\theta}}\langle O \rangle_{\bm{\theta}}}{2}\\
    &= 0,
  \end{aligned}
\end{equation}
where $\mathbb{E}_{\bm{\theta}}$ denotes the expectation over the uniform parameter space $\bm{\theta}\in [0,2\pi)$.

The variance of the gradient $\frac{\partial \langle O \rangle}{\partial{\theta_i}}$ can be expressed as:
\begin{equation}
  \begin{aligned}\label{eq:variance_gradient}
    \mathrm{Var}\left(\frac{\partial \langle O \rangle}{\partial{\theta_i}}\right) &= \mathbb{E}_{\bm{\theta}} \left(\frac{ \langle O \rangle_{\bm{\theta}+\frac{\pi}{2}e_i}-\langle O \rangle_{\bm{\theta}-\frac{\pi}{2}e_i}}{2} \right)^2\\
    &= \mathbb{E}_{\bm{\theta}} \left(\sum_{\vec{s}} \frac{f(\vec{s},\bm{\theta}+\frac{\pi}{2}e_i,O,\rho)-f(\vec{s},\bm{\theta}-\frac{\pi}{2}e_i,O,\rho)}{2}\right)^2\\
    &= \sum_{\vec{s},\vec{s'}} \frac{1}{4} \mathbb{E}_{\bm{\theta}} \left(f(\vec{s},\bm{\theta}+\frac{\pi}{2}e_i,O,\rho)-f(\vec{s},\bm{\theta}-\frac{\pi}{2}e_i,O,\rho)\right) \left(f(\vec{s'},\bm{\theta}+\frac{\pi}{2}e_i,O,\rho)-f(\vec{s'},\bm{\theta}-\frac{\pi}{2}e_i,O,\rho)\right)\\
    &= \sum_{\vec{s},\vec{s'}} \frac{1}{4} \biggl( \mathbb{E}_{\bm{\theta}}f(\vec{s},\bm{\theta}+\frac{\pi}{2}e_i,O,\rho)f(\vec{s'},\bm{\theta}+\frac{\pi}{2}e_i,O,\rho)-\mathbb{E}_{\bm{\theta}}f(\vec{s},\bm{\theta}-\frac{\pi}{2}e_i,O,\rho)f(\vec{s'},\bm{\theta}-\frac{\pi}{2}e_i,O,\rho)  \\
    &  \quad + \mathbb{E}_{\bm{\theta}}f(\vec{s},\bm{\theta}+\frac{\pi}{2}e_i,O,\rho)f(\vec{s'},\bm{\theta}-\frac{\pi}{2}e_i,O,\rho)-\mathbb{E}_{\bm{\theta}}f(\vec{s},\bm{\theta}-\frac{\pi}{2}e_i,O,\rho)f(\vec{s'},\bm{\theta}+\frac{\pi}{2}e_i,O,\rho) \biggr).
  \end{aligned}
\end{equation}
The first term in the above equation can be simplified as:
\begin{equation}
  \begin{aligned}
    &\frac{1}{4} \mathbb{E}_{\bm{\theta}}f(\vec{s},\bm{\theta}+\frac{\pi}{2}e_i,O,\rho)f(\vec{s'},\bm{\theta}+\frac{\pi}{2}e_i,O,\rho)\\
    =&\frac{1}{4} \tr{O s_{N_g}}\tr{O s'_{N_g}} \tr{s_0 \rho}\tr{s'_0 \rho} \prod_{j=1,j\neq i}^{{N_g}}\mathbb{E}_{\theta_j} \tr{s_j U_j(\theta_j) s_{j-1} U_j^\dagger(\theta_j)}\tr{s'_j U_j(\theta_j) s'_{j-1} U_j^\dagger(\theta_j)}\\
    & \quad \mathbb{E}_{\theta_i}\tr{s_i U_i(\theta_i+\frac{\pi}{2}) s_{i-1} U_i^\dagger(\theta_i+\frac{\pi}{2})}\tr{s'_i U_i(\theta_i+\frac{\pi}{2}) s'_{i-1} U_i^\dagger(\theta_i+\frac{\pi}{2})}
  \end{aligned}
\end{equation}

For terms $\mathbb{E}_{\theta_j} \tr{s_j U_j(\theta_j) s_{j-1} U_j^\dagger(\theta_j)}\tr{s'_j U_j(\theta_j) s'_{j-1} U_j^\dagger(\theta_j)}$ we have:
\begin{equation}
  \begin{aligned}
    &\mathbb{E}_{\theta_j} \tr{s_j U_j(\theta_j) s_{j-1} U_j^\dagger(\theta_j)}\tr{s'_j U_j(\theta_j) s'_{j-1} U_j^\dagger(\theta_j)}\\
    =&\mathbb{E}_{\theta_j} \tr{s_j \exp{-i \frac{\theta_j}{2} P_j}C_j s_{j-1} C_j^\dagger \exp{i \frac{\theta_j}{2} P_j}}\tr{s'_j \exp{-i \frac{\theta_j}{2} P_j}C_j s'_{j-1} C_j^\dagger \exp{i \frac{\theta_j}{2} P_j}}\\
    =&\frac{1}{4}\sum_{\theta_j \in \{0,\frac{\pi}{2},\pi,\frac{3\pi}{2}\}} \tr{s_j \exp{-i \frac{\theta_j}{2} P_j}C_j s_{j-1} C_j^\dagger \exp{i \frac{\theta_j}{2} P_j}}\tr{s'_j \exp{-i \frac{\theta_j}{2} P_j}C_j s'_{j-1} C_j^\dagger \exp{i \frac{\theta_j}{2} P_j}}\\
    =&\frac{1}{4}\sum_{\theta_j \in \{0,\frac{\pi}{2},\pi,\frac{3\pi}{2}\}} \tr{s_j U_j(\theta_j) s_{j-1} U_j^\dagger(\theta_j)}\tr{s'_j U_j(\theta_j) s'_{j-1} U_j^\dagger(\theta_j)},
  \end{aligned}
\end{equation}
where the second equality is obtained by using the fact in Eq.~\eqref{eq:two_design_cor}.

Similarly, for terms $\mathbb{E}_{\theta_i}\tr{s_i U_i(\theta_i+\frac{\pi}{2}) s_{i-1} U_i^\dagger(\theta_i+\frac{\pi}{2})}\tr{s'_i U_i(\theta_i+\frac{\pi}{2}) s'_{i-1} U_i^\dagger(\theta_i+\frac{\pi}{2})}$ we have:
\begin{equation}
  \begin{aligned}
    &\mathbb{E}_{\theta_i}\tr{s_i U_i(\theta_i+\frac{\pi}{2}) s_{i-1} U_i^\dagger(\theta_i+\frac{\pi}{2})}\tr{s'_i U_i(\theta_i+\frac{\pi}{2}) s'_{i-1} U_i^\dagger(\theta_i+\frac{\pi}{2})}\\
    =&\mathbb{E}_{\theta_i}\tr{s_i \exp{-i \frac{\theta_i+\frac{\pi}{2}}{2} P_i}C_i s_{i-1} C_i^\dagger \exp{i \frac{\theta_i+\frac{\pi}{2}}{2} P_i}}\tr{s'_i \exp{-i \frac{\theta_i+\frac{\pi}{2}}{2} P_i}C_i s'_{i-1} C_i^\dagger \exp{i \frac{\theta_i+\frac{\pi}{2}}{2} P_i}}\\
    =&\mathbb{E}_{\theta_i}\tr{s_i \exp{-i \frac{\theta_i}{2} P_i}\exp{-i \frac{\pi}{4} P_i}C_i s_{i-1} C_i^\dagger \exp{i \frac{\pi}{4} P_i} \exp{i \frac{\theta_i}{2} P_i}}\\
    &\tr{s'_i \exp{-i \frac{\theta_i}{2} P_i}\exp{-i \frac{\pi}{4} P_i}C_i s'_{i-1} C_i^\dagger \exp{i \frac{\pi}{4} P_i} \exp{i \frac{\theta_i}{2} P_i}}\\
    =&\frac{1}{4}\sum_{\theta_i \in \{0,\frac{\pi}{2},\pi,\frac{3\pi}{2}\}}\tr{s_i U_i(\theta_i+\frac{\pi}{2}) s_{i-1} U_i^\dagger(\theta_i+\frac{\pi}{2})}\tr{s'_i U_i(\theta_i+\frac{\pi}{2}) s'_{i-1} U_i^\dagger(\theta_i+\frac{\pi}{2})}.
  \end{aligned}
\end{equation}

Thus we have:
\begin{equation}
  \mathbb{E}_{\bm{\theta}}f(\vec{s},\bm{\theta}+\frac{\pi}{2}e_i,O,\rho)f(\vec{s'},\bm{\theta}+\frac{\pi}{2}e_i,O,\rho)=\frac{1}{4^{N_g}}\sum_{\bm{\theta} \in \{0,\frac{\pi}{2},\pi,\frac{3\pi}{2}\}^{N_g}}f(\vec{s},\bm{\theta}+\frac{\pi}{2}e_i,O,\rho)f(\vec{s'},\bm{\theta}+\frac{\pi}{2}e_i,O,\rho).
\end{equation}

As a result, by similarly replacing $\mathbb{E}_{\bm{\theta}}$ with $\frac{1}{4^{N_g}}\sum_{\bm{\theta} \in \{0,\frac{\pi}{2},\pi,\frac{3\pi}{2}\}^{N_g}}$ in each terms of Eq.~\eqref{eq:variance_gradient}, we can obtain:
\begin{equation}\label{ap:eq:discrete_variance_gradient}
  \begin{aligned}
    \mathrm{Var}\left(\frac{\partial \langle O \rangle}{\partial{\theta_i}}\right)& = \mathbb{E}_{\bm{\theta}} \left(\frac{ \langle O \rangle_{\bm{\theta}+\frac{\pi}{2}e_i}-\langle O \rangle_{\bm{\theta}-\frac{\pi}{2}e_i}}{2} \right)^2=\frac{1}{4^{N_g}}\sum_{\bm{\theta} \in \{0,\frac{\pi}{2},\pi,\frac{3\pi}{2}\}^{N_g}} \left(\frac{ \langle O \rangle_{\bm{\theta}+\frac{\pi}{2}e_i}-\langle O \rangle_{\bm{\theta}-\frac{\pi}{2}e_i}}{2} \right)^2\\
    &= \frac{1}{4^{N_g}}\sum_{\bm{\theta} \in \{0,\frac{\pi}{2},\pi,\frac{3\pi}{2}\}^{N_g}} \left(\sum_{\vec{s}} \frac{f(\vec{s},\bm{\theta}+\frac{\pi}{2}e_i,O,\rho)-f(\vec{s},\bm{\theta}-\frac{\pi}{2}e_i,O,\rho)}{2}\right)^2 \\
    &=\frac{1}{4^{N_g}}\sum_{\bm{\theta} \in \{0,\frac{\pi}{2},\pi,\frac{3\pi}{2}\}^{N_g}} \left\{\sum_{\vec{s}} \tr{O s_{N_g}} \tr{s_0 \rho} \prod_{j=1,j\neq i}^{{N_g}}\mathbb{E}_{\theta_j} \tr{s_j U_j(\theta_j) s_{j-1} U_j^\dagger(\theta_j)} \right.\\
    &\quad  \frac{\tr{s_i U_i(\theta_i+\frac{\pi}{2}) s_{i-1} U_i^\dagger(\theta_i+\frac{\pi}{2})} - \tr{s_i U_i(\theta_i-\frac{\pi}{2}) s_{i-1} U_i^\dagger(\theta_i-\frac{\pi}{2})}}{2} \Biggr\}^2.
  \end{aligned}
\end{equation}
On the other hand, by Eq~\eqref{eq:gate_term_in_f_discrete}, if the rotation angle $\theta_i$ in the gate $U_i(\theta_i) = C_i e^{-i\frac{\theta_i}{2} P}$ takes values in $\{0, \frac{\pi}{2}, \pi, \frac{3\pi}{2}\}$, we have:
\begin{equation}
  \frac{\tr{s_i U_i(\theta_i+\frac{\pi}{2}) s_{i-1} U_i^\dagger(\theta_i+\frac{\pi}{2})} - \tr{s_i U_i(\theta_i-\frac{\pi}{2}) s_{i-1} U_i^\dagger(\theta_i-\frac{\pi}{2})}}{2} = \begin{cases}
    0, & [P_i, s_{i-1}] = 0, \\
    \tr{s_i U_i(\theta_i+\frac{\pi}{2}) s_{i-1} U_i^\dagger(\theta_i+\frac{\pi}{2})}, & \{P_i, {s_{i-1}}\} = 0.
    \end{cases}
\end{equation}
Therefore, we have:
\begin{equation}
  \mathrm{Var}\left(\frac{\partial \langle O \rangle}{\partial{\theta_i}}\right) = \frac{1}{4^{N_g}}\sum_{\bm{\theta} \in \{0,\frac{\pi}{2},\pi,\frac{3\pi}{2}\}^{N_g}} \left(\sum_{\vec{s}: \{P_i, {s_{i-1}}\} = 0} f(\vec{s},\bm{\theta}+\frac{\pi}{2}e_i,O,\rho)\right)^2
\end{equation}

For noisy case, or there is PCS1 channel in circuit, we can also use the same method to estimate the variance of the gradient, and rewrite the variance of the gradient of the noisy expectation value $\langle \widetilde{O} \rangle$ as:
\begin{equation}\label{ap:eq:variance_gradient_noisy}
  \mathrm{Var}\left(\frac{\partial \langle \widetilde{O} \rangle}{\partial{\theta_i}}\right) = \frac{1}{4^{N_g}}\sum_{\bm{\theta} \in \{0,\frac{\pi}{2},\pi,\frac{3\pi}{2}\}^{N_g}} \left(\sum_{\vec{s}: \{P_i, {s_{i-1}}\} = 0} \widetilde{f}(\vec{s},\bm{\theta}+\frac{\pi}{2}e_i,O,\rho)\right)^2 = \mathbb{E}^*_{\bm{\theta}} h_{\mathrm{trainability}}(\bm{\theta}),
\end{equation}
where $\mathbb{E}^*_{\bm{\theta}}$ represents the expectation value that $\bm{\theta}$ takes values in $\{0, \frac{\pi}{2}, \pi, \frac{3\pi}{2}\}^{N_g}$ with uniform distribution, and
\begin{equation}
  h_{\mathrm{trainability}}(\bm{\theta}) = \left(\sum_{\vec{s}} \mathbf{1}_{ac}(s_{i-1},P_i) \widetilde{f}(\vec{s},\bm{\theta}+\frac{\pi}{2}e_i,O,\rho)\right)^2 ,
\end{equation}
where $\mathbf{1}_{ac}(s_{i-1},P_i)$ is the indicator function that equals to $1$ if $s_{i-1}$ and $P_i$ anti-commute, otherwise it equals to $0$.
There also holds that $\abs{h_{\mathrm{trainability}}(\bm{\theta})} = \left(\frac{ \langle \widetilde{O} \rangle_{\bm{\theta}+\frac{\pi}{2}e_i}-\langle \widetilde{O} \rangle_{\bm{\theta}-\frac{\pi}{2}e_i}}{2} \right)^2 \leq \norm{O}_\infty^2$.

\subsection{Noise induced barren plateaus}

In this section, we present an alternative proof for the noise induced barren plateaus phenomenon by using the Pauli path integral formalism.
For simplicity, we only consider the case that the observable $O$ is a single Pauli operator $O=P$.

Recall in Eq.~\eqref{eq:pauli_path_integral_noiseless}, the expectation value of an observable $O$ for circuit $\mathcal{C}(\bm{\theta})$ can be written as:
\begin{equation}
    \langle O \rangle = \sum_{\vec{s}} f(\vec{s},\bm{\theta},O,\rho)
\end{equation}
then we have 
\begin{equation}
\begin{aligned}
  \frac{\partial \langle P \rangle}{\partial{\theta_j}} &= \sum_{\vec{s}} \frac{\partial f(\vec{s},\bm{\theta},P,\rho)}{\partial{\theta_j}}\\
  &= \sum_{s_{N_g},s_{{N_g}-1},\cdots,s_0}  \tr{s_0 \rho} \tr{P s_{N_g}}\prod_{i\neq j}^{{N_g}} \tr{s_i U_i(\theta_i) s_{i-1} U_i^\dagger(\theta_i)} \frac{\partial }{\partial{\theta_j}}\left(\tr{s_j U_j(\theta_j) s_{j-1} U_j^\dagger(\theta_j)}\right).
\end{aligned}  
\end{equation}

We compute the explicit form of the term associated with the partial derivative:
\begin{equation}
  \begin{aligned}
    &\frac{\partial }{\partial{\theta_j}}\left(\tr{s_j U_j(\theta_j) s_{j-1} U_j^\dagger(\theta_j)}\right)\\
    = &\frac{\partial }{\partial{\theta_j}}\left(\tr{s_j\exp{-i \frac{\theta_j}{2} P_j} \underbrace{C_j s_{j-1} C_j^\dagger}_{Q_j} \exp{i \frac{\theta_j}{2} P_j}}\right)\\
    = &\begin{cases}
      0, & [P_j, Q] = 0, \\
      -\sin(\theta_j) & \{P_j, {Q_j}\} = 0, s_j = Q_j,\\
      \cos(\theta_j) & \{P_j, {Q_j}\} = 0, s_j = iQ_jP_j.
      \end{cases}
  \end{aligned}
\end{equation}

In Ref.~\cite{shao2024simulating}, it has been shown that when $O$ is a single Pauli operator, the contribution function $f(\vec{s},\bm{\theta},P,\rho)$ has the orthogonal property for different Pauli paths $\vec{s}$ and $\vec{s'}$:
\begin{equation}
  \mathbb{E}_{\bm{\theta}} f(\vec{s},\bm{\theta},P,\rho)f(\vec{s'},\bm{\theta},P,\rho)=0, \quad \vec{s}\neq \vec{s'}.
\end{equation}

Following the same proof in Ref.~\cite{shao2024simulating}, for partial derivative of contribution function $\frac{\partial }{\partial{\theta_j}}f(\vec{s},\bm{\theta},P,\rho)$, there also holds the orthogonal property:
\begin{equation}
  \mathbb{E}_{\bm{\theta}} \frac{\partial }{\partial{\theta_j}}f(\vec{s},\bm{\theta},P,\rho)\frac{\partial }{\partial{\theta_j}}f(\vec{s'},\bm{\theta},P,\rho)=0, \quad \vec{s}\neq \vec{s'}.
\end{equation}

Thus, the following expressions can be derived:
\begin{equation}
  \begin{aligned}
    \mathbb{E}_{\bm{\theta}}\left(\frac{\partial \langle P \rangle}{\partial{\theta_j}}\right)^2 & = \mathbb{E}_{\bm{\theta}}\left(\sum_{\vec{s}} \frac{\partial f(\vec{s},\bm{\theta},P,\rho)}{\partial{\theta_j}}\right)^2\\
    & = \mathbb{E}_{\bm{\theta}}\sum_{\vec{s}} \left( \frac{\partial f(\vec{s},\bm{\theta},P,\rho)}{\partial{\theta_j}}\right)^2
  \end{aligned}
\end{equation}

If there is a noise channel $\mathcal{N}^{\otimes n}$ following each $U_i({\theta})$, where $\mathcal{N}(\phi)=\left(1-\lambda\right) + \lambda \frac{\tr{\phi}}{\tr{\mathbb{I}}}\mathbb{I}$ is the depolarizing channel. 
We denote $\widetilde{U}_i(\theta_i)=\mathcal{N}^{\otimes n} \circ U_i(\theta_i)$ as the noisy gate in the $i$-layer. 
Then, as discussed in Ref.~\cite{shao2024simulating}, there is $\tr{s_i \widetilde{U}_i(\theta_i) (s_{i-1})}=\tr{\mathcal{N}^{\dagger \otimes n}(s_i) U_i(\theta_i) s_{i-1} U_i^\dagger(\theta_i)}=(1-\lambda)^\abs{s_i}\tr{s_i U_i(\theta_i) s_{i-1} U_i^\dagger(\theta_i)}$, where $\abs{s_i}$ is the number of non-identity Pauli components in $s_i$.
Thus, for noisy partial derivative $\frac{\partial \widetilde{f}(\vec{s},\bm{\theta},P,\rho)}{\partial{\theta_j}}$, there is:
\begin{equation}
  \frac{\partial \widetilde{f}(\vec{s},\bm{\theta},P,\rho)}{\partial{\theta_j}} = (1-\lambda)^\abs{\vec{s}} \frac{\partial f(\vec{s},\bm{\theta},P,\rho)}{\partial{\theta_j}},
\end{equation}
where $\abs{\vec{s}}=\sum_{i=0}^{N_g} \abs{s_i}$.

As a result, the following expressions can be derived:
\begin{equation}
  \begin{aligned}
  \mathbb{E}_{\bm{\theta}}\left(\frac{\partial \langle \widetilde{P} \rangle}{\partial{\theta_j}}\right)^2 & = \mathbb{E}_{\bm{\theta}} \left( \sum_{\vec{s}} \frac{\partial \widetilde{f}(\vec{s},\bm{\theta},P,\rho)}{\partial{\theta_j}}\right)^2\\
  &=\mathbb{E}_{\bm{\theta}} \sum_{\vec{s}} \left( \frac{\partial \widetilde{f}(\vec{s},\bm{\theta},P,\rho)}{\partial{\theta_j}}\right)^2\\  
  &=\mathbb{E}_{\bm{\theta}} \sum_{\vec{s}} (1-\lambda)^{2\abs{\vec{s}}} \left( \frac{\partial f(\vec{s},\bm{\theta},P,\rho)}{\partial{\theta_j}}\right)^2\\  
  & \leq \mathbb{E}_{\bm{\theta}} \sum_{\vec{s}} (1-\lambda)^{2N_g} \left( \frac{\partial f(\vec{s},\bm{\theta},P,\rho)}{\partial{\theta_j}}\right)^2\\  
  & = (1-\lambda)^{2N_g} \mathbb{E}_{\bm{\theta}} \sum_{\vec{s}} \left( \frac{\partial f(\vec{s},\bm{\theta},P,\rho)}{\partial{\theta_j}}\right)^2 = (1-\lambda)^{2N_g} \mathbb{E}_{\bm{\theta}}\left(\frac{\partial \langle P \rangle}{\partial{\theta_j}}\right)^2.
\end{aligned}
\end{equation}
The inequality holds because in any valid Pauli path, each Pauli word must contribute at least one non-trivial Pauli component; otherwise, the entire path becomes trivial and does not contribute to the sum. Therefore, when ${N_g} = \mathrm{poly}(n)$, the variance decays exponentially in $n$, leading to the emergence of barren plateaus.

\section{Expressibility}\label{sec:expressibility}
\subsection{Expressibility of ideal unitary PQC}
The expressibility of a parameterized quantum circuit refers to its ability to generate states that are well representative of the Hilbert space. 
A circuit with higher expressibility can explore a larger portion of the state space, which is often desirable in variational quantum algorithms. One widely used way to quantify the expressibility of a PQC ensemble $\mathcal{C}(\bm{\theta})$ is defined by how close the ensemble $\mathcal{C}(\bm{\theta})$ is to a unitary $t$-design~\cite{holmes2022connecting,sim2019expressibility,thanasilp2024exponential}~(or equivalently, to the Haar measure, and we use $t$-moment difference to quantify):
\begin{equation}\label{eq:t-moment_difference}
  \mathcal{A}_{\mathcal{C}}^{t}(\cdot)=\mathbb{E}_{U} U^{\otimes t}(\cdot)^{\otimes t}\left(U^{\dagger}\right)^{\otimes t}-\mathbb{E}_{\bm{\theta}} \mathcal{C}(\bm{\theta})^{\otimes t}(\cdot)^{\otimes t}\left(\mathcal{C}(\bm{\theta})^{\dagger}\right)^{\otimes t},
\end{equation}
where $U$ is drawn from Haar measure. In this work, we consider the case $t = 2$ and $\ketbra{0^n}$, and aim to evaluate $\norm{\mathcal{A}_{\mathcal{C}}^{2}(\ketbra{0^{n}})}_2^2$, where $\norm{\cdot}_2$ is the Hilbert-Schmidt norm. For simplicity, we denote $\mathcal{A}_{\mathcal{C}}^2 := \mathcal{A}_{\mathcal{C}}^2(\ket{0^n}\bra{0^n})$. The expressibility is then quantified by the following equation
\begin{equation}\label{eq:expressibility}
\begin{aligned}
  &\norm{\mathcal{A}_{\mathcal{C}}^{2}}_2^2 = \tr{\mathcal{A}_{\mathcal{C}}^{2} \left(\mathcal{A}_{\mathcal{C}}^{2}\right)^{\dagger}}\\
   =& \tr{\left[\mathbb{E}_{U}U^{\otimes 2}(\ketbra{0^{n}})^{\otimes 2}\left(U^{\dagger}\right)^{\otimes 2}\right]^2}-2\tr{\mathbb{E}_{U}U^{\otimes 2}(\ketbra{0^{n}})^{\otimes 2}\left(U^{\dagger}\right)^{\otimes 2}\mathbb{E}_{\bm{\theta}} \mathcal{C}(\bm{\theta})^{\otimes 2}(\ketbra{0^{n}})^{\otimes 2}\left(\mathcal{C}(\bm{\theta})^{\dagger}\right)^{\otimes 2}}\\
  & +  \tr{\left[\mathbb{E}_{\bm{\theta}} \mathcal{C}(\bm{\theta})^{\otimes 2}(\ketbra{0^{n}})^{\otimes 2}\left(\mathcal{C}(\bm{\theta})^{\dagger}\right)^{\otimes 2}\right]^2}.
  \end{aligned}  
\end{equation}
We first evaluate the first two terms. According to~\cite{PRXQuantum.3.010333}, for any Hermitian operator $\sigma \in \mathbb{C}^{q^2 \times q^2}$, it holds that
\begin{equation}\label{eq:2_moment_haar_random}
  \mathbb{E}_{U} U^{\otimes 2}\sigma\left(U^{\dagger}\right)^{\otimes 2} = \frac{\operatorname{Tr}(\sigma)-q^{-1} \operatorname{Tr}(\sigma S)}{q^2-1} I+\frac{\operatorname{Tr}(\sigma S)-q^{-1} \operatorname{Tr}(\sigma)}{q^2-1} S,
\end{equation}
where $S$ is the swap operator defined by $S\ket{\psi} \otimes \ket{\phi} = \ket{\phi} \otimes \ket{\psi}$ for all product states $\ket{\psi}, \ket{\phi}$. Taking $\sigma = (\ket{0^n} \bra{0^n})^{\otimes 2}$, we obtain
\begin{equation}\label{eq:2_design_output}
  \mathbb{E}_{U} U^{\otimes 2}(\ketbra{0^{n}})^{\otimes 2}\left(U^{\dagger}\right)^{\otimes 2} = \frac{I + S}{2^n \left(2^n+1\right)}.
\end{equation}
Substituting~\eqref{eq:2_design_output} into the first two terms of $\norm{\mathcal{A}_{\mathcal{C}}^{2}}_2^2$, we have
\begin{equation}\label{eq:expressibility_1}
  \begin{aligned}
    \tr{\left[\mathbb{E}_{U}U^{\otimes 2}(\ketbra{0^{n}})^{\otimes 2}\left(U^{\dagger}\right)^{\otimes 2}\right]^2} &= \frac{\tr{I + S}^2}{4^n \left(2^n+1\right)^2}\\
    & = \frac{1}{2^{n-1} \left(2^n+1\right)},
  \end{aligned}
\end{equation}

\begin{equation}\label{eq:expressibility_2}
  \begin{aligned}
    &\tr{\mathbb{E}_{U}U^{\otimes 2}(\ketbra{0^{n}})^{\otimes 2}\left(U^{\dagger}\right)^{\otimes 2}\mathbb{E}_{\bm{\theta}} \mathcal{C}(\bm{\theta})^{\otimes 2}(\ketbra{0^{n}})^{\otimes 2}\left(\mathcal{C}(\bm{\theta})^{\dagger}\right)^{\otimes 2}}\\
    =& \frac{1}{2^n \left(2^n+1\right)}\tr{\left(I + S \right)\mathbb{E}_{\bm{\theta}} \mathcal{C}(\bm{\theta})^{\otimes 2}(\ketbra{0^{n}})^{\otimes 2}\left(\mathcal{C}(\bm{\theta})^{\dagger}\right)^{\otimes 2}}\\
    =& \frac{1}{2^n \left(2^n+1\right)}\mathbb{E}_{\bm{\theta}}\tr{\left(I + S \right) \mathcal{C}(\bm{\theta})^{\otimes 2}(\ketbra{0^{n}})^{\otimes 2}\left(\mathcal{C}(\bm{\theta})^{\dagger}\right)^{\otimes 2}}\\
    =& \frac{2}{2^n \left(2^n+1\right)}\mathbb{E}_{\bm{\theta}}\tr{ \mathcal{C}(\bm{\theta})^{\otimes 2}(\ketbra{0^{n}})^{\otimes 2}\left(\mathcal{C}(\bm{\theta})^{\dagger}\right)^{\otimes 2}}\\
    =& \frac{1}{2^{n-1} \left(2^n+1\right)}.
  \end{aligned}
\end{equation}
We now turn to the third term of~\eqref{eq:expressibility}, which is
\begin{equation}\label{eq:expressibility_3}
  \begin{aligned}
    &\tr{\left[\mathbb{E}_{\bm{\theta}} \mathcal{C}(\bm{\theta})^{\otimes 2}(\ketbra{0^{n}})^{\otimes 2}\left(\mathcal{C}(\bm{\theta})^{\dagger}\right)^{\otimes 2}\right]^2}\\
   =&  \tr{\left[\mathbb{E}_{\bm{\theta} \in \{0,\frac{\pi}{2},\pi,\frac{3\pi}{2}\}^{N_g}}\mathcal{C}(\bm{\theta})^{\otimes 2}(\ketbra{0^{n}})^{\otimes 2}\left(\mathcal{C}(\bm{\theta})^{\dagger}\right)^{\otimes 2}\right]^2}\\
   =&\mathbb{E}_{\bm{\theta_1} \in \{0,\frac{\pi}{2},\pi,\frac{3\pi}{2}\}^{N_g}} \mathbb{E}_{\bm{\theta_2} \in \{0,\frac{\pi}{2},\pi,\frac{3\pi}{2}\}^{N_g}}\left|\bra{0^{n}}\left(\mathcal{C}(\bm{\theta_1})\right)^\dagger\mathcal{C}(\bm{\theta_2})\ket{0^{n}}\right|^4,
  \end{aligned}
\end{equation}
the last equality holds because ${\left\{R_P(\theta)\right\}}_{\theta=0, \pi/2, \pi, 3\pi/2}$ is a quantum rotation $2$-design. 
For situations for noisy circuit, this will be discussed in the next section.

\subsection{Expressibility of noisy PQC}
Due to the presence of noise, the original parameterized quantum circuit (PQC) is transformed into a noisy circuit, denoted by $\widetilde{\mathcal{C}}(\bm{\theta})$, as defined in Eq.~\eqref{eq:noisy_circuit}.
Similar to the Eq.~\eqref{eq:t-moment_difference}, we can define the difference between the  noisy circuit and the unitary $t$-design:
\begin{equation}\label{eq:noisy_t-moment_difference}
  \mathcal{A}_{\widetilde{\mathcal{C}}}^{t}(\cdot)=\mathbb{E}_{U} U^{\otimes t}(\cdot)^{\otimes t}\left(U^{\dagger}\right)^{\otimes t}-\mathbb{E}_{\bm{\theta}} \widetilde{\mathcal{C}}(\bm{\theta})(\cdot)^{\otimes t},
\end{equation}
where $\widetilde{\mathcal{C}}(\bm{\theta})$ is the noisy circuit.

Therefore, the $2$-order expressibility of noisy circuit $\widetilde{\mathcal{C}}(\bm{\theta})$ is given by:
\begin{equation}\label{eq:noisy_expressibility}
  \begin{aligned}
    &\norm{\mathbb{E}_{U} U^{\otimes 2}(\ketbra{0^n})^{\otimes 2}\left(U^{\dagger}\right)^{\otimes 2}-\mathbb{E}_{\bm{\theta}} \widetilde{\mathcal{C}}(\bm{\theta})(\ketbra{0^{n}})^{\otimes 2}}_2^2\\
     &= \tr{\left[\mathbb{E}_{U}U^{\otimes 2}(\ketbra{0^{n}})^{\otimes 2}\left(U^{\dagger}\right)^{\otimes 2}\right]^2}-2\tr{\mathbb{E}_{U}U^{\otimes 2}(\ketbra{0^{n}})^{\otimes 2}\left(U^{\dagger}\right)^{\otimes 2}\mathbb{E}_{\bm{\theta}} \widetilde{\mathcal{C}}(\bm{\theta})(\ketbra{0^{n}})^{\otimes 2}}\\
    & \quad +  \tr{\left[\mathbb{E}_{\bm{\theta}} \widetilde{\mathcal{C}}(\bm{\theta})(\ketbra{0^{n}})^{\otimes 2}\right]^2}.
  \end{aligned}  
\end{equation}
The first term is identical to that in~\eqref{eq:expressibility_1}, and we have 
\begin{equation}\label{eq:noisy_expressibility_1}
\tr{\left[\mathbb{E}_{U}U^{\otimes 2}(\ketbra{0^{n}})^{\otimes 2}\left(U^{\dagger}\right)^{\otimes 2}\right]^2} = \frac{1}{2^{n-1} (2^n + 1)}.
\end{equation}
We now evaluate the second and third terms in~\eqref{eq:noisy_expressibility} via Monte Carlo sampling. Recall from~\eqref{eq:expressibility_1} that:
\begin{equation}
  \mathbb{E}_{U} U^{\otimes 2}(\ketbra{0^{n}})^{\otimes 2}\left(U^{\dagger}\right)^{\otimes 2} = \frac{I + S}{2^n \left(2^n+1\right)}.
\end{equation}

Substituting this into the second term gives:
\begin{equation}\label{eq:noisy_expressibility_2}
  \begin{aligned}
    &\tr{\mathbb{E}_{U}U^{\otimes 2}(\ketbra{0^{n}})^{\otimes 2}\left(U^{\dagger}\right)^{\otimes 2}\mathbb{E}_{\bm{\theta}} \widetilde{\mathcal{C}}(\bm{\theta})(\ketbra{0^{n}})^{\otimes 2}}\\
    &= \frac{1}{2^n \left(2^n+1\right)}\tr{\left(I + S \right)\mathbb{E}_{\bm{\theta}} \widetilde{\mathcal{C}}(\bm{\theta})(\ketbra{0^{n}})^{\otimes 2}}\\
    &= \frac{1}{2^n \left(2^n+1\right)}\mathbb{E}_{\bm{\theta} \in \{0,\frac{\pi}{2},\pi,\frac{3\pi}{2}\}^{N_g}}\tr{\left(I + S \right) \widetilde{\mathcal{C}}(\bm{\theta})(\ketbra{0^{n}})^{\otimes 2}}\\
    &= \frac{1}{2^n \left(2^n+1\right)}+\frac{1}{2^n \left(2^n+1\right)}\mathbb{E}_{\bm{\theta} \in \{0,\frac{\pi}{2},\pi,\frac{3\pi}{2}\}^{N_g}}\tr{\left(\widetilde{\mathcal{C}}(\bm{\theta})(\ketbra{0^{n}})\right)^2}\\
    &= \frac{1}{2^n \left(2^n+1\right)}+\frac{1}{2^n \left(2^n+1\right)}\mathbb{E}_{\bm{\theta} \in \{0,\frac{\pi}{2},\pi,\frac{3\pi}{2}\}^{N_g}} \sum_{\sigma \in \{\sfrac{\mathbb{I}}{\sqrt{2}},\sfrac{X}{\sqrt{2}},\sfrac{Y}{\sqrt{2}},\sfrac{Z}{\sqrt{2}}\}^{\otimes n}} \tr{\widetilde{\mathcal{C}}(\bm{\theta})(\ketbra{0^{n}}) \sigma}^2 \\
    &= \frac{1}{2^n \left(2^n+1\right)}+\frac{2^n}{ \left(2^n+1\right)}\mathbb{E}_{\bm{\theta} \in \{0,\frac{\pi}{2},\pi,\frac{3\pi}{2}\}^{N_g}} \mathbb{E}_{\sigma \in \{\sfrac{\mathbb{I}}{\sqrt{2}},\sfrac{X}{\sqrt{2}},\sfrac{Y}{\sqrt{2}},\sfrac{Z}{\sqrt{2}}\}^{\otimes n}} \tr{\widetilde{\mathcal{C}}(\bm{\theta})(\ketbra{0^{n}}) \sigma}^2, \\
    &= \frac{1}{2^n \left(2^n+1\right)}+\frac{1}{ \left(2^n+1\right)}\mathbb{E}_{\bm{\theta} \in \{0,\frac{\pi}{2},\pi,\frac{3\pi}{2}\}^{N_g}} \mathbb{E}_{\sigma \in \{\mathbb{I},X,Y,Z\}^{\otimes n}} \tr{\widetilde{\mathcal{C}}(\bm{\theta})(\ketbra{0^{n}}) \sigma}^2, \\
  \end{aligned}
\end{equation}
where the second equality holds because $\{ R_P(\theta) \}_{\theta = 0, \pi/2, \pi, 3\pi/2}$ forms a single-qubit unitary 2-design, the third equality follows from the equation $\tr(S A \otimes B) = \tr(AB)$, and $\mathbb{E}_{\sigma}$ takes the uniform distribution $\mathbb{P}=\frac{1}{4^n}$ over all Pauli operators.

Similarly, we evaluate the last term in Eq.~\eqref{eq:noisy_expressibility}, we have:
\begin{equation}\label{eq:noisy_expressibility_3}
  \begin{aligned}
    &\tr{\left[\mathbb{E}_{\bm{\theta}} \widetilde{\mathcal{C}}(\bm{\theta})(\ketbra{0^{n}})^{\otimes 2}\right]^2}\\
    &=  \tr{\left[\mathbb{E}_{\bm{\theta} \in \{0,\frac{\pi}{2},\pi,\frac{3\pi}{2}\}^{N_g}}\widetilde{\mathcal{C}}(\bm{\theta})(\ketbra{0^{n}})^{\otimes 2}\right]^2}\\
    &=\mathbb{E}_{\bm{\theta_1} \in \{0,\frac{\pi}{2},\pi,\frac{3\pi}{2}\}^{N_g}} \mathbb{E}_{\bm{\theta_2} \in \{0,\frac{\pi}{2},\pi,\frac{3\pi}{2}\}^{N_g}} \tr{\prod_{i=1}^2\widetilde{\mathcal{C}}(\bm{\theta_i})(\ketbra{0^{n}})^{\otimes 2}}\\
    &=\mathbb{E}_{\bm{\theta_1} \in \{0,\frac{\pi}{2},\pi,\frac{3\pi}{2}\}^{N_g}} \mathbb{E}_{\bm{\theta_2} \in \{0,\frac{\pi}{2},\pi,\frac{3\pi}{2}\}^{N_g}} \tr{\prod_{i=1}^2 \widetilde{\mathcal{C}}(\bm{\theta_i})(\ketbra{0^{n}})}^2\\
    &=\mathbb{E}_{\bm{\theta_1} \in \{0,\frac{\pi}{2},\pi,\frac{3\pi}{2}\}^{N_g}} \mathbb{E}_{\bm{\theta_2} \in \{0,\frac{\pi}{2},\pi,\frac{3\pi}{2}\}^{N_g}} \tr{\widetilde{\mathcal{C}}(\bm{\theta_1})(\ketbra{0^{n}}) \widetilde{\mathcal{C}}(\bm{\theta_2})(\ketbra{0^{n}})}^2\\
    &=\mathbb{E}_{\bm{\theta_1} \in \{0,\frac{\pi}{2},\pi,\frac{3\pi}{2}\}^{N_g}} \mathbb{E}_{\bm{\theta_2} \in \{0,\frac{\pi}{2},\pi,\frac{3\pi}{2}\}^{N_g}} \tr{\widetilde{\mathcal{C}}^\dagger (\bm{\theta_2}) \left(\widetilde{\mathcal{C}}(\bm{\theta_1})(\ketbra{0^{n}}) \right) \ketbra{0^{n}}}^2\\
    &=\mathbb{E}_{\bm{\theta_1} \in \{0,\frac{\pi}{2},\pi,\frac{3\pi}{2}\}^{N_g}} \mathbb{E}_{\bm{\theta_2} \in \{0,\frac{\pi}{2},\pi,\frac{3\pi}{2}\}^{N_g}} \tr{\widetilde{\mathcal{C}}^\dagger (\bm{\theta_2}) \left(\widetilde{\mathcal{C}}(\bm{\theta_1})(\ketbra{0^{n}}) \right) \left(\frac{I+Z}{2}\right)^{\otimes{n}} }^2\\
    &=\mathbb{E}_{\bm{\theta_1} \in \{0,\frac{\pi}{2},\pi,\frac{3\pi}{2}\}^{N_g}} \mathbb{E}_{\bm{\theta_2} \in \{0,\frac{\pi}{2},\pi,\frac{3\pi}{2}\}^{N_g}} \left[\mathbb{E}_{\sigma \in \{I,Z\}^{\otimes n}} \tr{\widetilde{\mathcal{C}}^\dagger (\bm{\theta_2}) \left(\widetilde{\mathcal{C}}(\bm{\theta_1})(\ketbra{0^{n}}) \right) \sigma}\right]^2,
  \end{aligned}
\end{equation}
where the notation $\widetilde{\mathcal{C}}^\dagger (\bm{\theta_2})$ denotes the adjoint channel of $\widetilde{\mathcal{C}}(\bm{\theta_2})$, and the $\mathbb{E}_{\sigma}$ takes the uniform distribution $\mathbb{P}=\frac{1}{2^n}$ over all $\{I,Z\}^{\otimes n}$.
To efficiently estimate the term $\tr{\widetilde{\mathcal{C}}^\dagger (\bm{\theta_2}) \left(\widetilde{\mathcal{C}}(\bm{\theta_1})(\ketbra{0^{n}}) \right) \sigma}$, we need the additional assumption that the adjoint channel $\widetilde{\mathcal{C}}^\dagger (\bm{\theta_2})$ can be efficiently back-propagated.
This assumption is satisfied by the noisy circuit with the noise model is not only PCS1, but also Pauli row-wise sum at most one:
\begin{definition}{(Pauli row-wise sum at most one channel)}\label{def:PRS1}
  A quantum channel $\mathcal{E}$ is a Pauli row-wise sum at most one channel~(PCS1) if its Pauli transform matrix $\left(\mathcal{S}_\mathcal{E}\right)_{i,j}=\tr{\mathcal{E}(\sigma_i) \sigma_j}$ satisfies that the 1-norm $\norm{\cdot}_1$ of its each row is no more than $1$, where $\sigma_i$ and $\sigma_j$ are the $i$-th and $j$-th Pauli operator in the normalized Pauli group $\bm{P}_n$, respectively.
  This can be expressed as:
  \begin{equation}
    \sum_{j} \abs{\left(\mathcal{S}_\mathcal{E}\right)_{i,j}}\leq 1.
  \end{equation}

\end{definition}

If a channel is both PCS1 and PRS1, then the channel is unital, thus the amplitude damping and the thermal relaxation process are not valid.
But luckily, the Pauli error channels and the measurement-based feedback control channels with Pauli action are still PRS1.

Consequentially, combining Eq.~\eqref{eq:noisy_expressibility_1}, Eq.~\eqref{eq:noisy_expressibility_2}, and Eq.~\eqref{eq:noisy_expressibility_3} we have the following equation for the expressibility:
\begin{equation}\label{eq:noisy_expressibility_all}
  \begin{aligned}
    &\norm{\mathbb{E}_{U} U^{\otimes 2}(\ketbra{0^n})^{\otimes 2}\left(U^{\dagger}\right)^{\otimes 2}-\mathbb{E}_{\bm{\theta}} \widetilde{\mathcal{C}}(\bm{\theta})^{\otimes 2}(\cdot)^{\otimes 2}\left(\widetilde{\mathcal{C}}(\bm{\theta})^{\dagger}\right)^{\otimes 2}}_2^2\\
    &= \frac{1}{2^{n-1} (2^n + 1)} - 2\left( \frac{1}{2^n \left(2^n+1\right)}+\frac{1}{ \left(2^n+1\right)}\mathbb{E}_{\bm{\theta} \in \{0,\frac{\pi}{2},\pi,\frac{3\pi}{2}\}^{N_g}} \mathbb{E}_{\sigma \in \{\mathbb{I},X,Y,Z\}^{\otimes n}} \tr{\widetilde{\mathcal{C}}(\bm{\theta})(\ketbra{0^{n}}) \sigma}^2 \right)\\
    & \quad + \mathbb{E}_{\bm{\theta_1} \in \{0,\frac{\pi}{2},\pi,\frac{3\pi}{2}\}^{N_g}} \mathbb{E}_{\bm{\theta_2} \in \{0,\frac{\pi}{2},\pi,\frac{3\pi}{2}\}^{N_g}} \left[\mathbb{E}_{\sigma \in \{I,Z\}^{\otimes n}} \tr{\widetilde{\mathcal{C}}^\dagger (\bm{\theta_2}) \left(\widetilde{\mathcal{C}}(\bm{\theta_1})(\ketbra{0^{n}}) \right) \sigma}\right]^2 \\
    & = \mathbb{E}_{\bm{\theta_1} \in \{0,\frac{\pi}{2},\pi,\frac{3\pi}{2}\}^{N_g}} \mathbb{E}_{\bm{\theta_2} \in \{0,\frac{\pi}{2},\pi,\frac{3\pi}{2}\}^{N_g}} \left[\mathbb{E}_{\sigma \in \{I,Z\}^{\otimes n}} \tr{\widetilde{\mathcal{C}}^\dagger (\bm{\theta_2}) \left(\widetilde{\mathcal{C}}(\bm{\theta_1})(\ketbra{0^{n}}) \right) \sigma}\right]^2 \\
    & \quad - \frac{2}{ \left(2^n+1\right)}\mathbb{E}_{\bm{\theta} \in \{0,\frac{\pi}{2},\pi,\frac{3\pi}{2}\}^{N_g}} \mathbb{E}_{\sigma \in \{\mathbb{I},X,Y,Z\}^{\otimes n}} \tr{\widetilde{\mathcal{C}}(\bm{\theta})(\ketbra{0^{n}}) \sigma}^2.
  \end{aligned}
\end{equation}

Therefore, rewriting the expressibility in Eq.~\eqref{eq:noisy_expressibility_all}, we have:
\begin{small}
  \begin{equation}\label{ap:eq:noisy_expressibility}
    \begin{aligned}
      \norm{\mathbb{E}_{U} U^{\otimes 2}(\ketbra{0^n})^{\otimes 2}\left(U^{\dagger}\right)^{\otimes 2}-\mathbb{E}_{\bm{\theta}} \widetilde{\mathcal{C}}(\bm{\theta})^{\otimes 2}(\cdot)^{\otimes 2}\left(\widetilde{\mathcal{C}}(\bm{\theta})^{\dagger}\right)^{\otimes 2}}_2^2 =
   \mathbb{E}_{\bm{\theta_1}, \bm{\theta_2}} \left[\mathbb{E}_{\sigma} h_{\mathrm{express}}(\bm{\theta_1}, \bm{\theta_2},\sigma)\right]^2 - \frac{2}{ \left(2^n+1\right)} \mathbb{E}_{\bm{\theta},\sigma'} h'_{\mathrm{express}}(\bm{\theta},\sigma')\\
   = 
   \mathbb{E}_{\bm{\theta_1}, \bm{\theta_2}}\mathbb{E}_{\sigma_1,\sigma_2} h_{\mathrm{express}}(\bm{\theta_1}, \bm{\theta_2},\sigma_1)h_{\mathrm{express}}(\bm{\theta_1}, \bm{\theta_2},\sigma_2) - \frac{2}{ \left(2^n+1\right)} \mathbb{E}_{\bm{\theta},\sigma'} h'_{\mathrm{express}}(\bm{\theta},\sigma')
    \end{aligned}
\end{equation}
\end{small}
where $\bm{\theta}, \bm{\theta_1}, \bm{\theta_2}$ takes values from $\{0,\frac{\pi}{2},\pi,\frac{3\pi}{2}\}^{N_g}$ with uniform distribution, $\sigma, \sigma_1, \sigma_2$ take from $\{I,Z\}^{\otimes n}$ and $\sigma'$ takes from $\{\mathbb{I},X,Y,Z\}^{\otimes n}$ with uniform distribution, respectively.
Here we define the functions:
\begin{equation}
  h_{\mathrm{express}}(\bm{\theta_1}, \bm{\theta_2}, \sigma) =  \tr{\widetilde{\mathcal{C}}^\dagger (\bm{\theta_2}) \left(\widetilde{\mathcal{C}}(\bm{\theta_1})(\ketbra{0^{n}}) \right) \sigma},
\end{equation}
and
\begin{equation}
  h'_{\mathrm{express}}(\bm{\theta}, \sigma') = \tr{\widetilde{\mathcal{C}}(\bm{\theta})(\ketbra{0^{n}}) \sigma'}^2.
\end{equation}

Because, when $\bm{\theta_1},\bm{\theta_2}\in \{0,\frac{\pi}{2},\pi,\frac{3\pi}{2}\}^{N_g}$, $\sigma\in \{I,Z\}^{\otimes n}$ and $\sigma' \in \{\mathbb{I},X,Y,Z\}^{\otimes n}$, both the terms $\tr{\widetilde{\mathcal{C}}^\dagger (\bm{\theta_2}) \left(\widetilde{\mathcal{C}}(\bm{\theta_1})(\ketbra{0^{n}}) \right) \sigma}$ and $\tr{\widetilde{\mathcal{C}}(\bm{\theta})(\ketbra{0^{n}}) \sigma'}^2$ have a bound of $\norm{\sigma}_\infty=\norm{\sigma'}_\infty^2=1$, therefore there are $\abs{h_{\mathrm{express}}(\bm{\theta_1}, \bm{\theta_2}, \sigma)} \leq 1$ and $\abs{h'_{\mathrm{express}}(\bm{\theta}, \sigma')} \leq 1$.

\subsection{Lower bound for expressibility}
For channels that are PCS1 but not PRS1, we can compute a tractable lower bound on the expressibility.
The obstruction caused by non-PRS1 channels is that we cannot efficiently sample the bounded terms $\left[\mathbb{E}_{\sigma \in \{I,Z\}^{\otimes n}} \tr{\widetilde{\mathcal{C}}^\dagger (\bm{\theta_2}) \left(\widetilde{\mathcal{C}}(\bm{\theta_1})(\ketbra{0^{n}}) \right) \sigma}\right]^2$ in Eq.~\eqref{eq:noisy_expressibility_3}.

We now show that the last term in Eq.~\eqref{eq:noisy_expressibility} admits the following lower bound:
\begin{equation}\label{eq:noisy_expressibility_3_lower_bound}
  \begin{aligned}
    &\tr{\left[\mathbb{E}_{\bm{\theta}} \widetilde{\mathcal{C}}(\bm{\theta})(\ketbra{0^{n}})^{\otimes 2}\right]^2}\\
    &=  \tr{\left[\mathbb{E}_{\bm{\theta} \in \{0,\frac{\pi}{2},\pi,\frac{3\pi}{2}\}^{N_g}}\widetilde{\mathcal{C}}(\bm{\theta})(\ketbra{0^{n}})^{\otimes 2}\right]^2}\\
    &= \sum_{P \in \{\mathbb{I}/ \sqrt{2}, X/ \sqrt{2}, Y/ \sqrt{2}, Z/ \sqrt{2}\}^{\otimes 2n}} \tr{\left[\mathbb{E}_{\bm{\theta} \in \{0,\frac{\pi}{2},\pi,\frac{3\pi}{2}\}^{N_g}}\widetilde{\mathcal{C}}(\bm{\theta})(\ketbra{0^{n}})^{\otimes 2}\right] P}^2\\
    & \geq \sum_{P' \in \{\mathbb{I}/ \sqrt{2}, X/ \sqrt{2}, Y/ \sqrt{2}, Z/ \sqrt{2}\}^{\otimes n}} \tr{\left[\mathbb{E}_{\bm{\theta} \in \{0,\frac{\pi}{2},\pi,\frac{3\pi}{2}\}^{N_g}}\widetilde{\mathcal{C}}(\bm{\theta})(\ketbra{0^{n}})^{\otimes 2}\right] P'^{ \otimes 2} }^2\\
    &= \sum_{P' \in \{\mathbb{I}/ \sqrt{2}, X/ \sqrt{2}, Y/ \sqrt{2}, Z/ \sqrt{2}\}^{\otimes n}} \tr{\mathbb{E}_{\bm{\theta} \in \{0,\frac{\pi}{2},\pi,\frac{3\pi}{2}\}^{N_g}}\widetilde{\mathcal{C}}(\bm{\theta})(\ketbra{0^{n}})^{\otimes 2} P'^{ \otimes 2} }^2\\
    &= \sum_{P' \in \{\mathbb{I}/ \sqrt{2}, X/ \sqrt{2}, Y/ \sqrt{2}, Z/ \sqrt{2}\}^{\otimes n}} \left\{\mathbb{E}_{\bm{\theta} \in \{0,\frac{\pi}{2},\pi,\frac{3\pi}{2}\}^{N_g}} \tr{\widetilde{\mathcal{C}}(\bm{\theta})(\ketbra{0^{n}})^{\otimes 2} P'^{ \otimes 2} }\right\}^2\\
    &= \sum_{P' \in \{\mathbb{I}/ \sqrt{2}, X/ \sqrt{2}, Y/ \sqrt{2}, Z/ \sqrt{2}\}^{\otimes n}} \left\{\mathbb{E}_{\bm{\theta} \in \{0,\frac{\pi}{2},\pi,\frac{3\pi}{2}\}^{N_g}} \tr{\widetilde{\mathcal{C}}(\bm{\theta})(\ketbra{0^{n}}) P' }^2 \right\}^2\\
    &= \frac{1}{4^n} \sum_{\sigma \in \{\mathbb{I}, X, Y, Z\}^{\otimes n}} \left\{\mathbb{E}_{\bm{\theta} \in \{0,\frac{\pi}{2},\pi,\frac{3\pi}{2}\}^{N_g}} \tr{\widetilde{\mathcal{C}}(\bm{\theta})(\ketbra{0^{n}}) \sigma }^2 \right\}^2\\
    &= \mathbb{E}_{\sigma \in \{\mathbb{I},X,Y,Z\}^{\otimes n}} \left\{\mathbb{E}_{\bm{\theta} \in \{0,\frac{\pi}{2},\pi,\frac{3\pi}{2}\}^{N_g}} \tr{\widetilde{\mathcal{C}}(\bm{\theta})(\ketbra{0^{n}}) \sigma }^2 \right\}^2,
  \end{aligned}
\end{equation}
where in the second equality we expand in the normalized Pauli basis $\{\mathbb{I}/ \sqrt{2}, X/ \sqrt{2}, Y/ \sqrt{2}, Z/ \sqrt{2}\}^{\otimes 2n}$, and the inequality follows by restricting to the subset $\left\{P' \otimes P' \mid P' \in \{\mathbb{I}/ \sqrt{2}, X/ \sqrt{2}, Y/ \sqrt{2}, Z/ \sqrt{2}\}^{\otimes n}\right\}$ $\subseteq$ $\left\{P  \mid P \in \{\mathbb{I}/ \sqrt{2}, X/ \sqrt{2}, Y/ \sqrt{2}, Z/ \sqrt{2}\}^{\otimes 2n}\right\}$.

Consequentially, combining Eq.~\eqref{eq:noisy_expressibility_1}, Eq.~\eqref{eq:noisy_expressibility_2}, and Eq.~\eqref{eq:noisy_expressibility_3_lower_bound} we obtain the following lower bound for the expressibility:
\begin{equation}\label{eq:noisy_expressibility_lower_bound}
  \begin{aligned}
    &\norm{\mathbb{E}_{U} U^{\otimes 2}(\ketbra{0^n})^{\otimes 2}\left(U^{\dagger}\right)^{\otimes 2}-\mathbb{E}_{\bm{\theta}} \widetilde{\mathcal{C}}(\bm{\theta})^{\otimes 2}(\cdot)^{\otimes 2}\left(\widetilde{\mathcal{C}}(\bm{\theta})^{\dagger}\right)^{\otimes 2}}_2^2\\
    & \geq \frac{1}{2^{n-1} (2^n + 1)} - 2\left( \frac{1}{2^n \left(2^n+1\right)}+\frac{1}{ \left(2^n+1\right)}\mathbb{E}_{\bm{\theta} \in \{0,\frac{\pi}{2},\pi,\frac{3\pi}{2}\}^{N_g}} \mathbb{E}_{\sigma \in \{\mathbb{I},X,Y,Z\}^{\otimes n}} \tr{\widetilde{\mathcal{C}}(\bm{\theta})(\ketbra{0^{n}}) \sigma}^2 \right)\\
    & \quad + \mathbb{E}_{\sigma \in \{\mathbb{I},X,Y,Z\}^{\otimes n}} \left\{\mathbb{E}_{\bm{\theta} \in \{0,\frac{\pi}{2},\pi,\frac{3\pi}{2}\}^{N_g}} \tr{\widetilde{\mathcal{C}}(\bm{\theta})(\ketbra{0^{n}}) \sigma }^2 \right\}^2 \\
    & = \mathbb{E}_{\sigma \in \{\mathbb{I},X,Y,Z\}^{\otimes n}} \left[ \left\{\mathbb{E}_{\bm{\theta} \in \{0,\frac{\pi}{2},\pi,\frac{3\pi}{2}\}^{N_g}} \tr{\widetilde{\mathcal{C}}(\bm{\theta})(\ketbra{0^{n}}) \sigma }^2 \right\}^2  - \frac{2}{ \left(2^n+1\right)}\mathbb{E}_{\bm{\theta} \in \{0,\frac{\pi}{2},\pi,\frac{3\pi}{2}\}^{N_g}} \tr{\widetilde{\mathcal{C}}(\bm{\theta})(\ketbra{0^{n}}) \sigma}^2\right] \\
  \end{aligned}
\end{equation}

Crucially, the bound in Eq.~\eqref{eq:noisy_expressibility_lower_bound} does not involve the adjoint channel; hence it can be applied to circuits withchannels that are PCS1 but not PRS1.

Rewriting the lower bound in Eq.~\eqref{eq:noisy_expressibility_lower_bound}, we have:
\begin{equation}\label{ap:eq:noisy_expressibility_lower_bound}
  \norm{\mathbb{E}_{U} U^{\otimes 2}(\ketbra{0^n})^{\otimes 2}\left(U^{\dagger}\right)^{\otimes 2}-\mathbb{E}_{\bm{\theta}} \widetilde{\mathcal{C}}(\bm{\theta})^{\otimes 2}(\cdot)^{\otimes 2}\left(\widetilde{\mathcal{C}}(\bm{\theta})^{\dagger}\right)^{\otimes 2}}_2^2 \geq
   \mathbb{E}_{\sigma'} \left[ \left\{ \mathbb{E}_{\bm{\theta}} h_{\mathrm{expressibility}}^{\mathrm{lower}}(\bm{\theta},\sigma') \right\}^2 -\frac{2}{2^n+1}\mathbb{E}_{\bm{\theta}} h_{\mathrm{expressibility}}^{\mathrm{lower}}(\bm{\theta},\sigma') \right],
\end{equation}
where $\bm{\theta}$ takes values from $\{0,\frac{\pi}{2},\pi,\frac{3\pi}{2}\}^{N_g}$ with uniform distribution, $\sigma'$ takes values from $\{\mathbb{I},X,Y,Z\}^{\otimes n}$ with uniform distribution, and $h_{\mathrm{expressibility}}^{\mathrm{lower}}(\bm{\theta},\sigma') = \tr{\widetilde{\mathcal{C}}(\bm{\theta})(\ketbra{0^{n}}) \sigma' }^2$ with $\abs{h_{\mathrm{expressibility}}^{\mathrm{lower}}(\bm{\theta},\sigma')} \leq 1$.

The $\mathbb{E}_{\sigma'} \left[ \left\{ \mathbb{E}_{\bm{\theta}} h_{\mathrm{expressibility}}^{\mathrm{lower}}(\bm{\theta},\sigma') \right\}^2 -\frac{2}{2^n+1}\mathbb{E}_{\bm{\theta}} h_{\mathrm{expressibility}}^{\mathrm{lower}}(\bm{\theta},\sigma') \right]$ can be further reformulated to match the form of $\mathbb{E}_{\bm{\theta}} h_{*}(\cdot)$:
\begin{equation}
  \begin{aligned}
    &\mathbb{E}_{\sigma'} \left[ \left\{ \mathbb{E}_{\bm{\theta}} h_{\mathrm{expressibility}}^{\mathrm{lower}}(\bm{\theta},\sigma') \right\}^2 -\frac{2}{2^n+1}\mathbb{E}_{\bm{\theta}} h_{\mathrm{expressibility}}^{\mathrm{lower}}(\bm{\theta},\sigma') \right] \\
    =& \mathbb{E}_{\bm{\theta_1},\bm{\theta_2}} \mathbb{E}_{\sigma'} h_{\mathrm{expressibility}}^{\mathrm{lower}}(\bm{\theta_1},\sigma') h_{\mathrm{expressibility}}^{\mathrm{lower}}(\bm{\theta_2},\sigma') - \frac{2}{2^n+1} \mathbb{E}_{\bm{\theta}} \mathbb{E}_{\sigma'} h_{\mathrm{expressibility}}^{\mathrm{lower}}(\bm{\theta},\sigma').
  \end{aligned}
\end{equation}

\subsection{$l_1$-Lower bound for expressibility}

Except for $\norm{\mathbb{E}_{U} U^{\otimes 2}(\ketbra{0^n})^{\otimes 2}\left(U^{\dagger}\right)^{\otimes 2}-\mathbb{E}_{\bm{\theta}} \widetilde{\mathcal{C}}(\bm{\theta})(\ketbra{0^{n}})^{\otimes 2}}_2^2$, we can also derive a computable lower bound for $\norm{\mathbb{E}_{U} U^{\otimes 2}(\ketbra{0^n})^{\otimes 2}\left(U^{\dagger}\right)^{\otimes 2}-\mathbb{E}_{\bm{\theta}} \widetilde{\mathcal{C}}(\bm{\theta})(\ketbra{0^{n}})^{\otimes 2}}_1$.
Depending on the $l_1$-lower bound, we can obtain similar conclusions at Ref.~\cite{holmes2022connecting}, i.e., the existance of the strong expressibility-induced barren plateau phenomenon.

First, by Holder's inequality, we have:
\begin{small}
\begin{equation}
  \begin{aligned}
  \tr{\left(\mathbb{E}_{\bm{\theta}} \widetilde{\mathcal{C}}(\bm{\theta})(\ketbra{0^{n}})^{\otimes 2} - \mathbb{E}_{U} U^{\otimes 2}(\ketbra{0^n})^{\otimes 2}\left(U^{\dagger}\right)^{\otimes 2}\right) O^{\otimes 2}} &\leq \norm{O^{\otimes 2}}_\infty \norm{\mathbb{E}_{U} U^{\otimes 2}(\ketbra{0^n})^{\otimes 2}\left(U^{\dagger}\right)^{\otimes 2}-\mathbb{E}_{\bm{\theta}} \widetilde{\mathcal{C}}(\bm{\theta})(\ketbra{0^{n}})^{\otimes 2}}_1 \\
  & = \norm{O}^2_\infty \norm{\mathbb{E}_{U} U^{\otimes 2}(\ketbra{0^n})^{\otimes 2}\left(U^{\dagger}\right)^{\otimes 2}-\mathbb{E}_{\bm{\theta}} \widetilde{\mathcal{C}}(\bm{\theta})(\ketbra{0^{n}})^{\otimes 2}}_1.
  \end{aligned}
\end{equation}
\end{small}

On the other hand, for the variance of expectation value, we have:
\begin{equation}
  \begin{aligned}
  \mathrm{Var}\left( \langle \widetilde{O} \rangle_{\bm{\theta}} \right) & = \mathbb{E}_{\bm{\theta}} \left(  \langle \widetilde{O} \rangle_{\bm{\theta}} \right)^2 - \left( \mathbb{E}_{\bm{\theta}}  \langle \widetilde{O} \rangle_{\bm{\theta}} \right)^2 \\
  & = \mathbb{E}_{\bm{\theta}}   \tr{\widetilde{\mathcal{C}}(\bm{\theta})(\ketbra{0^{n}}) O}^2 - \left( \mathbb{E}_{\bm{\theta}}  \langle \widetilde{O} \rangle_{\bm{\theta}} \right)^2 \\
  & = \mathbb{E}_{\bm{\theta}}   \tr{\widetilde{\mathcal{C}}(\bm{\theta})(\ketbra{0^{n}})^{\otimes 2} O^{\otimes 2}} - \left( \mathbb{E}_{\bm{\theta}}  \langle \widetilde{O} \rangle_{\bm{\theta}} \right)^2.
  \end{aligned}
\end{equation}

Combining the Eq.~\eqref{eq:2_design_output} that $\mathbb{E}_{U} U^{\otimes 2}(\ketbra{0^{n}})^{\otimes 2}\left(U^{\dagger}\right)^{\otimes 2} = \frac{I + S}{2^n \left(2^n+1\right)}$ and the traceless assumption of $O$, we obtain:
\begin{equation}
  \tr {\mathbb{E}_{U} U^{\otimes 2}(\ketbra{0^n})^{\otimes 2}\left(U^{\dagger}\right)^{\otimes 2} O^{\otimes 2}} = \frac{\tr{(I + S) O^{\otimes 2}} }{2^n \left(2^n+1\right)} = \frac{\tr{O^2}+\tr{O}^2}{2^n \left(2^n+1\right)} = \frac{\tr{O^2}}{2^n \left(2^n+1\right)}.
\end{equation}

As a result, we have:
\begin{equation}\label{eq:l1_lower_bound}
  \begin{aligned}
    \norm{\mathbb{E}_{U} U^{\otimes 2}(\ketbra{0^n})^{\otimes 2}\left(U^{\dagger}\right)^{\otimes 2}-\mathbb{E}_{\bm{\theta}} \widetilde{\mathcal{C}}(\bm{\theta})(\ketbra{0^{n}})^{\otimes 2}}_1 &\geq \frac{\tr{\left(\mathbb{E}_{\bm{\theta}} \widetilde{\mathcal{C}}(\bm{\theta})(\ketbra{0^{n}})^{\otimes 2} - \mathbb{E}_{U} U^{\otimes 2}(\ketbra{0^n})^{\otimes 2}\left(U^{\dagger}\right)^{\otimes 2}\right) O^{\otimes 2}}}{\norm{O}^2_\infty} \\
    & = \frac{1}{\norm{O}^2_\infty} \left[ \mathrm{Var}\left( \langle \widetilde{O} \rangle_{\bm{\theta}} \right) + \left( \mathbb{E}_{\bm{\theta}}  \langle \widetilde{O} \rangle_{\bm{\theta}} \right)^2 - \frac{\tr{O^2}+\tr{O}^2}{2^n \left(2^n+1\right)} \right] \\
    & \geq \frac{1}{\norm{O}^2_\infty} \left[ \mathrm{Var}\left( \langle \widetilde{O} \rangle_{\bm{\theta}} \right) - \frac{\tr{O^2}}{2^n \left(2^n+1\right)} \right].
  \end{aligned}
\end{equation}
In Eq.~\eqref{eq:l1_lower_bound}, the observable $O$ can be any traceless observable, different $O$ will lead to different $l_1$-lower bound for the expressibility.
The variance $\mathrm{Var}\left( \langle \widetilde{O} \rangle_{\bm{\theta}} \right)$ then can be computed efficiently using the method in Sec.~\ref{sec:OBPPP}, so the $l_1$-lower bound for the expressibility is computable.

\subsection{Connection between expressibility and barren plateau}\label{sec:expressibility_barren_plateau}
Theoretically, we can use Eq.~\eqref{eq:l1_lower_bound} to connect the expressibility and the barren plateau phenomenon.
For Pauli observable $O$, we have $\tr{O^2}=2^n$ and $\norm{O}^2_\infty=1$.
If $\norm{\mathbb{E}_{U} U^{\otimes 2}(\ketbra{0^n})^{\otimes 2}\left(U^{\dagger}\right)^{\otimes 2}-\mathbb{E}_{\bm{\theta}} \widetilde{\mathcal{C}}(\bm{\theta})(\ketbra{0^{n}})^{\otimes 2}}_1$ is vanishing, i.e., $\norm{\mathbb{E}_{U} U^{\otimes 2}(\ketbra{0^n})^{\otimes 2}\left(U^{\dagger}\right)^{\otimes 2}-\mathbb{E}_{\bm{\theta}} \widetilde{\mathcal{C}}(\bm{\theta})(\ketbra{0^{n}})^{\otimes 2}}_1=\order{\frac{1}{2^n}}$, then $\mathrm{Var}\left( \langle \widetilde{O} \rangle_{\bm{\theta}} \right)$ is vanishing.
On the other hand, there is:
\begin{equation}
  \begin{aligned}
  \mathrm{Var}\left(\frac{\partial \langle \widetilde{O} \rangle}{\partial{\theta_i}}\right) &= \mathbb{E}_{\bm{\theta}} \left(\frac{ \langle \widetilde{O} \rangle_{\bm{\theta}+\frac{\pi}{2}e_i}-\langle \widetilde{O} \rangle_{\bm{\theta}-\frac{\pi}{2}e_i}}{2} \right)^2 \\
  &= \frac{1}{4}\mathbb{E}_{\bm{\theta}} \left[ \left[\langle \widetilde{O} \rangle_{\bm{\theta}+\frac{\pi}{2}e_i} - \left(\mathbb{E}_{\bm{\theta}}\langle \widetilde{O} \rangle_{\bm{\theta}}\right)\right]
  - \left[ \langle \widetilde{O} \rangle_{\bm{\theta}-\frac{\pi}{2}e_i}  - \left(\mathbb{E}_{\bm{\theta}}\langle \widetilde{O} \rangle_{\bm{\theta}}\right) \right] \right]^2 \\
  & \leq \frac{1}{2}\mathbb{E}_{\bm{\theta}}  \left\{ \left[\langle \widetilde{O} \rangle_{\bm{\theta}+\frac{\pi}{2}e_i} - \left(\mathbb{E}_{\bm{\theta}}\langle \widetilde{O} \rangle_{\bm{\theta}}\right)\right]^2
  + \left[ \langle \widetilde{O} \rangle_{\bm{\theta}-\frac{\pi}{2}e_i}  - \left(\mathbb{E}_{\bm{\theta}}\langle \widetilde{O} \rangle_{\bm{\theta}}\right) \right]^2 \right\} \\
  & = \mathrm{Var}\left( \langle \widetilde{O} \rangle_{\bm{\theta}} \right).
  \end{aligned}
\end{equation}
Thus, if $\mathrm{Var}\left( \langle \widetilde{O} \rangle_{\bm{\theta}} \right)$ is vanishing, then $\mathrm{Var}\left(\frac{\partial \langle \widetilde{O} \rangle}{\partial{\theta_i}}\right)$ is vanishing, this implies that the strong expressibility can induce the barren plateau phenomenon.

\section{Estimation algorithm}\label{sec:OBPPP}

\subsection{2MC-OBPPP}

In this section, we introduce an efficient classical algorithm to estimate diagnostic quantities for noisy parameterized quantum circuits with PCS1 noise channels.
The algorithm is inspired by the OBPPP algorithm proposed in Ref.~\cite{shao2024simulating} for simulating noisy quantum circuits.
First, we list the diagnostic quantities that we aim to estimate, including noise robustness, noise sensitivity, trainability, expressivity and its lower bound in Table~\ref{tab:table_diagnostic}.

\begin{table*}[h!]
\centering
\begin{small}
\begin{tabular}{|p{0.11\textwidth}||p{0.15\textwidth}|p{0.20\textwidth}|p{0.13\textwidth}|p{0.225\textwidth}|p{0.19\textwidth}|}
\hline
Diagnostic Quantities & Noise robustness & Noise sensitivity & Trainability & Expressivity & Expressivity Lower bound \\ \hline
Definition & $\mathbb{E}_{\bm{\theta}} \left( \langle O \rangle_{\bm{\theta}} - \langle \widetilde{O} \rangle_{\bm{\theta}} \right)^2$ & $\abs{\frac{\partial}{\partial \lambda_{i,j}} \mathbb{E}_{\bm{\theta}} \left( \langle O \rangle_{\bm{\theta}} - \langle \widetilde{O} \rangle_{\bm{\theta}} \right)^2}$ & $\mathrm{Var}\left(\frac{\partial \langle O \rangle}{\partial{\theta_i}}\right)$ & $\norm{\mathcal{A}_{\widetilde{\mathcal{C}}}^{2}(\ketbra{0}^{\otimes n})}_2^2$ & Eq.~\eqref{eq:noisy_expressibility_lower_bound} \\ \hline
Reformulation & $\mathbb{E}^*_{\bm{\theta}} h_{\mathrm{robustness}}(\bm{\theta})$ \eqref{ap:eq:mse_noise_robustness} & $\abs{\mathbb{E}^*_{\bm{\theta}} h_{\mathrm{sensitivity}}(\bm{\theta})}$ \eqref{ap:eq:noise_sensitivity} for Pauli noise and \eqref{ap:eq:amplitude_damping_noise_sensitivity} for amplitude damping noise & $\mathbb{E}^*_{\bm{\theta}} h_{\mathrm{trainability}}(\bm{\theta})$ \eqref{ap:eq:variance_gradient_noisy} & $\mathbb{E}^*_{\bm{\theta_1}, \bm{\theta_2}} \left[\mathbb{E}_{\sigma} h_{\mathrm{express}}(\bm{\theta_1}, \bm{\theta_2},\sigma)\right]^2$ $- \frac{2}{ \left(2^n+1\right)} \mathbb{E}^*_{\bm{\theta},\sigma'} h'_{\mathrm{express}}(\bm{\theta},\sigma')$ \eqref{ap:eq:noisy_expressibility} & $\mathbb{E}_{\sigma'} \left[ \left\{ \mathbb{E}^*_{\bm{\theta}} h_{\mathrm{express}}^{\mathrm{lower}}(\bm{\theta},\sigma') \right\}^2 \right. $ $-\left. \frac{2}{2^n+1}\mathbb{E}^*_{\bm{\theta}} h_{\mathrm{express}}^{\mathrm{lower}}(\bm{\theta},\sigma') \right]$ \eqref{ap:eq:noisy_expressibility_lower_bound} \\ \hline
Upper bound for $\abs{h_{*}}$ & $4\norm{O}_\infty^2$ & $8 \norm{O}_\infty^2$ & $\norm{O}_\infty^2$ & 1 & 1 \\ \hline
Constraint & PCS1 channels & PCS1 channels & PCS1 channels & PCS1 + PRS1 channels & PCS1 channels \\ \hline
Runtime for a inner-layer sample & $\order{nL}$ & $\order{nL}$ & $\order{nL}$ & $\order{nL}$ & $\order{nL}$ \\ \hline
\end{tabular}
\end{small}
\caption{A summary of the diagnostic quantities considered in this work. 
``Reformulation'' show how each diagnostic quantity can be expressed as an expectation value over a bounded function $h_{*}$.
In the reformulation expressions, $\mathbb{E}^*_{\bm{\theta}}$ represents the expectation value that $\bm{\theta}$, $\bm{\theta_1}$, and $\bm{\theta_2}$ are sampled from the discrete set $\{0,\frac{\pi}{2},\pi,\frac{3\pi}{2}\}^{N_g}$ with equal probability, where $N_g$ is the number of parameterized gates in the circuit.
$\sigma$ and $\sigma'$ are Pauli operators sampled from the sets $\{\mathbb{I},Z\}^{\otimes n}$ and $\{\mathbb{I},X,Y,Z\}^{\otimes n}$ with equal probability, respectively.
Here, $h_{*}$ represents the corresponding function for each diagnostic quantity, $n$ is the number of qubits, and $L$ is depth of the quantum circuit.}
\label{tab:table_diagnostic}
\end{table*}

As we briefly illustrated in main text, the 2MC-OBPPP algorithm has two layers, the outer-layer is used to sample the parameters $\bm{\theta}$ and other configurations (i.e., the $\sigma$ in expressibility and its lower bound), while the inner-layer is used to estimate the bounded function $h_{*}(\cdot)$ for each sampled configuration.

\textbf{Outer-layer:}
In the outer-layer, we sample the parameters $\bm{\theta}$ and any additional configurations according to their respective distributions. 
For each sampled configuration, we call the inner-layer to estimate the bounded function $h_{*}(\cdot)$.
For example, for noise robustness, noise sensitivity, and trainability, the parameters $\bm{\theta}$ are sampled with equal probability from the discrete sets $\{0,\frac{\pi}{2},\pi,\frac{3\pi}{2}\}^{N_g}$, where $N_g$ is the number of parameterized gates in the circuit.
The corresponding diagnostic quantity is then estimated by averaging the inner-layer outputs over $N_{\text{out}}$ outer-layer samples:
\begin{equation}
  \mathbb{E}^*_{\bm{\theta}} h_{*}(\bm{\theta}) \approx \frac{1}{N_{\text{out}}} \sum_{i=1}^{N_{\text{out}}} \widehat{h}_{*}(\bm{\theta}_i),
\end{equation}
where $N_{\text{out}}$ is the number of outer-layer samples, $\bm{\theta}_i$ is the $i$-th sampled parameter configuration, and $\widehat{h}_{*}(\bm{\theta}_i)$ is the corresponding inner-layer estimate.

For expressibility and its lower bound, the parameters $\bm{\theta}$, $\bm{\theta_1}$, and $\bm{\theta_2}$ in Table~\ref{tab:table_diagnostic} are also sampled with equal probability from the discrete sets $\{0,\frac{\pi}{2},\pi,\frac{3\pi}{2}\}^{N_g}$, while $\sigma$ and $\sigma'$ are sampled with equal probability from the set $\{\mathbb{I},Z\}^{\otimes n}$ and $\{\mathbb{I},X,Y,Z\}^{\otimes n}$, respectively.
The target diagnostic quantities are estimated via a cross-product estimator:
\begin{equation}
  \mathbb{E}^*_{\bm{\theta_1}, \bm{\theta_2}} \left[\mathbb{E}_{\sigma} h_{\mathrm{express}}(\bm{\theta_1}, \bm{\theta_2},\sigma)\right]^2 \approx \frac{1}{N_{\text{out}}} \sum_{i=1}^{N_{\text{out}}} \left( \widehat{h}_{\mathrm{express}}(\bm{\theta}_{1,i}, \bm{\theta}_{2,i},\sigma_i) \widehat{h}_{\mathrm{express}}(\bm{\theta}_{1,i}, \bm{\theta}_{2,i},\sigma'_i)  \right),
\end{equation}
where $\bm{\theta}_{1,i}$, $\bm{\theta}_{2,i}$, $\sigma_i$ and $\sigma'_i$ are the $i$-th sampled configurations, and $\widehat{h}_{\mathrm{express}}(\bm{\theta}_{1,i}, \bm{\theta}_{2,i},\sigma_i)$ and $\widehat{h}_{\mathrm{express}}(\bm{\theta}_{1,i}, \bm{\theta}_{2,i},\sigma'_i)$ denotes the corresponding inner-layer estimates.
Other terms can be estimated analogously using the same cross-product construction.

In terms of computational cost, since $N_g = \order{L}$, the outer-layer cost is $\order{N_{\text{out}} (n+L)}$.

\textbf{Inner-layer:}
In the inner-layer, we estimate the bounded function $h_{*}(\cdot)$ for each sampled configuration from the outer-layer.
The definition of bounded function $h_{*}(\cdot)$ depends on the diagnostic quantity of interest, and it can be expressed as a statistic of Pauli paths.
We summarize the definitions of $h_{*}(\cdot)$ for all diagnostic quantities in Table~\ref{tab:table_diagnostic_h}:
\begin{table*}[h!]
\centering
\begin{small}
\begin{tabular}{|p{0.11\textwidth}||p{0.85\textwidth}|}
\hline
Diagnostic Quantities & Definition of $h_{*}(\cdot)$ \\ \hline
Noise robustness & $h_{\mathrm{robustness}}(\bm{\theta}) = \left( \sum_{\vec{s}} f(\vec{s},\bm{\theta},O,\rho) - \widetilde{f}(\vec{s},\bm{\theta},O,\rho) \right)^2$ \eqref{ap:eq:mse_noise_robustness} \\ \hline
Noise sensitivity & $h_{\mathrm{sensitivity}}(\bm{\theta}) = -2 \left(\sum_{\vec{s}} (1-g_{\mathcal{N}}(\vec{s}))f(\vec{s},\bm{\theta},O,\rho) \right) \left( \sum_{\vec{s'}}  \left( \mathbf{1}_{(s_l|_{\mathcal{N}_{l}}=\mathbb{I}/\sqrt{2})}(\vec{s'})-1 \right)g_{\mathcal{N}\setminus \mathcal{N}_{l}}(\vec{s'})  f(\vec{s'},\bm{\theta},O,\rho)  \right)$ \eqref{ap:eq:noise_sensitivity} and 
$h_{\mathrm{sensitivity}}^{(amp)}(\bm{\theta}) = -2 \left(  \sum_{\vec{s}} f(\vec{s},\bm{\theta},O,\rho) - \widetilde{f}(\vec{s},\bm{\theta},O,\rho) \right) \left\{\sum_{\vec{s},\vec{\tau}} 2 \widetilde{f}_{*}^{(\vec{\tau})}(\vec{s},\bm{\theta},O,\rho) \right\}$ \eqref{ap:eq:amplitude_damping_noise_sensitivity} \\ \hline
Trainability & $h_{\mathrm{trainability}}(\bm{\theta}) = \left(\sum_{\vec{s}} \mathbf{1}_{ac}(s_{i-1},P_i) \widetilde{f}(\vec{s},\bm{\theta}+\frac{\pi}{2}e_i,O,\rho)\right)^2$ \eqref{ap:eq:variance_gradient_noisy} \\ \hline
Expressivity & $h_{\mathrm{express}}(\bm{\theta_1}, \bm{\theta_2}, \sigma) =  \tr{\widetilde{\mathcal{C}}^\dagger (\bm{\theta_2}) \left(\widetilde{\mathcal{C}}(\bm{\theta_1})(\ketbra{0^{n}}) \right) \sigma}$
and
$h'_{\mathrm{express}}(\bm{\theta}, \sigma') = \tr{\widetilde{\mathcal{C}}(\bm{\theta})(\ketbra{0^{n}}) \sigma'}^2$ \eqref{ap:eq:noisy_expressibility} \\ \hline
Expressivity Lower bound & $h_{\mathrm{expressibility}}^{\mathrm{lower}}(\bm{\theta},\sigma') = \tr{\widetilde{\mathcal{C}}(\bm{\theta})(\ketbra{0^{n}}) \sigma' }^2$ \eqref{ap:eq:noisy_expressibility_lower_bound} \\ \hline
\end{tabular}
\end{small}
\caption{The definition of the bounded function $h_{*}(\cdot)$ for each diagnostic quantity.}
\label{tab:table_diagnostic_h}
\end{table*}

From Table~\ref{tab:table_diagnostic_h}, we observe that, for each diagnostic quantity, the bounded function $h_{*}(\cdot)$ can be written as a statistic of Pauli paths. 
For expressibility and its lower bound, $h_{*}(\cdot)$ is not given directly as such a statistic; however, it can be estimated from Pauli measurement expectation values, which themselves can be expressed as Pauli path statistics using Eq.~\eqref{eq:pauli_path_integral_noisy}.
As a result, estimating $h_{*}(\cdot)$ reduces to evaluating the following quantities, when $\bm{\theta} \in \{0,\frac{\pi}{2},\pi,\frac{3\pi}{2}\}^{N_g}$:
\begin{equation}\label{ap:eq:general_pauli_path_statistic}
  \sum_{\vec{s}} g_*(\vec{s})f(\vec{s},\bm{\theta},O,\rho), \quad \sum_{\vec{s}} g_*(\vec{s})\widetilde{f}(\vec{s},\bm{\theta},O,\rho), \quad \sum_{\vec{s},\vec{\tau}} g_*(\vec{s},\vec{\tau}) \widetilde{f}^{(\vec{\tau})}(\vec{s},\bm{\theta},O,\rho),
\end{equation}
where $g_*(\cdot)$ is a known function depending on the specific diagnostic quantity. The $f(\vec{s},\bm{\theta},O,\rho)$ is defined in Eq.~\eqref{eq:contribution_function_noiseless}, $\widetilde{f}(\vec{s},\bm{\theta},O,\rho)$ is defined in Eq.~\eqref{eq:contribution_function_noise}, and $\widetilde{f}^{(\vec{\tau})}(\vec{s},\bm{\theta},O,\rho)$ is defined in Eq.~\eqref{eq:contribution_function_noise_tau}. 
For $\sum_{\vec{s}} g_*(\vec{s})f(\vec{s},\bm{\theta},O,\rho)$, it can be viewed as a special case of the second quantity when the noise channels are identity channels.
By the fact that $\widetilde{f}(\vec{s},\bm{\theta},O,\rho) = \sum_{\vec{\tau}} \widetilde{f}^{(\vec{\tau})}(\vec{s},\bm{\theta},O,\rho)$, the second quantity in Eq.~\eqref{ap:eq:general_pauli_path_statistic} can be obtained by the third quantity. 
Thus, we mainly focus on designing an efficient algorithm to estimate the third quantity in Eq.~\eqref{ap:eq:general_pauli_path_statistic}.

Furthermore, when $\theta_i \in \{0,\frac{\pi}{2},\pi,\frac{3\pi}{2}\}$, the unitary $U_i(\theta_i)$ is a Clifford gate, and its conjugation map $U_i(\theta_i) (\cdot) U_i^\dagger(\theta_i)$ is therefore a PCS1 channel.
Without loss of generality, the entire parameterized quantum circuit $\widetilde{\mathcal{C}}(\bm{\theta})$ can be written as a composition of PCS1 channels:
\begin{equation}\label{ap:eq:PCS1_circuit}
  \widetilde{\mathcal{C}}(\bm{\theta}) = \mathcal{C}_{L} \circ \mathcal{C}_{L-1} \circ \cdots \circ \mathcal{C}_1,
\end{equation}
where $\left\{\mathcal{C}_i \right\}_{i=1}^L$ includes both the parameterized gates and the noise channels and $L$ is the depth of the quantum circuit.
Using the explicit form of $\sum_{\vec{s},\vec{\tau}} g_*(\vec{s},\vec{\tau}) \widetilde{f}^{(\vec{\tau})}(\vec{s},\bm{\theta},O,\rho) = \sum_{\vec{s},\vec{\tau}} g_*(\vec{s},\vec{\tau}) \tr{O s_{N_g}} \tr{s_0 \rho} \prod_{i=1}^{{N_g}} \tr{\tau_i U_i(\theta_i) s_{i-1} U_i^\dagger(\theta_i)}\tr{s_i \mathcal{N}_i(\tau_i)}$, and the decomposition in Eq.~\eqref{ap:eq:PCS1_circuit}, we can replace the unitary conjugation maps $U_i(\theta_i) (\cdot) U_i^\dagger(\theta_i)$ and the noise channels $\mathcal{N}_i$ with PCS1 channel $\mathcal{C}_i$.
This yields:
\begin{equation}\label{ap:eq:general_pauli_path_statistic_PCS1}
  \sum_{\vec{s},\vec{\tau}} g_*(\vec{s},\vec{\tau}) \widetilde{f}^{(\vec{\tau})}(\vec{s},\bm{\theta},O,\rho) = \sum_{\vec{s}} g_*(\vec{s}) \tr{O s_{L}} \tr{s_0 \rho} \prod_{i=1}^{{L}} \tr{s_i \mathcal{C}_i (s_{i-1})},
\end{equation}
where we replace $(\vec{s},\vec{\tau})$ by a single, unified Pauli path, while retaining the notation $\vec{s}$; explicitly, $\vec{s}=(s_0,\tau_1,s_1,\tau_2,s_2,\cdots,\tau_{N_g},s_{N_g})$.

Here, we design a algorithm to estimate the statistic in Eq.~\eqref{ap:eq:general_pauli_path_statistic_PCS1} of a parameterized quantum circuit $\widetilde{\mathcal{C}}(\bm{\theta})$, where the rotation angles $\{\theta_i\}$ taking the value in $\{0,\frac{\pi}{2},\pi,\frac{3\pi}{2}\}$.
To ensure the input of the quantum circuits is classical efficient, we assume that the initial state $\rho$ is sparse, i.e $\rho=\sum_{a,b} \rho_{a,b}\ketbra{a}{b}$, where $\ket{a}$ and $\ket{b}$ are computational basis states and there are $N_\rho=\order{\mathrm{Poly}(n)}$ of non-zero $\rho_{a,b}$.
Furthermore, we assume the Pauli $1$-norm of observable $O$ is polynomially bounded.
Concretely, we write $O=\sum_{h=1}^{N_O} c_h P_h$, where $P_h$ are Pauli operators, $c_h$ are the corresponding coefficient and $N_O$ is number of nonzero coefficients.
We call $O$ has polynomially bounded Pauli $1$-norm if $\norm{O}_{\mathrm{Pauli},1}\coloneq \sum_{h=1}^{N_O} \abs{c_h} =\order{\mathrm{Poly}(n)}$.

To illustrate the procedure, we first rewrite Eq.~\eqref{ap:eq:general_pauli_path_statistic_PCS1} as:
\begin{small}
\begin{equation}
  \begin{aligned}
    &\sum_{\vec{s}} g_*(\vec{s}) \tr{O s_{L}} \tr{s_0 \rho} \prod_{i=1}^{{L}} \tr{s_i \mathcal{C}_i (s_{i-1})} \\
    =& \sum_{h=1}^{N_O} c_h  \left[ \sum_{\vec{s}} g_*(\vec{s}) \tr{P_h s_{L}} \tr{s_0 \rho} \prod_{i=1}^{{L}} \tr{s_i \mathcal{C}_i (s_{i-1})}\right] \\
    =& \sum_{h=1}^{N_O} c_h  \left[ \sum_{\vec{s}} g_*(\vec{s}) \tr{P_h s_{L}} \tr{s_0 \rho} \prod_{i=1}^{{L}} \frac{\abs{\tr{s_i \mathcal{C}_i (s_{i-1})}}}{\sum_\tau \abs{\tr{s_i \mathcal{C}_i (\tau)}}} \left(\sum_\tau \abs{\tr{s_i \mathcal{C}_i (\tau)}} \right) \operatorname{sign}\left(\tr{s_i \mathcal{C}_i (s_{i-1})}\right) \right] \\
    =& \sum_{h=1}^{N_O} c_h  \left[ \sum_{\vec{s}\mid s_L=\frac{P_h}{\sqrt{2^n}}} g_*(\vec{s}) \sqrt{2^n} \tr{s_0 \rho} \prod_{i=1}^{{L}} \frac{\abs{\tr{s_i \mathcal{C}_i (s_{i-1})}}}{\sum_\tau \abs{\tr{s_i \mathcal{C}_i (\tau)}}} \left(\sum_\tau \abs{\tr{s_i \mathcal{C}_i (\tau)}} \right) \operatorname{sign}\left(\tr{s_i \mathcal{C}_i (s_{i-1})}\right) \right] \\
    =& \sum_{h=1}^{N_O} \frac{\abs{c_h}}{\sum_{k=1}^{N_O} \abs{c_k}} \left(\sum_{k=1}^{N_O} \abs{c_k}\right) \operatorname{sign}\left(c_h\right)  \left[ \sum_{\vec{s}\mid s_L=\frac{P_h}{\sqrt{2^n}}} g_*(\vec{s}) \sqrt{2^n} \tr{s_0 \rho} \prod_{i=1}^{{L}} \frac{\abs{\tr{s_i \mathcal{C}_i (s_{i-1})}}}{\sum_\tau \abs{\tr{s_i \mathcal{C}_i (\tau)}}} \left(\sum_\tau \abs{\tr{s_i \mathcal{C}_i (\tau)}} \right) \operatorname{sign}\left(\tr{s_i \mathcal{C}_i (s_{i-1})}\right) \right] \\
  \end{aligned}
\end{equation}  
\end{small}

where $\sum_{\vec{s}\mid s_L=\frac{P_h}{\sqrt{2^n}}}$ denotes the summation over all Pauli paths $\vec{s}$ with the last element $s_L$ fixed to $\frac{P_h}{\sqrt{2^n}}$.

Defining $p\left(s_L=\frac{P_h}{\sqrt{2^n}}\right) = \frac{\abs{c_h}}{\sum_{k=1}^{N_O} \abs{c_k}}$ as the probability of choosing the last element $s_L=\frac{P_h}{\sqrt{2^n}}$ and $p\left( s_{i-1}\mid s_i \right) = \frac{\abs{\tr{s_i \mathcal{C}_i (s_{i-1})}}}{\sum_\tau \abs{\tr{s_i \mathcal{C}_i (\tau)}}}$ as the conditional probability of $s_{i-1}$ given $s_i$, we can express the entire summation as an expectation value over the distribution of Pauli paths:
\begin{small}
\begin{equation}\label{eq:pauli_path_expectation}
  \begin{aligned}
    &\sum_{\vec{s}} g_*(\vec{s}) \tr{O s_{L}} \tr{s_0 \rho} \prod_{i=1}^{{L}} \tr{s_i \mathcal{C}_i (s_{i-1})} \\
    =& \sum_{h=1}^{N_O} \frac{\abs{c_h}}{\sum_{k=1}^{N_O} \abs{c_k}} \left(\sum_{k=1}^{N_O} \abs{c_k}\right) \operatorname{sign}\left(c_h\right)  \left[ \sum_{\vec{s}\mid s_L=\frac{P_h}{\sqrt{2^n}}} g_*(\vec{s}) \sqrt{2^n} \tr{s_0 \rho} \prod_{i=1}^{{L}} \frac{\abs{\tr{s_i \mathcal{C}_i (s_{i-1})}}}{\sum_\tau \abs{\tr{s_i \mathcal{C}_i (\tau)}}} \left(\sum_\tau \abs{\tr{s_i \mathcal{C}_i (\tau)}} \right) \operatorname{sign}\left(\tr{s_i \mathcal{C}_i (s_{i-1})}\right) \right] \\
    =& \sum_{s_L} p\left(s_L\right) \left(\sum_{k=1}^{N_O} \abs{c_k}\right) \operatorname{sign}\left(\text{coff}(s_L)\right) \left[ \sum_{\vec{s}\mid s_L=\frac{P_h}{\sqrt{2^n}}} g_*(\vec{s}) \sqrt{2^n} \tr{s_0 \rho} \prod_{i=1}^{{L}} p\left( s_{i-1}\mid s_i \right) \left(\sum_\tau \abs{\tr{s_i \mathcal{C}_i (\tau)}} \right) \operatorname{sign}\left(\tr{s_i \mathcal{C}_i (s_{i-1})}\right) \right] \\
    =& \sum_{\vec{s}} \left[ p\left(s_L\right) \left(\sum_{k=1}^{N_O} \abs{c_k}\right) \operatorname{sign}\left(\text{coff}(s_L)\right) g_*(\vec{s}) \sqrt{2^n} \tr{s_0 \rho} \prod_{i=1}^{{L}} p\left( s_{i-1}\mid s_i \right) \left(\sum_\tau \abs{\tr{s_i \mathcal{C}_i (\tau)}} \right) \operatorname{sign}\left(\tr{s_i \mathcal{C}_i (s_{i-1})}\right) \right] \\
    =& \sum_{\vec{s}} \left[ p\left(\vec{s}\right) g_*(\vec{s}) \sqrt{2^n} \tr{s_0 \rho} \left(\sum_{k=1}^{N_O} \abs{c_k}\right) \operatorname{sign}\left(\text{coff}(s_L)\right) \prod_{i=1}^{{L}} \left(\sum_\tau \abs{\tr{s_i \mathcal{C}_i (\tau)}} \right) \operatorname{sign}\left(\tr{s_i \mathcal{C}_i (s_{i-1})}\right) \right] \\
    =& \mathbb{E}_{\vec{s}} \left[ g_*(\vec{s}) \sqrt{2^n} \tr{s_0 \rho} \left(\sum_{k=1}^{N_O} \abs{c_k}\right) \operatorname{sign}\left(\text{coff}(s_L)\right) \prod_{i=1}^{{L}} \left(\sum_\tau \abs{\tr{s_i \mathcal{C}_i (\tau)}} \right) \operatorname{sign}\left(\tr{s_i \mathcal{C}_i (s_{i-1})}\right) \right]
  \end{aligned}
\end{equation}  
\end{small}

where $\text{coff}(s_L)$ denotes the coefficient $c_h$ corresponding to the Pauli operator $P_h$ such that $s_L=\frac{P_h}{\sqrt{2^n}}$, and $p\left(\vec{s}\right)= p\left(s_L\right) \prod_{i=1}^L p\left( s_{i-1}\mid s_i \right)$ is the joint probability of the Pauli path $\vec{s}$.
Here, $\mathbb{E}_{\vec{s}}$ is the expectation value over the distribution of Pauli paths, i.e., $\mathbb{E}_{\vec{s}}[X]=\sum_{\vec{s}} p(\vec{s}) X(\vec{s})$, where $p(\vec{s})=p\left(s_L\right) \prod_{i=1}^L p\left( s_{i-1}\mid s_i \right)$.

Depending on Eq.~\eqref{eq:pauli_path_expectation}, we estimate the statistic $\sum_{\vec{s},\vec{\tau}} g_*(\vec{s},\vec{\tau}) \widetilde{f}^{(\vec{\tau})}(\vec{s},\bm{\theta},O,\rho)$ by sampling Pauli paths $\vec{s}$ from the distribution $p(\vec{s})$ and evaluating the corresponding contribution for each sampled path.
The algorithm starts from the final observable and back-propagates in the Hesienberg picture, the procedure is summarized as follows:
\begin{enumerate}
  \item Initialize by sampling the last element $s_{N_g}$ of the Pauli path $\vec{s}$ from the distribution $p\left(s_{N_g}\right) = \frac{\abs{c_h}}{\sum_{k=1}^{N_O} \abs{c_k}}$, where $s_{N_g}=\frac{P_h}{\sqrt{2^n}}$.
  Compute the factor $\left(\sum_{k=1}^{N_O} \abs{c_k}\right) \operatorname{sign}\left(c_h\right)$.
  The cost of sampling $s_{N_g}$ is $\order{n}$ by leveraging the tree data structure~\cite{shao2024simulating}.
  \item For each layer $i$ from $N_g$ down to $1$, given the Pauli operator $s_i$ at the $i$-th layer, we sample $s_{i-1}$ from the conditional distribution $p\left( s_{i-1}\mid s_i \right) = \frac{\abs{\tr{s_i \mathcal{C}_i (s_{i-1})}}}{\sum_\tau \abs{\tr{s_i \mathcal{C}_i (\tau)}}}$ and compute the factor $\left(\sum_\tau \abs{\tr{s_i \mathcal{C}_i (\tau)}} \right) \operatorname{sign}\left(\tr{s_i \mathcal{C}_i (s_{i-1})}\right)$.
  The computational cost of this step is:
  \begin{itemize}
    \item If $\mathcal{C}_i$ is a $\order{1}$-local PCS1 channel, the number of non-zero terms in the set $\{\tr{s_i \mathcal{C}_i (\tau)}\}_\tau$ is $\order{1}$. Hence one can enumerate all non-zero terms and sample $s_{i-1}$ with cost $\order{1}$.
    \item If $\mathcal{C}_i$ is a unitary conjugation map of a Clifford gate, we can use the stabilizer formalism to calculate $\mathcal{C}_i^\dagger (s_i)$ efficiently, and set $s_{i-1}=\mathcal{C}_i^\dagger (s_i)$. 
    This costs $\order{n}$, since $\mathcal{C}_i$ may act on all $n$ qubits and the Pauli label must be updated on each qubit.
  \end{itemize}
  \item[\vdots] 
  \item[] Continue the layer-wise back-propagation until reaching the initial state $\rho$ and yield the final Pauli operator $s_0$.
  \item Compute the factor $\tr{s_0 \rho}$ with $\tr{s_0 \rho}=\sum_{a,b}\tr{s_0  (\rho_{a,b}\ketbra{a}{b})}=\sum_{a,b}\rho_{a,b} \bra{b}s_0 \ket{a}=\sum_{a,b}\rho_{a,b} \prod_{j=1}^n \bra{b}_j (s_0)|_j \ket{a}_j$, where $|_j$ is the restriction to the $j$-th qubit and $\ket{\cdot}_j$ denotes $j$-th component of $\ket{\cdot}$.
  The cost of this step is $\order{N_\rho n}$.
  \item Finally, combine all the calculated factors to obtain the contribution of the sampled Pauli path $\vec{s}$, denoted by $T(\vec{s})$:
  \begin{equation}
      T(\vec{s}) := g_*(\vec{s}) \sqrt{2^n} \tr{s_0 \rho} \left(\sum_{k=1}^{N_O} \abs{c_k}\right) \operatorname{sign}\left(\text{coff}(s_{L})\right) \prod_{i=1}^{{L}} \left(\sum_\tau \abs{\tr{s_i \mathcal{C}_i (\tau)}} \right) \operatorname{sign}\left(\tr{s_i \mathcal{C}_i (s_{i-1})}\right).
  \end{equation}
\end{enumerate}
To estimate the statistic $\mathbb{E}_{\vec{s}} T(\vec{s})$, we repeat the above procedure $N_{\text{in}}$ times to obtain independent samples of Pauli paths $\{\vec{s}_i\}_{i=1}^{N_{\text{in}}}$, and approximate the expectation by the empirical mean:
\begin{equation}\label{ap:eq:inner_layer_empirical_mean}
  \mathbb{E}_{\vec{s}} T(\vec{s}) \approx \frac{1}{N_{\text{in}}} \sum_{i=1}^{N_{\text{in}}} T(\vec{s}_i).
\end{equation}
The empirical mean estimate the statistic $\sum_{\vec{s},\vec{\tau}} g_*(\vec{s},\vec{\tau}) \widetilde{f}^{(\vec{\tau})}(\vec{s},\bm{\theta},O,\rho)$ and therefore the bounded function $h_{*}(\cdot)$ for the sampled configuration from the outer-layer.
Assuming that there are $N_{\text{in}}$ samples in the inner-layer, the computational cost of the inner-layer is $\order{N_{\text{in}} (nL + N_\rho n)}$.

On the other hand, each outer-layer sample requires calling the inner-layer to estimate the bounded function $h_{*}(\cdot)$.
Although there are several Pauli path statistics in a single $h_{*}(\cdot)$ (e.g., noise robustness has two Pauli path statistics in Eq.~\eqref{ap:eq:mse_noise_robustness}), they can share the same sampled Pauli paths $\{\vec{s}_i\}_{i=1}^{N_{\text{in}}}$ from the inner-layer.
Thus, the inner-layer only needs to be called once for each outer-layer sample.
Assuming there are $N_{\text{out}}$ samples in the outer-layer, the total computational cost of the two-layer sampling algorithm is:
\begin{equation}\label{ap:eq:total_sample_complexity}
  \order{N_{\text{out}} N_{\text{in}} (nL + N_\rho n)} + \order{N_{\text{out}} (n+L)} = \order{N_{\text{out}} N_{\text{in}} (nL + N_\rho n)}.
\end{equation}

\subsection{Error analysis}

We analyze the estimation error of the two-layer sampling algorithm.
Note that in the definitions of ``Expressivity'' and ``Expressivity Lower bound'', the terms $- \frac{2}{ \left(2^n+1\right)} \mathbb{E}^*_{\bm{\theta},\sigma'} h'_{\mathrm{express}}(\bm{\theta},\sigma')$ and $-\frac{2}{2^n+1}\mathbb{E}^*_{\bm{\theta}} h_{\mathrm{express}}^{\mathrm{lower}}(\bm{\theta},\sigma')$ in Eqs.~\eqref{ap:eq:noisy_expressibility} and~\eqref{ap:eq:noisy_expressibility_lower_bound}, respectively, contribute at most $\frac{2}{2^n+1}$ error, since $\abs{h'_{\mathrm{express}}(\bm{\theta},\sigma')}\leq 1$ and $\abs{h_{\mathrm{express}}^{\mathrm{lower}}(\bm{\theta},\sigma')}\leq 1$.
We therefore omit these terms in the subsequent error analysis.

\textbf{Inner layer:}
We first analyze the inner-layer estimation error.
The bounded function $h_{*}(\cdot)$ is computed from the Pauli path statistics in Eq.~\eqref{ap:eq:general_pauli_path_statistic}, we estimate using the empirical mean in Eq.~\eqref{ap:eq:inner_layer_empirical_mean} based on $N_{\text{in}}$ independent samples.
Hoeffding's inequality then yields a bound on the estimation error. Specifically, the magnitude of a single sample $T(\vec{s})$ satisfies:
\begin{equation}\label{ap:eq:inner_layer_Ts_bound}
  \begin{aligned}
    \abs{T(\vec{s})} &\leq \abs{g_*(\vec{s})} \sqrt{2^n} \abs{\tr{s_0 \rho}} \left(\sum_{k=1}^{N_O} \abs{c_k}\right) \prod_{i=1}^{{L}} \left(\sum_\tau \abs{\tr{s_i \mathcal{C}_i (\tau)}} \right) \\
    &\leq \abs{g_*(\vec{s})} \sqrt{2^n} \norm{\rho}_1 \norm{s_0}_\infty \left(\sum_{k=1}^{N_O} \abs{c_k}\right)  \\
    & = \abs{g_*(\vec{s})} \left(\sum_{k=1}^{N_O} \abs{c_k}\right), \\
    & \leq \max_{\vec{s}}\abs{g_*(\vec{s})} \norm{O}_{\mathrm{Pauli},1} \\
    & \leq 2 \norm{O}_{\mathrm{Pauli},1},
  \end{aligned}
\end{equation}
where we use the fact that $\left(\sum_\tau \abs{\tr{s_i \mathcal{C}_i (\tau)}} \right)\leq 1$ for PCS1 channels, $\abs{\tr{s_0 \rho}} \leq \norm{\rho}_1 \norm{s_0}_\infty = \frac{1}{\sqrt{2^n}}$, and $\max_{\vec{s}}\abs{g_*(\vec{s})} \leq 2$ from Table~\ref{tab:table_diagnostic_h}.
Consequently, Hoeffding's inequality implies:
\begin{equation}\label{ap:eq:inner_layer_hoeffding}
  \mathbb{P} \left(  
    \abs{\mathbb{E}_{\vec{s}} T(\vec{s})
   - \frac{1}{N_{\text{in}}} \sum_{i=1}^{N_{\text{in}}} T(\vec{s}_i)} \geq \varepsilon_{\text{in}} \right) \leq 2 \exp{ -\frac{2 N_{\text{in}} \varepsilon_{\text{in}}^2}{\left(2 \norm{O}_{\mathrm{Pauli},1}\right)^2}}.
\end{equation}
For any $\varepsilon_{\text{in}}>0$ and $\delta_{\text{in}}\in(0,1)$, it suffices to choose:
\begin{equation}\label{ap:eq:inner_layer_sample_complexity}
  N_{\text{in}} \geq \frac{2 \norm{O}_{\mathrm{Pauli},1}^2}{\varepsilon_{\text{in}}^2} \ln{\frac{2}{\delta_{\text{in}}}},
\end{equation}
in which case the inner layer estimates the statistic in Eq.~\eqref{ap:eq:general_pauli_path_statistic} with an error at most $\varepsilon_{\text{in}}$ with probability at least $1-\delta_{\text{in}}$.

Finally, the bounded function $h_{*}(\cdot)$ associated with each diagnostic quantity can be estimated with an error that depends on $\varepsilon_{\text{in}}$.
Assume that the inner-layer estimation error is at most $\varepsilon_{\text{in}}$ with probability at least $1-\delta_{\text{in}}$ for each term in Eq.~\eqref{ap:eq:general_pauli_path_statistic}.
Using the union bound, we can derive the estimation error for each bounded function $h_{*}(\cdot)$.
We summarize the results in Table~\ref{tab:table_diagnostic_h_error}.

\begin{remark}
  For the expressibility and its lower bound, the $\norm{O}_{\mathrm{Pauli},1}$ in Eq.~\eqref{ap:eq:inner_layer_sample_complexity} is $1$ due to the choice of observable $\sigma\in \{ \mathbb{I}, X, Y, Z\}^{\otimes n}$.
\end{remark}

\begin{table*}[h!]
\centering
\begin{small}
\begin{tabular}{|p{0.2\textwidth}||p{0.8\textwidth}|}
\hline
$h_{*}(\cdot)$ & Estimation error \\ \hline
$h_{\mathrm{robustness}}(\bm{\theta})$~\eqref{ap:eq:mse_noise_robustness} & $\mathbb{P} \left(\abs{h_{\mathrm{robustness}}(\bm{\theta}) - \widehat{h}_{\mathrm{robustness}}(\bm{\theta})} \leq 8\norm{O}_\infty\varepsilon_{\text{in}} + 4\varepsilon_{\text{in}}^2 \right) \geq 1 - 2\delta_{\text{in}} $  \\ \hline
$h_{\mathrm{sensitivity}}(\bm{\theta})$ \eqref{ap:eq:noise_sensitivity} & $\mathbb{P} \left(\abs{h_{\mathrm{sensitivity}}(\bm{\theta}) - \widehat{h}_{\mathrm{sensitivity}}(\bm{\theta})} \leq 4\norm{O}_\infty\varepsilon_{\text{in}} + 2\varepsilon_{\text{in}}^2 \right) \geq 1 - 2\delta_{\text{in}} $ \\ \hline
$h_{\mathrm{sensitivity}}^{(amp)}(\bm{\theta})$ \eqref{ap:eq:amplitude_damping_noise_sensitivity}& $\mathbb{P} \left(\abs{h_{\mathrm{sensitivity}}^{(amp)}(\bm{\theta}) - \widehat{h}_{\mathrm{sensitivity}}^{(amp)}(\bm{\theta})} \leq 8\norm{O}_\infty\varepsilon_{\text{in}} + 4\varepsilon_{\text{in}}^2 \right) \geq 1 - 2\delta_{\text{in}} $\\ \hline
$h_{\mathrm{trainability}}(\bm{\theta})$ \eqref{ap:eq:variance_gradient_noisy}& $\mathbb{P} \left(\abs{h_{\mathrm{trainability}}(\bm{\theta}) - \widehat{h}_{\mathrm{trainability}}(\bm{\theta})} \leq 2\norm{O}_\infty\varepsilon_{\text{in}} + \varepsilon_{\text{in}}^2 \right) \geq 1 - 2\delta_{\text{in}} $  \\ \hline
$h_{\mathrm{express}}(\bm{\theta_1}, \bm{\theta_2}, \sigma)$ \eqref{ap:eq:noisy_expressibility} & $\mathbb{P} \left(\abs{h_{\mathrm{express}}(\bm{\theta_1}, \bm{\theta_2}, \sigma) - \widehat{h}_{\mathrm{express}}(\bm{\theta_1}, \bm{\theta_2}, \sigma)} \leq \varepsilon_{\text{in}} \right) \geq 1 -\delta_{\text{in}} $ and $\norm{O}_{\mathrm{Pauli},1}=1$\\ \hline
$h_{\mathrm{express}}^{\mathrm{lower}}(\bm{\theta},\sigma')$ \eqref{ap:eq:noisy_expressibility_lower_bound} & $\mathbb{P} \left(\abs{h_{\mathrm{express}}^{\mathrm{lower}}(\bm{\theta},\sigma') - \widehat{h}_{\mathrm{express}}^{\mathrm{lower}}(\bm{\theta},\sigma')} \leq 2\varepsilon_{\text{in}} + \varepsilon_{\text{in}}^2 \right) \geq 1 -\delta_{\text{in}} $ and $\norm{O}_{\mathrm{Pauli},1}=1$\\ \hline
\end{tabular}
\end{small}
\caption{Error bounds for estimating the bounded functions $h_{*}(\cdot)$ (and thus each diagnostic quantity) induced by the inner-layer sampling error $\varepsilon_{\mathrm{in}}$ with failure probability at most $\delta_{\mathrm{in}}$.}
\label{tab:table_diagnostic_h_error}
\end{table*}

Let $\varepsilon_h$ and $\delta_h$ denote the desired estimation error and failure probability for the bounded function $h_{*}(\cdot)$, respectively.
We consider the worst-case bound among the first four diagnostic quantities ($h_{\mathrm{robustness}}, h_{\mathrm{sensitivity}}, h_{\mathrm{sensitivity}}^{(amp)}, h_{\mathrm{trainability}}$) in Table~\ref{tab:table_diagnostic_h_error}, i.e $8\norm{O}_\infty\varepsilon_{\text{in}} + 4\varepsilon_{\text{in}}^2 \leq \varepsilon_h$ and $2\delta_{\text{in}} \leq \delta_h$.
Accordingly, we choose
\begin{equation}
  \varepsilon_{\text{in}} = \frac{-8\norm{O}_\infty + \sqrt{64\norm{O}_\infty^2 + 16\varepsilon_h}}{8} \leq \frac{\varepsilon_h}{8\norm{O}_\infty}, \quad \delta_{\text{in}} = \frac{\delta_h}{2},
\end{equation}
which ensure that the estimation error of $h_{*}(\cdot)$ is at most $\varepsilon_h$ with probability at least $1-\delta_h$.
Substituting these choices into Eq.~\eqref{ap:eq:inner_layer_sample_complexity}, yields the required number of inner-layer samples:
\begin{equation}\label{ap:eq:inner_layer_sample_complexity_final}
  N_{\text{in}} \geq \frac{256 \norm{O}_{\mathrm{Pauli},1}^2 \norm{O}_\infty^2}{\varepsilon_h^2} \ln{\frac{4}{\delta_h}},
\end{equation}
which guarantees the desired accuracy and confidence level for estimating $h_{\mathrm{robustness}}, h_{\mathrm{sensitivity}}, h_{\mathrm{sensitivity}}^{(amp)}, h_{\mathrm{trainability}}$.

For $h_{\mathrm{express}}(\bm{\theta_1}, \bm{\theta_2}, \sigma)$ and $h_{\mathrm{express}}^{\mathrm{lower}}(\bm{\theta},\sigma')$, considering the error $2\varepsilon_{\text{in}} + \varepsilon_{\text{in}}^2 \leq \varepsilon_h$ and $\delta_{\text{in}} \leq \delta_h$, we choose
\begin{equation}
  \varepsilon_{\text{in}} = \sqrt{1+\varepsilon_h} -1 \leq \frac{\varepsilon_h}{2}, \quad \delta_{\text{in}} = \delta_h.
\end{equation}
Substituting these choices and $\norm{\sigma}_{\mathrm{Pauli},1}=\norm{\sigma}_\infty=1$ into Eq.~\eqref{ap:eq:inner_layer_sample_complexity}, yields the required number of inner-layer samples:
\begin{equation}
  N_{\text{in}} \geq \frac{8}{\varepsilon_h^2} \ln{\frac{2}{\delta_h}},
\end{equation}
which guarantees the desired accuracy and confidence level for estimating $h_{\mathrm{express}}(\bm{\theta_1}, \bm{\theta_2}, \sigma)$ and $h_{\mathrm{express}}^{\mathrm{lower}}(\bm{\theta},\sigma')$.

\textbf{Outer layer:}
Next, we analyze the outer-layer estimation error.
In the outer layer, the diagnostic quantities are estimated by taking the empirical mean of the bounded inner-layer estimator $\widehat{h}_{*}(\cdot)$ over $N_{\text{out}}$ independent samples.
Suppose that each evaluation of $h_{*}(\cdot)$ is approximated with error at most $\varepsilon_h$ with probability at least $1-\delta_h$, i.e, $\mathbb{P} \left(\abs{h_{*}(\cdot) - \widehat{h}_{*}(\cdot, r)} \leq \varepsilon_h \right) \geq 1-\delta_h$ over the randomness of $r$, where $r$ denotes the random variable in the inner-layer sampling.
Note that the estimator $\widehat{h}_{*}(\cdot,r)$ shares the same bound for $h_{*}(\cdot)$ as shown in Table~\ref{tab:table_diagnostic}, by using Eq.~\eqref{ap:eq:inner_layer_Ts_bound}.

For ``Noise robustness'', ``Noise sensitivity'', and ``Trainability'', the diagnostic quantities take the form of $\mathbb{E}^*_{\bm{\theta}} h_{*}(\bm{\theta})$.
We estimate this expectation using the empirical mean over $N_{\text{out}}$ independent samples:
\begin{equation}\label{ap:eq:outer_layer_empirical_mean}
  \mathbb{E}^*_{\bm{\theta}} h_{*}(\bm{\theta}) \approx \frac{1}{N_{\text{out}}} \sum_{i=1}^{N_{\text{out}}} \widehat{h}_{*}(\bm{\theta}_i,r_i).
\end{equation}

The resulting error admits the decomposition:
\begin{equation}
  \abs{\mathbb{E}^*_{\bm{\theta}} h_{*}(\bm{\theta}) - \frac{1}{N_{\text{out}}} \sum_{i=1}^{N_{\text{out}}} \widehat{h}_{*}(\bm{\theta}_i, r_i)} \leq \abs{\mathbb{E}^*_{\bm{\theta}} h_{*}(\bm{\theta}) - \mathbb{E}^*_{\bm{\theta}}\mathbb{E}_{r} \widehat{h}_{*}(\bm{\theta},r)} + \abs{\mathbb{E}^*_{\bm{\theta}} \mathbb{E}_{r} \widehat{h}_{*}(\bm{\theta},r) - \frac{1}{N_{\text{out}}} \sum_{i=1}^{N_{\text{out}}} \widehat{h}_{*}(\bm{\theta}_i, r_i)}.
\end{equation}
For $\abs{\mathbb{E}^*_{\bm{\theta}} h_{*}(\bm{\theta}) - \mathbb{E}^*_{\bm{\theta}} \mathbb{E}_{r} \widehat{h}_{*}(\bm{\theta}, r)}$, we have:
\begin{equation}
  \abs{\mathbb{E}^*_{\bm{\theta}} h_{*}(\bm{\theta}) - \mathbb{E}^*_{\bm{\theta}} \mathbb{E}_{r} \widehat{h}_{*}(\bm{\theta},r)} \leq \mathbb{E}^*_{\bm{\theta}} \mathbb{E}_{r} \abs{h_{*}(\bm{\theta}) - \widehat{h}_{*}(\bm{\theta},r)} \leq (1-\delta_h) \varepsilon_h + \delta_h \cdot \max_{\bm{\theta}} (\abs{h_{*}(\bm{\theta})} + \abs{\widehat{h}_{*}(\bm{\theta},r)}) \leq \varepsilon_h + 16 \norm{O}_\infty^2 \delta_h,
\end{equation}
where we use the bounds of $h_{*}(\cdot)$ and $\widehat{h}_{*}(\cdot)$ from Table~\ref{tab:table_diagnostic}.

For the second term $\abs{\mathbb{E}^*_{\bm{\theta}} \mathbb{E}_{r} \widehat{h}_{*}(\bm{\theta},r) - \frac{1}{N_{\text{out}}} \sum_{i=1}^{N_{\text{out}}} \widehat{h}_{*}(\bm{\theta}_i,r_i)}$, Hoeffding's inequality yields:
\begin{equation}
  \mathbb{P} \left(
    \abs{\mathbb{E}^*_{\bm{\theta}} \mathbb{E}_{r} \widehat{h}_{*}(\bm{\theta},r) - \frac{1}{N_{\text{out}}} \sum_{i=1}^{N_{\text{out}}} \widehat{h}_{*}(\bm{\theta}_i,r_i)} \geq \varepsilon_{\text{out}} \right) \leq 2 \exp{ -\frac{2 N_{\text{out}} \varepsilon_{\text{out}}^2}{\left(16 \norm{O}_\infty^2\right)^2}} = 2 \exp{ -\frac{N_{\text{out}} \varepsilon_{\text{out}}^2}{128 \norm{O}_\infty^4}}
\end{equation}
Combining the above two parts, to ensure an overall error at most $\varepsilon$ with probability at least $1-\delta$, i.e.,
\begin{equation}
  \mathbb{P} \left(
    \abs{\mathbb{E}^*_{\bm{\theta}} h_{*}(\bm{\theta}) - \frac{1}{N_{\text{out}}} \sum_{i=1}^{N_{\text{out}}} \widehat{h}_{*}(\bm{\theta}_i,r_i)} \geq \varepsilon \right) \leq \delta,
\end{equation}
it suffices to enforce $\varepsilon_{\text{out}} + \varepsilon_h + 16 \norm{O}_\infty^2 \delta_h = \varepsilon$ and $2 \exp{ -\frac{N_{\text{out}} \varepsilon_{\text{out}}^2}{128 \norm{O}_\infty^4}} \leq \delta$.
One convenient choice is:
\begin{equation}\label{ap:eq:outer_layer_error_requirements}
  \varepsilon_{\text{out}} \leq \frac{\varepsilon}{3}, \quad \varepsilon_h \leq \frac{\varepsilon}{3}, \quad \delta_h \leq \frac{\varepsilon}{48 \norm{O}_\infty^2}.
\end{equation}
With this choice, the required number of outer-layer samples is:
\begin{equation}
  N_{\text{out}} \geq \frac{1152 \norm{O}_\infty^4}{\varepsilon^2} \ln{\frac{2}{\delta}}.
\end{equation}
Substituting Eq.~\eqref{ap:eq:outer_layer_error_requirements} into the inner-layer sample complexity~\eqref{ap:eq:inner_layer_sample_complexity_final}, we obtain:
\begin{equation}
  N_{\text{in}} \geq \frac{2304 \norm{O}_{\mathrm{Pauli},1}^2 \norm{O}_\infty^2}{\varepsilon^2} \ln{\frac{\norm{O}_\infty^2}{\varepsilon}}.
\end{equation}

Therefore, estimating ``Noise robustness'', ``Noise sensitivity'', and ``Trainability'' with error at most $\varepsilon$ and failure probability at most $\delta$, requires total cost:
\begin{equation}
  \order{N_{\text{out}} N_{\text{in}} (nL + N_\rho n)}=\order{\frac{\norm{O}_{\mathrm{Pauli},1}^2 \norm{O}_\infty^6 (nL + N_\rho n)}{\varepsilon^4} \ln{\frac{\norm{O}_\infty^2}{\varepsilon}}\ln{\frac{1}{\delta}} }
\end{equation}
which follows by substituting the abovee $N_{\text{out}}$ and $N_{\text{in}}$ into Eq.~\eqref{ap:eq:total_sample_complexity}.

For ``Expressivity'' and ``Expressivity Lower bound'', we take $\mathbb{E}^*_{\bm{\theta_1}, \bm{\theta_2}} \left[\mathbb{E}_{\sigma} h_{\mathrm{express}}(\bm{\theta_1}, \bm{\theta_2},\sigma)\right]^2$ as an representative example, the remaining terms can be analyzed analogously.
We estimate the expectation via the empirical mean over $N_{\text{out}}$ independent samples:
\begin{equation}
  \mathbb{E}^*_{\bm{\theta_1}, \bm{\theta_2}} \left[\mathbb{E}_{\sigma} h_{\mathrm{express}}(\bm{\theta_1}, \bm{\theta_2},\sigma)\right]^2 \approx \frac{1}{N_{\text{out}}} \sum_{i=1}^{N_{\text{out}}} \left( \widehat{h}_{\mathrm{express}}(\bm{\theta}_{1,i}, \bm{\theta}_{2,i},\sigma_i) \widehat{h}_{\mathrm{express}}(\bm{\theta}_{1,i}, \bm{\theta}_{2,i},\sigma'_i)  \right).
\end{equation}
Moreover, using the inner-layer guarantee $\mathbb{P} \left(\abs{\widehat{h}_{\mathrm{express}}-\widehat{h}_{\mathrm{express}}} > \varepsilon_h \right) \leq \delta_h$ and the bound $\abs{h_{\mathrm{express}}(\cdot)}, \abs{\widehat{h}_{\mathrm{express}}(\cdot)} \leq 1$, we have:
\begin{equation}
  \mathbb{P} \left(\abs{\widehat{h}_{\mathrm{express}}(\bm{\theta}_{1,i}, \bm{\theta}_{2,i},\sigma_i) \widehat{h}_{\mathrm{express}}(\bm{\theta}_{1,i}, \bm{\theta}_{2,i},\sigma'_i) - {h}_{\mathrm{express}}(\bm{\theta}_{1,i}, \bm{\theta}_{2,i},\sigma_i) {h}_{\mathrm{express}}(\bm{\theta}_{1,i}, \bm{\theta}_{2,i},\sigma'_i)} \geq \varepsilon_h^2 + 2\varepsilon_h \right) \leq 2\delta_h.
\end{equation}
Consequently, relative to the analyses for ``Noise robustness'', ``Noise sensitivity'', and ``Trainability'' the only change is that the effective inner-layer accuracy requirement becomes $\varepsilon_h \to \varepsilon_h^2 + 2\varepsilon_h$ and the failure probability scales as $\delta_h \to 2\delta_h$.

\subsection{Performance benchmark of estimation algorithm}
To benchmark the performance of the estimation algorithm, we consider the specific instance employed in the numerical simulations of Ref.~\cite{ragone2024lie}.
The circuit consists of $p$ blocks on $n$ qubits arranged linearly, see Fig.~\ref{ap:fig:line_ZXX_circuit}.
Each block contains a layer of single-qubit $R_Z$ rotations applied to all $n$ qubits and nearest-neighbor two-qubit $R_{XX}$ rotations, contributing $n-1$ two-qubit gates per block.
The observable is chosen as $O = X_q X_{q+1} + Z_q$, where $q = \lfloor n/2 \rfloor$ denotes the middle qubit, and the initial state is $\ket{0}^{\otimes n}$.

\begin{figure}[H]
  \centering
  \includegraphics[width=0.3\textwidth]{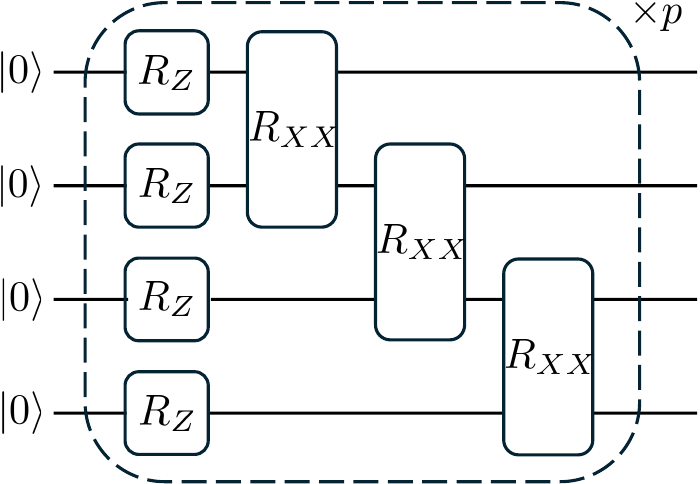}
  \caption{The circuit architecture used in the benchmark, composed of $p$ blocks, with one block shown in the dashed box.}
  \label{ap:fig:line_ZXX_circuit}
\end{figure}

Following the analytical framework of Refs.~\cite{ragone2024lie,fontana2024characterizing}, when the dynamical Lie algebra (DLA) is simple, the variance of the expectation value $\langle O \rangle$ over the circuit ensemble induced by DLA $\mathfrak{g}$ is
\begin{equation}
  \mathrm{Var}(\langle O \rangle) = \frac{\mathcal{P}_{\mathfrak{g}}(\ketbra{0}^{\otimes n})\mathcal{P}_{\mathfrak{g}}(O)}{\dim{\mathfrak{g}}},
\end{equation}
where $\dim{\mathfrak{g}}$ is the dimension of the DLA and $\mathcal{P}_{\mathfrak{g}}$ denotes the $\mathfrak{g}$-purity~\cite{somma2004nature,somma2005quantum}:
\begin{equation}
  \mathcal{P}_{\mathfrak{g}}(H) = \tr{H_{\mathfrak{g}}^2} = \sum_{i=1}^{\dim{\mathfrak{g}}} \tr{B_i^\dagger H}^2,
\end{equation}
with $H_{\mathfrak{g}} := \Pi_{\mathfrak{g}}(H)$ the Hilbert-Schmidt projection of $H$ onto the complexified algebra $\mathfrak{g}_{\mathbb{C}} := \mathrm{span}_{\mathbb{C}}\,\mathfrak{g}$, and $\{B_i\}_{i=1}^{\dim \mathfrak{g}}$ an orthonormal basis of $\mathfrak{g}_{\mathbb{C}}$ with respect to the Hilbert-Schmidt inner product.

The DLA generated by the gate set in Fig.~\ref{ap:fig:line_ZXX_circuit} is:
\begin{equation}
  \mathfrak{g} = \left\langle \left( \left\{ iX_jX_{j+1} \right\}^{n-1}_{j=1} \right) \cup \left( \left\{ iZ_j \right\}^{n}_{j=1} \right) \right\rangle_{\mathrm{Lie}},
\end{equation}
where $\langle \cdot \rangle_{\mathrm{Lie}}$ denotes the Lie algebra generated by the specified operators under the commutator $[A,B]=AB-BA$.
As shown in Refs.~\cite{ragone2024lie,kokcu2022fixed}, this DLA is simple and isomorphic to $\mathfrak{so}(2n)$; consequently, $\dim \mathfrak{g} = n(2n-1)$.
An orthogonal basis of $\mathfrak{g}_{\mathbb{C}}$ can be chosen as:
\begin{equation}
  \{\widehat{X_iX_j}, \widehat{X_iY_j}, \widehat{Y_iX_j}, \widehat{Y_iY_j}\}_{1\leq i < j \leq n}, \quad \mathrm{where}\quad \widehat{A_iB_j} \coloneqq A_iZ_{i+1} \cdots Z_{j-1} B_j.
\end{equation}

For the initial state $\ket{0}^{\otimes n}$, the $\mathfrak{g}$-purity evaluates to:
\begin{equation}
  \mathcal{P}_{\mathfrak{g}}(\ketbra{0}^{\otimes n}) = \sum_{i=1}^{n(2n-1)} \tr{B_i^\dagger \ketbra{0}^{\otimes n}}^2 = \frac{n}{2^n}.
\end{equation}
Because the observable $O$ belongs to the DLA, we have $O_{\mathfrak{g}} = O$, and hence:
\begin{equation}
  \mathcal{P}_{\mathfrak{g}}(O) = \tr{O_{\mathfrak{g}}^2} = 2^{n+1}.
\end{equation}
Therefore, when the number of blocks $p$ is large enough that the circuit ensemble forms an approximate design, the variance of $\langle O \rangle$ converges to
\begin{equation}\label{ap:eq:benchmark_variance}
  \mathrm{Var}(\langle O \rangle) = \frac{\mathcal{P}_{\mathfrak{g}}(\ketbra{0}^{\otimes n})\mathcal{P}_{\mathfrak{g}}(O)}{\dim{\mathfrak{g}}} = \frac{2}{2n-1}.
\end{equation}

We employ the estimation algorithm in Sec.~\ref{sec:OBPPP} to obtain the estimated variance $\widehat{\mathrm{Var}}(\langle O \rangle)$ for various system sizes $n$.
To compare against the analytical result in Eq.~\eqref{ap:eq:benchmark_variance}, we set $p = 2048$.
Since there is no PCS1 channel present in the experiment, we can fix $N_{\text{in}} = 1$.
We then analyze the estimation error as a function of the number of samples $N_{\text{out}}$ (denoted as “$\# samples$” in the figures).

The relative error defined as:
\begin{equation}
  \frac{\mathrm{Var}(\langle O \rangle) - \widehat{\mathrm{Var}}(\langle O \rangle)}{\mathrm{Var}(\langle O \rangle)},
\end{equation}
is shown in Fig.~\ref{ap:fig:convergence}.
For all $n$ considered, the error decreases as the number of samples increases, and the asymptotic behavior is consistent with $(\# samples)^{-1/2}$, as expected.

\begin{figure}[H]
  \centering
  \includegraphics[width=0.61\linewidth]{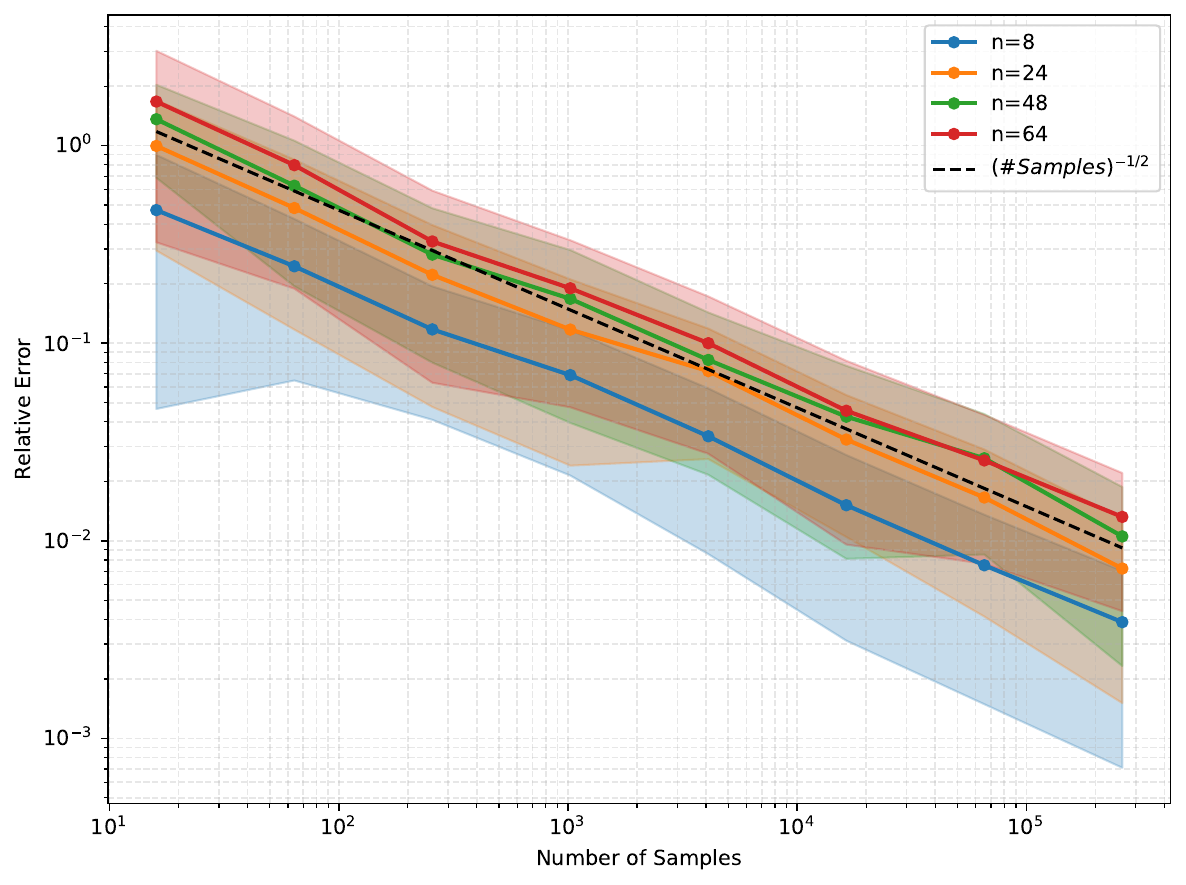}
  \caption{
  Relative error vs.\ number of samples $\# samples$ for different number of qubits $n$.
  Solid lines denote the mean over $100$ independent trials; shaded bands indicate the corresponding standard deviation.
  The black dashed line shows the reference slope $(\# samples)^{-1/2}$.
  }
  \label{ap:fig:convergence}
\end{figure}

The estimated variance $\widehat{\mathrm{Var}}(\langle O \rangle)$ is reported in Fig.~\ref{ap:fig:calibration}.
The estimates agree with the theoretical values: the points lie along the reference line $y=x$, indicating that the estimator is unbiased.
Moreover, the standard deviation decreases with the number of samples, consistent with the central limit theorem.

\begin{figure}[H]
  \centering
  \includegraphics[width=0.5\linewidth]{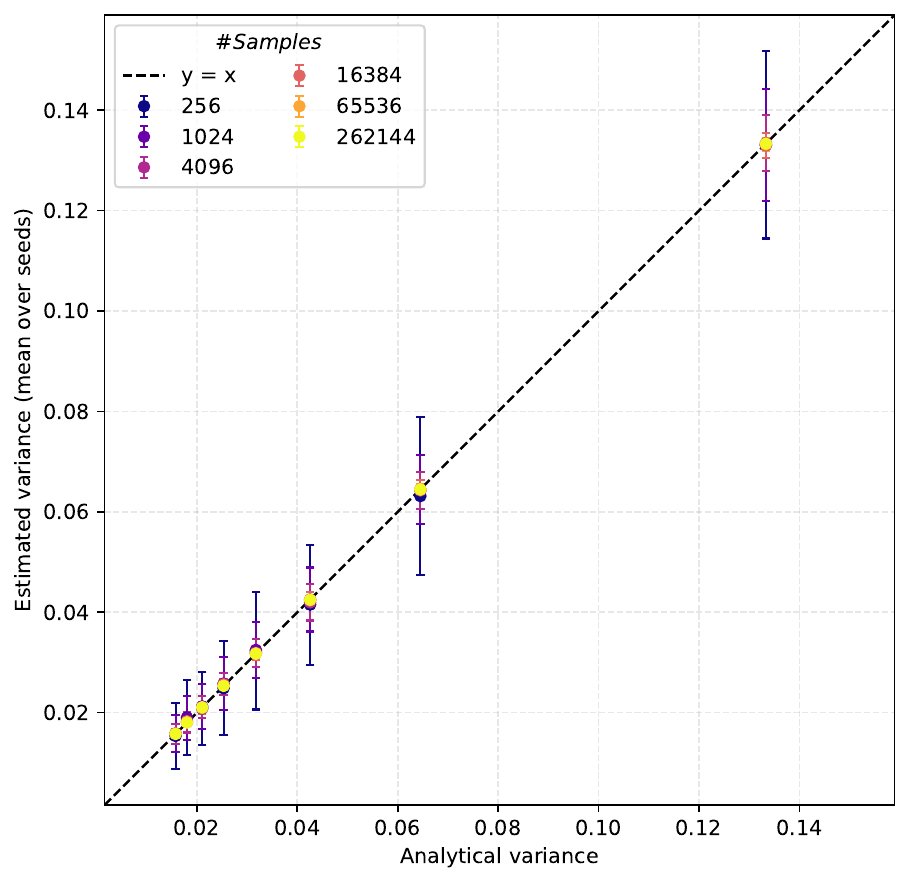}
  \caption{
  Estimated variance (mean over $100$ independent trials) vs.\ theoretical value.
  Error bars denote the standard deviation across $100$ trials for different numbers of samples (“$\# samples$”).
  The black dashed line is the reference $y=x$.
  }
  \label{ap:fig:calibration}
\end{figure}

\section{Numerical Simulation}
In this section, we will illustrate the details of the numerical simulation.

\subsection{Noise robustness and noise bottleneck}
In the Noise Robustness and Noise Bottleneck section of the main text, we use a 75-layer, 435-qubit circuit, designed on a hypothetical chip architecture and composed of $R_X$, $R_Z$ and $R_{ZZ}$ gates. 
The circuit is shown in Fig.~\ref{fig:circuit_435}, and the observable is given by:
\begin{equation}
  O = Z_{(13,14)},
\end{equation}
where $Z_{(13,14)}$ is the Pauli $Z$ operator on the qubit in the $13$-th row~(from top to bottom) and $14$-th column~(from left to right) of the hypothetical chip architecture.
During the circuit evolution, the noise is applied after each gate, and the noise model is given by:
\begin{equation}
  \mathcal{N}_{dep}(\rho)=(1-\lambda)\rho + \lambda\frac{\mathbb{I}}{2},
\end{equation}
where $\mathbb{I}$ is the identity operator, and $\lambda$ is the noise strength.
For two-qubit gates, the noise applied after gates is the tensor product of the noise applied after each qubit $\mathcal{N}_{dep}^{\otimes 2}$.

In the estimation of the noise robustness and noise sensitivity, we set the number of samples are $N_{\text{out}}=2^{22}$ and $N_{\text{in}}=1$ (since the effect of Pauli noise can use Eq.~\eqref{eq:pauli_noise_effect_factor} directly), thus the total number of circuit evaluations is $2^{22}$.
In the estimation of Reduction in MSE, the number of samples is $N_{\text{out}}=2^{18}$ ($N_{\text{in}}=1$) and we repeat the estimation $100$ times to get the standard deviation.

\begin{figure}[H]
  \centering
  \includegraphics[width=\columnwidth]{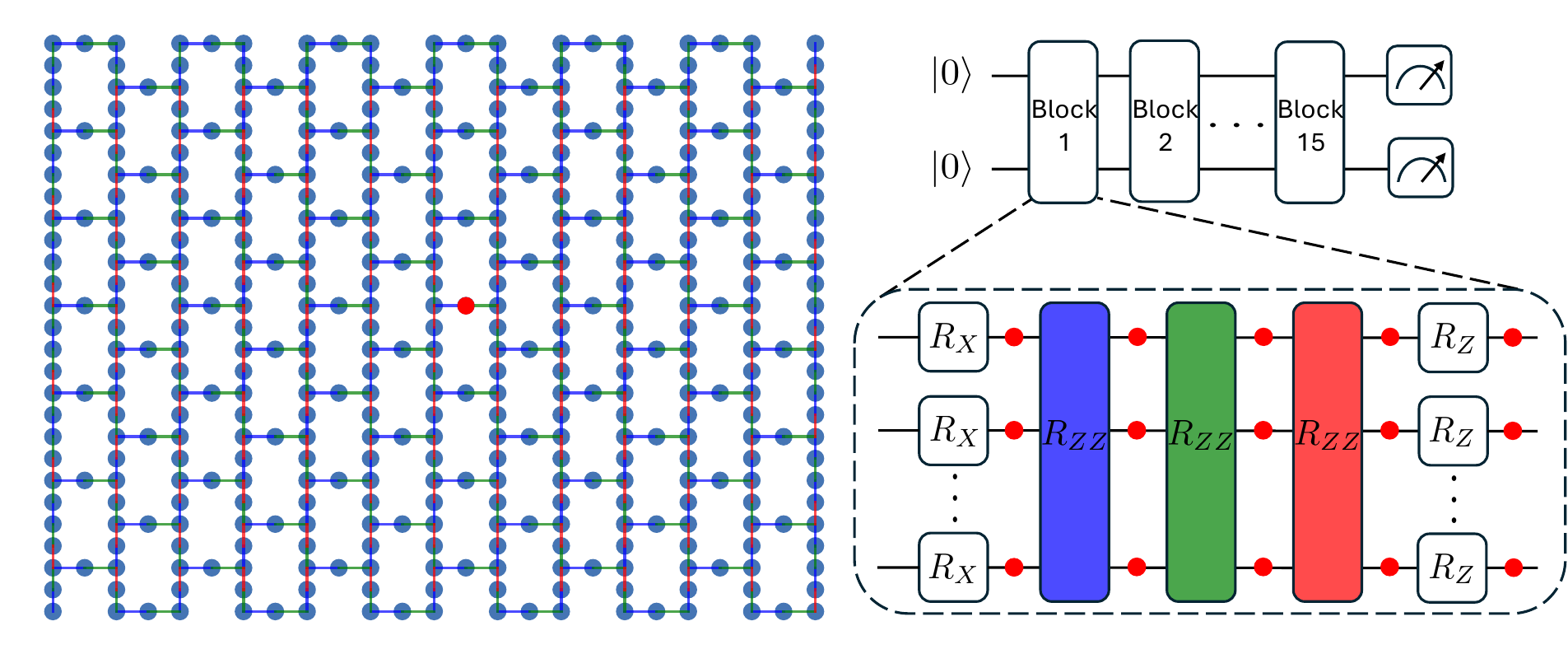}
  \caption{
    The circuit for Noise robustness and Noise Bottleneck estimation.
    The left figure is the hypothetical chip architecture, where the qubits are arranged in a 2D grid, and the edges are the physical connections between the qubits, the red qubit in the center is the final measurement qubit.
    These edges are colored by three colors, two-qubit edges in same color can be performed simultaneously, and the three colors are used for the three layers of $R_{ZZ}$ gates.
    The circuit are shown in the right figure, where the circuit contains $15$ blocks, each block contains $5$ layers, one $R_X$ layer, three $R_{ZZ}$ layers and one $R_Z$ layer.
    The three layers of $R_{ZZ}$ gates are labeled by three different colors, blue, red and green, corresponding to the corresponding colors of the edges in the graph.
    The noise are labeled by the red points, which are applied after gates, to describe the gate-based noise model.
  }
  \label{fig:circuit_435}
\end{figure}

\subsection{Trainability Estimation}
For the trainability estimation, we employed circuits similar to that in Fig.~\ref{fig:circuit_435}, with the key difference that the two-qubit $R_{ZZ}$ gates were replaced by the  $CZ$ gates in order to visualize the trainability across the qubit layout.
And the number of blocks is no longer $15$, but $1,2,3$, corresponding to $5,10,15$ layers of gates.
The circuit is shown in Fig.~\ref{fig:circuit_435_trainability}, and the observable is given by:
\begin{equation}
  O = Z_{(13,1)} Z_{(13,7)} Z_{(13,13)} Z_{(13,19)} Z_{(13,25)} ,
\end{equation}
where $Z_{(i,j)}$ is the Pauli $Z$ operator on the qubit in the $i$-th row~(from top to bottom) and $j$-th column~(from left to right) of the hypothetical chip architecture.
The noise is similar to the gate-based noise model in Fig.~\ref{fig:circuit_435}, but the noise type is changed to the amplitude damping noise:
\begin{equation}
\mathcal{N}_{ad}\left(\begin{pmatrix}
  \rho_{00} & \rho_{01} \\
  \rho_{10} & \rho_{11}
\end{pmatrix}\right)=\begin{pmatrix}
  \rho_{00}+\gamma\rho_{11} & \sqrt{1-\gamma}\rho_{01} \\
  \sqrt{1-\gamma}\rho_{10} & (1-\gamma)\rho_{11}
\end{pmatrix},
\end{equation}
where $\gamma$ is the damping parameter.
Because the noise is PCS1 channel, thus the estimation is running in Algorithm~\ref{alg:OBPPP2}.

In the estimation of the trainability, when $\gamma=0$, there is no noise and the number of samples are $N_{\text{out}}=2^{18}$ and $N_{\text{in}}=1$, when $\gamma\neq 0$, the number of samples are $N_{\text{out}}=2^{12}$ and $N_{\text{in}}=2^6$.
For both cases, we repeat the estimation $100$ times to get the standard deviation.

\begin{figure}[H]
  \centering
  \includegraphics[width=\columnwidth]{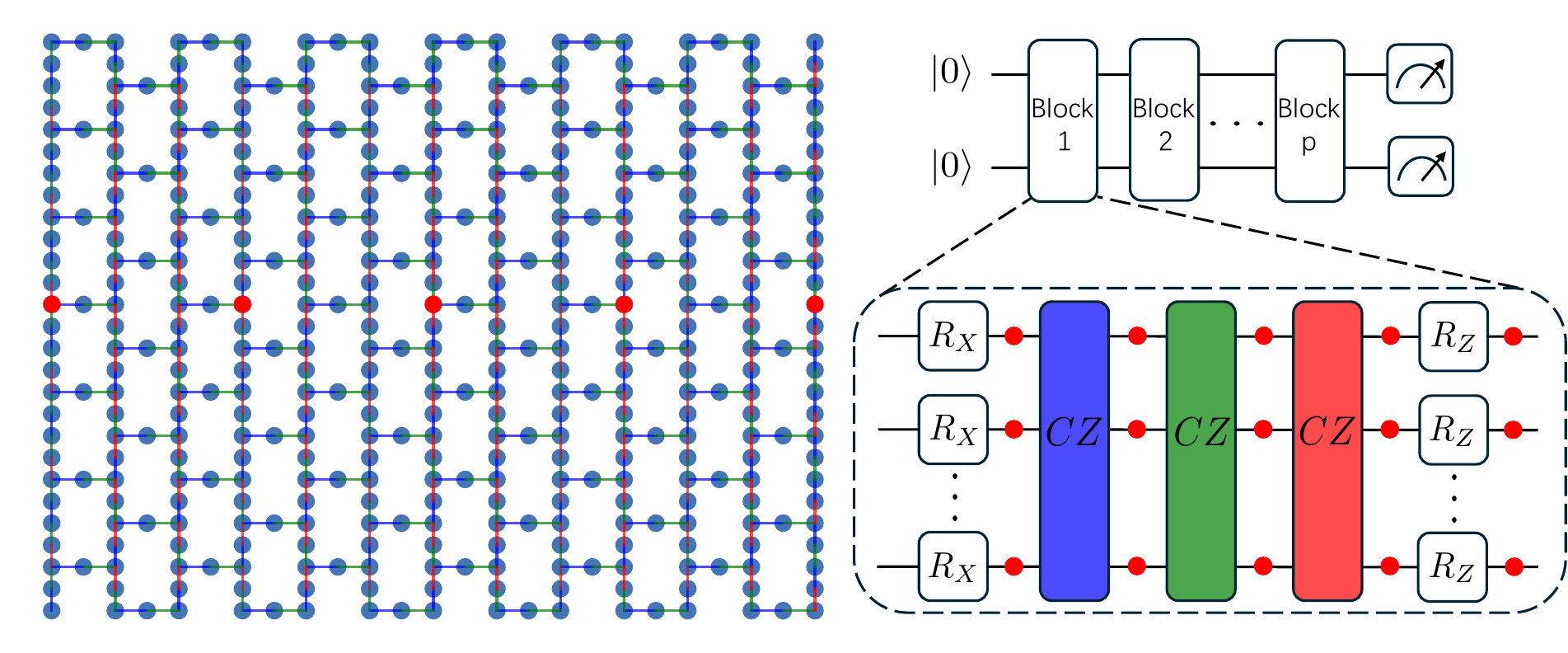}
  \caption{
    The circuit for Trainability Estimation.
    The left figure is the hypothetical chip architecture, where the qubits are arranged in a 2D grid, and the edges are the physical connections between the qubits, the red qubit in the center is the final measurement qubit.
    These edges are colored by three colors, two-qubit edges in same color can be performed simultaneously, and the three colors are used for the three layers of controlled-Z ($CZ$) gates.
    The circuit are shown in the right figure, where the circuit contains $p=1,2,3$ blocks, each block contains $5$ layers, one $R_X$ layer, three $CZ$ layers and one $R_Z$ layer.
    The three layers of $CZ$ gates are labeled by three different colors, blue, red and green, corresponding to the corresponding colors of the edges in the graph.
    The noise are labeled by the red points, which are applied after gates, to describe the gate-based noise model.
  }
  \label{fig:circuit_435_trainability}
\end{figure}

\subsection{Expressibility Estimation}
For the expressibility estimation, we employed 8-qubit circuits with $R_X$, $CZ$ and $R_Z$ gates, which can be viewd as a reduced unit of the circuit in Fig.~\ref{fig:circuit_435_trainability}.
The number of blocks is $1,2,3$, corresponding to $4,8,12$ layers of gates.
The circuit is shown in Fig.~\ref{fig:circuit_8_exp}, and there is no observable involved in the estimation.
The noise is the gate-based amplitude damping noise:
\begin{equation}
  \mathcal{N}_{ad}\left(\begin{pmatrix}
    \rho_{00} & \rho_{01} \\
    \rho_{10} & \rho_{11}
  \end{pmatrix}\right)=\begin{pmatrix}
    \rho_{00}+\gamma\rho_{11} & \sqrt{1-\gamma}\rho_{01} \\
    \sqrt{1-\gamma}\rho_{10} & (1-\gamma)\rho_{11}
  \end{pmatrix},
\end{equation}
where $\gamma$ is the damping parameter.

When $\gamma=0$, it is noiseless case, we can use Eq.~\eqref{eq:noisy_expressibility_all} to estimate the expressibility.
When $\gamma\neq 0$, because the adjoint of amplitude damping is no longer PCS1, we instead estimate the lower bound of the expressibility $\widetilde{\mathcal{M}}_{2\leq}^2$ using Eq.~\eqref{eq:noisy_expressibility_lower_bound}.
The number of samples are $N_{\text{out}}=2^{22}$ and $N_{\text{in}}=2^8$.
In all cases, we repeat the estimation $100$ times to get the standard deviation.

\begin{figure}[H]
  \centering
  \includegraphics[width=0.6\columnwidth]{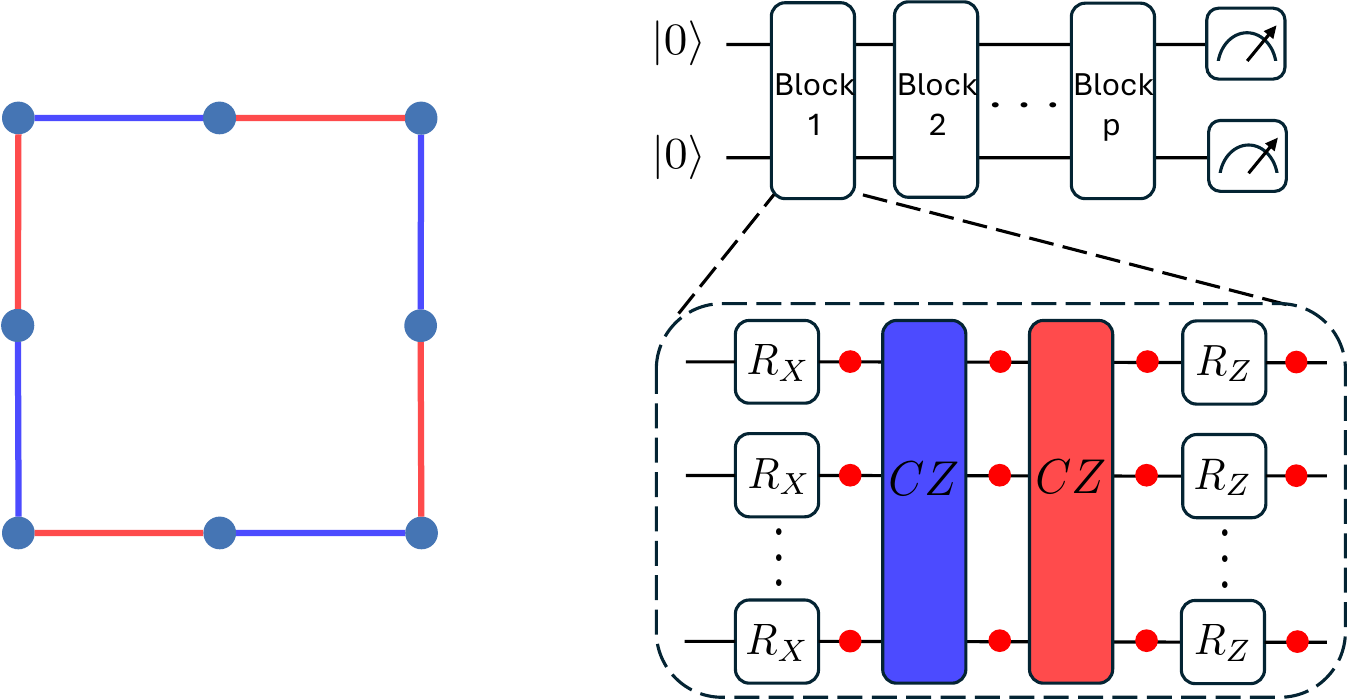}
  \caption{
    The circuit for Expressibility Estimation.
    The left figure is the hypothetical chip architecture, where the qubits are arranged in a circle, and the edges are the physical connections between the qubits. 
    These edges are colored by two colors, two-qubit edges in same color can be performed simultaneously, and the two colors are used for the two layers of controlled-Z ($CZ$) gates.
    The circuit are shown in the right figure, where the circuit contains $p=1,2,3$ blocks, each block contains $4$ layers, one $R_X$ layer, two $CZ$ layers and one $R_Z$ layer.
    The two layers of $CZ$ gates are labeled by two different colors, blue and red, corresponding to the corresponding colors of the edges in the graph.
    The noise are labeled by the red points, which are applied after gates, to describe the gate-based noise model.
    It can be viewed as a reduced unit of the circuit in Fig.~\ref{fig:circuit_435_trainability}.
  }
  \label{fig:circuit_8_exp}
\end{figure}

%\bibliographystyle{apsrev}

% ****** End of file apssamp.tex ******

\end{document}